\pdfoutput=1
\documentclass[11pt]{article}

\usepackage[T1]{fontenc}
\usepackage[utf8]{inputenc}
\usepackage{lmodern}
\usepackage[letterpaper,margin=1in]{geometry}
\usepackage{amsmath,amssymb,amsfonts,amsthm}
\usepackage{graphicx}
\usepackage{booktabs,array}
\usepackage{enumitem}
\usepackage{placeins}
\usepackage{xcolor}
\usepackage{url}
\usepackage{mleftright}
\usepackage{tocloft}
\usepackage[round,authoryear]{natbib}
\usepackage[colorlinks,citecolor=blue!60!black,linkcolor=blue!60!black,urlcolor=blue!60!black]{hyperref}

\mleftright
\allowdisplaybreaks
\expandafter\def\expandafter\UrlBreaks\expandafter{\UrlBreaks\do\-\do\_}

\theoremstyle{plain}
\newtheorem{theorem}{Theorem}
\newtheorem{lemma}[theorem]{Lemma}
\newtheorem{proposition}[theorem]{Proposition}

\theoremstyle{definition}
\newtheorem{assumption}[theorem]{Assumption}
\newtheorem{algo}[theorem]{Algorithm}
\theoremstyle{remark}
\newtheorem{remark}[theorem]{Remark}
\newcommand{\BlackBox}{\rule{1.5ex}{1.5ex}}
\renewenvironment{proof}[1][Proof]{\par\noindent{\bf #1.\ }}{\hfill\BlackBox\\[2mm]}

\newcommand{\E}{\mathbb{E}}
\newcommand{\PP}{\mathbb{P}}
\newcommand{\1}{\mathbf{1}}
\newcommand{\FDR}{\mathrm{FDR}}
\newcommand{\FDP}{\mathrm{FDP}}
\newcommand{\Bin}{\mathrm{Bin}}

\newcommand{\Hzero}{\mathcal{H}_0}

\newcommand{\Aset}{\mathcal{A}}
\newcommand{\D}{\mathcal{D}}
\newcommand{\F}{\mathcal{F}}

\newcommand{\N}{\mathbb{N}}

\title{Finite-Sample FDR Control for Greedy Aggregation over Networks}
\author{Mehrdad Pournaderi\\[2pt]
  \normalsize Mofid Securities\\
  \normalsize\texttt{m.pournaderi@emofid.com}}
\date{}

\begin{document}
\maketitle

\begin{abstract}
In distributed multiple testing, $N$ sites hold p-values for their own
hypotheses and a center chooses a network-wide rejection set with false
discovery rate (FDR) at most $\alpha$ and power comparable to centralized
testing, without transmitting all $m$ p-values. Greedy interval aggregation
(Pournaderi and Xiang, 2023) sends the counts of $O(\sqrt{m})$ p-value
intervals, $O(\sqrt{m}\log{m})$ bits, and is asymptotically as powerful as
the best interval rule, but controls FDR only asymptotically: its FDR can
exceed $\alpha$ at finite $m$ because the same counts both rank the
intervals and estimate their false discovery proportion. We quantify this
winner's-curse selection bias of the greedy and propose \emph{budgeted
BONuS-GA}: each site adds synthetic uniform p-values, the center ranks
candidate intervals across sites by pooled count, and each site estimates
its false discoveries from its synthetic counts under a per-site share of
$\alpha$. It controls $\mathrm{FDR}\le\alpha$ at every $m$ for any fixed
assignment of hypotheses to sites assuming only independent uniform null
p-values, and keeps the $O(\sqrt{m}\log{m})$ bits when $N=O(\sqrt{m})$. We
further propose a sample-splitting alternative that has finite-sample
control under a random-site model. Our proposed methods show competitive
power in a range of simulation settings.
\end{abstract}


\setcounter{tocdepth}{1}
{\small\tableofcontents}

\section{Introduction}
\label{sec:intro}

\subsection{The problem: one error rate for discoveries made at many sites}

Sensor networks, federated learning systems, scientific collaborations, and
large-scale observatories routinely produce hypothesis tests partitioned
across sites that cannot freely share their raw data. Centralizing every
p-value at one node may be infeasible because of communication budgets,
privacy constraints, or sheer volume. The \emph{distributed multiple
testing} problem asks: how should the sites collaborate so that the network
as a whole controls a global error criterion---most commonly the false
discovery rate (FDR) of \citet{BH1995}---at a prescribed level $\alpha$,
while staying within a small communication budget?

Controlling the error rate separately at each site does not solve this
problem. Suppose ten sites each test one true null hypothesis and each
rejects when its p-value is at most $0.05$. Every site operates at level
$5\%$, but the combined set of discoveries contains at least one false
discovery, and hence has FDP equal to one, with probability
$1-0.95^{10}\approx0.40$: the FDR of the combined discoveries is about
$40\%$. The effect is the same, only less visible, when each site runs the
Benjamini--Hochberg (BH) procedure at level $\alpha$ on its own p-values:
sites with no signal still make occasional false rejections, and these
accumulate in the network-wide discovery set (Section~\ref{sec:sim},
Experiment~1, shows the global FDR of such ``local BH'' rising from $0.14$
to $0.33$ as quiet sites are added). The network therefore needs a rule
that controls the FDR of the \emph{union} of the discoveries.

\subsection{What is observed, what is unknown, and what is controlled}
\label{sec:observed}

Throughout, the statistical input is a collection of p-values; we assume
that the null p-values are correctly calibrated (uniform on $[0,1]$) and
do not model the underlying tests. Table~\ref{tab:observed} summarizes the
objects involved.

\begin{table}[b]
\centering
\small
\caption{The objects of the problem. Site $i$ holds $m^{(i)}$ hypotheses;
$m=\sum_im^{(i)}$.}
\label{tab:observed}
\begin{tabular}{@{}p{0.27\textwidth}p{0.69\textwidth}@{}}
\toprule
Category & Contents \\
\midrule
Observed at site $i$ & Its p-values $P^{(i)}_1,\dots,P^{(i)}_{m^{(i)}}$,
their hypothesis identifiers, and the site size $m^{(i)}$. \\[2pt]
Unknown & Which hypotheses are null; the null proportions
$r_0^{(i)}$ and the distributions of the non-null p-values. \\[2pt]
Generated by the procedure & Synthetic uniform p-values (BONuS-GA) or
random split assignments (CFGA); both are drawn independently of the
data. \\[2pt]
Communicated & Interval counts sent from the sites to the center, and the
accepted-interval indices sent back; $O(\log m)$ bits per message. \\[2pt]
Output & A set $\widehat{\mathcal R}$ of \emph{original} hypotheses
declared discoveries, assembled from the sites' rejections. \\
\bottomrule
\end{tabular}
\end{table}

The principal object of study is the network-wide rejection set
\[
\widehat{\mathcal R}\;\subseteq\;\mathcal H:=\{(i,k):1\le i\le N,\ 1\le k\le m^{(i)}\},
\]
and the objective is
\begin{equation}\label{eq:fdr-goal}
\FDR\;=\;\E\!\left[\frac{|\widehat{\mathcal R}\cap\mathcal H_0|}
{|\widehat{\mathcal R}|\vee1}\right]\;\le\;\alpha ,
\end{equation}
where $\mathcal H_0\subseteq\mathcal H$ is the unknown set of true nulls.
The expectation is over the data and over any randomization used by the
procedure; it cannot be checked on a single realization, but it is what
a guarantee is a statement about. Two points are easily confused.
First, every observed p-value belongs to either a null or a non-null
hypothesis, and this status is unknown; ``null p-values'' are the
unobserved subset of the data that the guarantee concerns. Second, the
synthetic p-values generated by BONuS-GA are known to be synthetic by the
algorithm; \emph{masking} does not hide them, it restricts which
information is allowed to influence the selection of rejection regions
(Section~\ref{sec:bonus-main}).

All procedures in this paper reject unions of p-value \emph{intervals}: each
site is assigned a region $\widehat\Gamma^{(i)}\subseteq[0,1]$, a union of
short intervals chosen by the center, and
$\widehat{\mathcal R}=\{(i,k):P_k^{(i)}\in\widehat\Gamma^{(i)}\}$. The
interval counts, null-rate estimates, and synthetic counts that appear
below are ingredients used to choose these regions and to decide when to
stop; the guarantee \eqref{eq:fdr-goal} concerns only
$\widehat{\mathcal R}$.

\subsection{Greedy aggregation and its finite-sample failure}
\label{sec:greedy-background}

The starting point is the greedy aggregation algorithm of
\citet{PournaderiXiang2024}. Each site lays a grid of short intervals over
its p-value range, with widths chosen so that a plug-in model predicts the
same number of null p-values in every interval; the sites upload their
interval counts; and the center admits intervals across the network in
decreasing order of count, stopping when a plug-in estimate of the false
discovery proportion exceeds $\alpha$ (Section~\ref{sec:greedy} gives the
complete description). With $O(\sqrt m)$ intervals in total and $O(\log m)$
bits per interval, the whole network communicates $O(\sqrt m\log m)$ bits.
Under a random-site mixture model and as $m\to\infty$, this procedure
controls FDR asymptotically and is asymptotically as powerful as the best
rule that rejects unions of intervals \citep[Prop.~4 and
Thm.~7]{PournaderiXiang2024}.

An asymptotic bound describes the algorithm as $m\to\infty$, not at the
$m$ of a particular run, and in our experiments the FDR of greedy
aggregation exceeds $\alpha$ in several configurations, with the excess
persisting at the larger tested sample sizes (Section~\ref{sec:sim}). A
mechanism consistent with this can be stated in one sentence: the same
interval counts are used both to select the intervals with the largest
counts and to estimate how many null p-values those intervals contain, so
the selected intervals look denser than they are
(Figure~\ref{fig:running}, left). Appendix~\ref{sec:gap} quantifies this
selection bias of the interval densities; as Section~\ref{sec:why}
explains, it does not by itself prove that the mechanism accounts for the
observed FDR excess.

\subsection{Contributions}
\label{sec:contributions}

The organizing idea is to \emph{separate the counts that select intervals
from the counts that calibrate the stopping decision}
(Figure~\ref{fig:running}). We do this in two ways.

\begin{enumerate}[leftmargin=*,label=(\arabic*)]
\item \textbf{Masking: budgeted BONuS-GA (Section~\ref{sec:bonus-main}),
the principal contribution.} Each site adds synthetic uniform p-values to
its own data. The center ranks candidate intervals \emph{across} sites
using the pooled counts of real and synthetic p-values, which do not reveal
which p-values are real; the number of false discoveries is then estimated
\emph{within} each site from that site's synthetic counts, and each site is
given a share $\alpha_i$ of the level, $\sum_i\alpha_i\le\alpha$. Because
the calibration is done site by site, the procedure never needs to know how
the true nulls are distributed across sites. Theorem~\ref{thm:bonus-pernode}
shows that it controls $\FDR\le\alpha$ at every $m$, with no oracle input,
for any allocation of synthetic p-values that does not depend on the data,
and for any \emph{fixed} assignment of hypotheses and null status to
sites, assuming only that the null p-values are independent uniforms,
independent of the non-null p-values. It uses all $m$ p-values for both
selection and inference, needs two communication rounds,
and keeps the $O(\sqrt m\log m)$ budget when $N=O(\sqrt m)$.

\item \textbf{Sample splitting: cross-fit greedy aggregation (CFGA;
Section~\ref{sec:alt}).} Each site randomly splits its hypotheses into two
halves; one half ranks the intervals and the other supplies the counts
that decide where to stop. With known per-site null masses $\pi^{(i)}$ this controls
$\FDR\le\alpha$ at every $m$ (Theorem~\ref{thm:cfga}); an inflated null-rate
estimate removes the oracle at a slack of $1/m$ (Theorem~\ref{thm:aug});
and combining the two halves' scores through e-BH over several random
splits (e-CFGA, Theorem~\ref{thm:ecfga}) avoids dividing $\alpha$ between
the halves. These proofs use a random-site model in which hypotheses are
assigned to sites at random (Assumption~\ref{ass:model}); the
budgeted procedure does not.
\end{enumerate}

Two further results support these. A data-driven bandwidth rule
(Section~\ref{sec:eps-adaptive}) chooses the interval width by
cross-validation inside the ranking half, so the CFGA-family guarantees
are unaffected (Proposition~\ref{prop:adaptive-FDR}).
Appendix~\ref{sec:gap} quantifies the winner's-curse bias of the interval
densities selected by the original greedy, for a variant whose candidate
intervals lie in the Storey head $(0,\lambda]$: it is of order
$m^{-1/4}\sqrt{\log m}$ at the standard bandwidth, an upper bound in
general that is attained under the global null. Table~\ref{tab:methods}
lists the procedures with their requirements and guarantees.
Section~\ref{sec:sim} returns to the monitoring-network example of
Section~\ref{sec:example} and compares FDR and power; communication is
accounted for in Sections~\ref{sec:bonus-comm} and~\ref{sec:alt} and
Appendix~\ref{app:protocols}.
Section~\ref{sec:discuss} discusses limitations and open questions. The
proofs, the detailed communication protocols, and secondary experiments
are in the appendices.

\paragraph{What is borrowed and what is new.} The masking construction
follows the counting knockoffs of \citet{BarberCandes2015} and the
bag-of-null-statistics (BONuS) procedure of \citet{Yang2021BONuS}; the
cross-fitting idea follows the sample-splitting selective inference of
\citet{JinCandes2022,Bates2023} and the cross-fit e-values of
\citet{RenBarber2024evalues,RenWeiCandes2023}. What is new here is the
distributed construction (ranking across sites on interval counts, under
a communication budget), the per-site calibration with budget shares that
are themselves computed from the pooled data, whose proof must hold
conditionally on those data, and the communication accounting.

\begin{figure}[tbp]
\centering
\includegraphics[width=\textwidth]{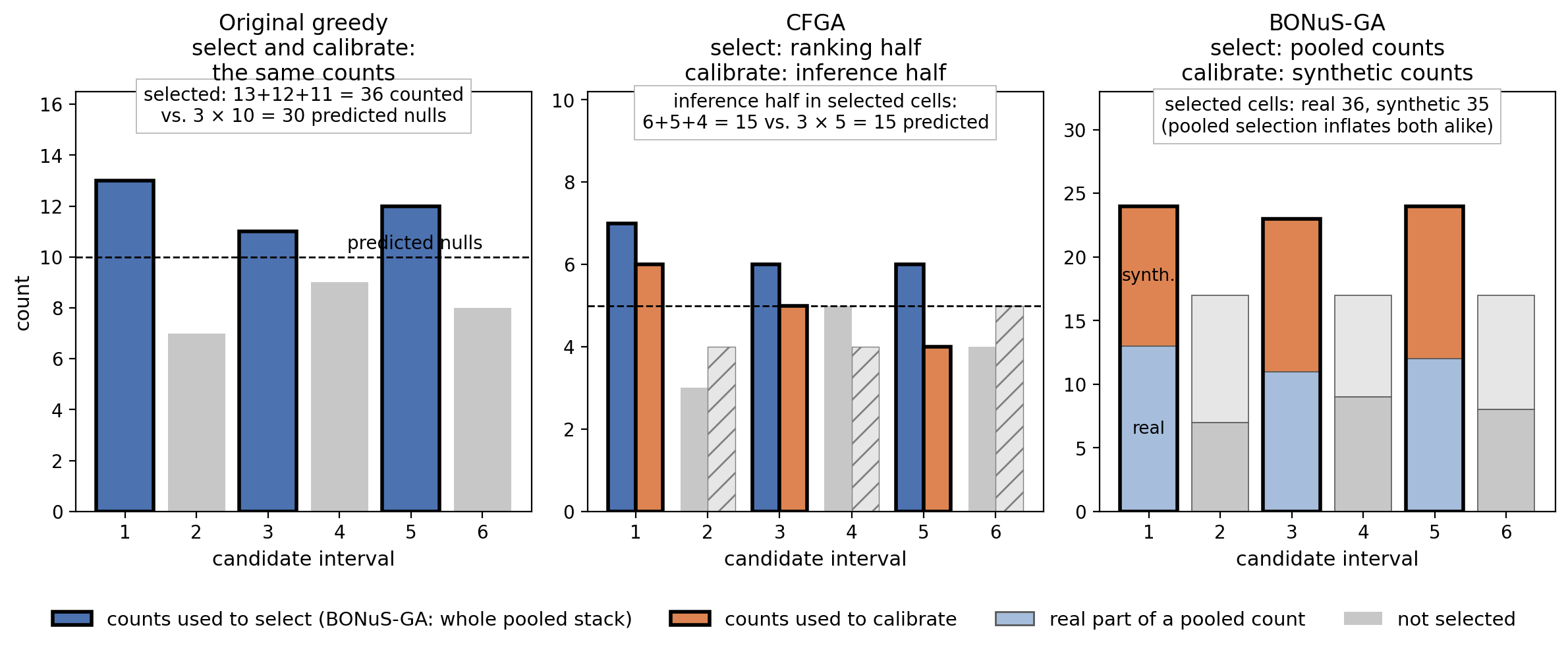}
\caption{Which counts choose the intervals and which counts set the
stopping rule, on six candidate intervals holding only nulls (ten per
interval in expectation), with $M=3$ selected. Left: the original greedy
selects and calibrates with the same counts, so the selected intervals look
denser than they are ($36$ counted against $30$ predicted). Centre: CFGA
selects with the ranking half and calibrates with the held-out half, whose
counts are unaffected by the selection ($15$ against $15$). Right: BONuS-GA
selects with the pooled real-plus-synthetic counts and calibrates with the
synthetic counts; selection inflates real nulls and synthetic p-values alike
($36$ against $35$).}
\label{fig:running}
\end{figure}

\begin{table}[t]
\centering
\small
\setlength{\tabcolsep}{4pt}
\caption{The procedures at a glance; the original greedy is that of
\citet{PournaderiXiang2024}. ``Oracle input'' lists what must be known in
advance. ``Labels'' states how the proof assumes hypotheses are assigned
to nodes: at random, or fixed (possibly adversarially). In the guarantees
$\eta=1/m$. Communication is for the standard bandwidth and $N=O(\sqrt m)$
nodes; for e-CFGA, $S$ is the number of random splits and $\bar R\le m$
the number of hypotheses that receive a nonzero score. ``Rounds'' counts
communication rounds when each node sends all its interval counts at
once; the original greedy is listed in its interactive form, which admits
one interval per round. An adaptive bandwidth adds one round
(Section~\ref{sec:eps-adaptive}). The protocols are in
Appendix~\ref{app:protocols}.}
\label{tab:methods}
\begin{tabular}{@{}p{0.27\textwidth}p{0.12\textwidth}p{0.14\textwidth}p{0.11\textwidth}p{0.18\textwidth}l@{}}
\toprule
Procedure & Oracle input & Guarantee & Labels & Bits, $O(\cdot)$ & Rounds \\
\midrule
Original greedy & --- & asymptotic & random & $\sqrt m\log m$ & $O(\sqrt m)$ \\
CFGA (Thm.~\ref{thm:cfga}) & null masses $\pi^{(i)}$ & $\FDR\le\alpha$ & random & $\sqrt m\log m$ & 2 \\
CFGA$+$Storey & --- & heuristic & --- & $\sqrt m\log m$ & 2 \\
Inflated CFGA (Thm.~\ref{thm:aug}) & --- & $\FDR\le\alpha+\eta$ & random & $\sqrt m\log m$ & 2 \\
Inflated e-CFGA (Thm.~\ref{thm:ecfga}) & --- & $\FDR\le\alpha+\eta$ & random & $(S\sqrt m{+}\bar R)\log m$ & 3 \\
Oracle BONuS-GA (Thm.~\ref{thm:bonus}) & null masses $\pi^{(i)}$ & $\FDR\le\alpha$ & random & $\sqrt m\log m$ & 2 \\
BONuS-GA, global plug-in & --- & heuristic & --- & $\sqrt m\log m$ & 2 \\
Budgeted BONuS-GA (Thm.~\ref{thm:bonus-pernode}) & --- & $\FDR\le\alpha$ & fixed or random & $\sqrt m\log m$ & 2 \\
\bottomrule
\end{tabular}
\end{table}

\subsection{A running example: a monitoring network}
\label{sec:example}

We use one setting throughout. A network of $N$ sites monitors a system;
each site tests $m^{(i)}$ hypotheses of the form ``signal $k$ is absent at
site $i$'', and reports a one-sided p-value for each. A few sites are
\emph{hubs} with many p-values, a fraction of which are non-null; most sites
are \emph{quiet}, with few p-values, all null. In the simulations of
Section~\ref{sec:sim} there are two hubs with $1000$ p-values each (30\%
non-null) and $300$ quiet sites with $20$ p-values each, so $m=8000$ and
$N=302$. One hypothesis is one (site, signal) pair; one site observes only
its own p-values; the center receives interval counts, never p-values.
Different sites benefit from different rejection regions: a hub with
strong signal should reject a wider interval of small p-values than a
quiet site, where any rejection is false. The guarantee of
Theorem~\ref{thm:bonus-pernode} in this setting is that the discoveries
collected from all $302$ sites, taken together, have FDR at most $\alpha$,
for every $m$ and whatever the true pattern of signals. A worked
realization of the main procedure on a small instance of this example is
given in Section~\ref{sec:worked}.

\subsection{Related work}

The closest distributed-FDR precursor is the sample-and-forward
algorithm of \citet{PournaderiXiang2022Sample}, by the same authors as the
greedy method we build on. Rather than transmitting (quantized) p-values,
each node samples its local empirical CDF of p-values at a few points in
$[0,\alpha]$ and forwards those samples, together with its count $m^{(i)}$,
to the center using $O(\log m)$ bits per node; the center aggregates them
into an approximate pooled CDF $\widehat F$ and computes a BH-type global
threshold $\widehat\tau=\sup\{t:\widehat F(t)=t/\alpha\}$. Sample-and-forward
is thus designed to \emph{approximate the pooled (centralized) BH rejection
region} over the network: it enjoys \emph{finite-sample} FDR control under
the same PRDS assumption as BH, while its power loss relative to pooled-BH is
characterized only \emph{asymptotically}, under a mixture model. It can be
combined with QuTE \citep{Ramdas2017QuTE} to handle general topologies.

The greedy aggregation method that we make finite-sample
(Section~\ref{sec:greedy}) pursues a different goal: it does not aim to
reproduce BH, but to attain \emph{asymptotic power-optimality} in the network
setting, among procedures that reject finite unions of intervals
\citep{PournaderiXiang2024}. For a fixed bandwidth $\varepsilon$ it too
communicates only $O(\log m)$ bits per round; the $O(\sqrt m\log m)$ total
reported in Section~\ref{sec:greedy} is not a cost of FDR control but
the price of power, arising from the $O(1/\varepsilon)=O(\sqrt m)$ candidate
intervals induced by the power-driven choice $\varepsilon\propto m^{-1/2}$.
A final clarification, since both methods involve using ``part of the
data'': sample-and-forward evaluates the empirical CDF at a few points
as a \emph{communication-compression} device, whereas CFGA randomly
splits the hypothesis \emph{indices} at each node into two folds as a
\emph{statistical} device---cross-fitting that decouples greedy
selection from inferential thresholding, which is what upgrades the
greedy method's asymptotic FDR control to a finite-sample guarantee.
The two mechanisms are unrelated.

The broader literature on distributed multiple testing with FDR control
includes the QuTE algorithm \citep{Ramdas2017QuTE}, which forwards p-values
across a general network, together with communication-efficient variants that
avoid transmitting the raw p-values: quantized distributed FDR control
\citep{XiangQuantized2019}, the sample-and-forward method
\citep{PournaderiXiang2022Sample}, and distribution-free inference over networks
\citep{PournaderiXiangDistFree2023}. Related directions include FDR control under
Byzantine corruption of the reported p-values
\citep{ZhangPournaderiXiangVarshney2025} and frameworks for spatial multiple
testing \citep{Golz2022}. The greedy aggregation method of
\citet{PournaderiXiang2024} that we build on belongs to this line. The
theoretical machinery we use---adaptive BH+ with finite-sample FDR control
\citep{StoreyTaylorSiegmund2004,HeesenJanssen2016}, sample-split
selection-inference \citep{JinCandes2022,Bates2023}, cross-fitting
\citep{Chernozhukov2018DML,RenBarber2024evalues,RenWeiCandes2023}, masked
calibration with a bag of synthetic null statistics in the
knockoff tradition \citep{BarberCandes2015,Yang2021BONuS}, and e-value
aggregation \citep{WangRamdas2022}---is
well-established in non-distributed settings; the contribution here is to
adapt it to the interval-based, communication-constrained greedy
aggregation procedure.

\section{Observations, assumptions, and error criterion}
\label{sec:setting}

The network has one center node and $N$ sites (we also say nodes), with $N$
fixed. Site $i$ holds $m^{(i)}\ge0$ p-values
$P^{(i)}=(P_1^{(i)},\ldots,P_{m^{(i)}}^{(i)})$ with unknown true-null set
$\Hzero^{(i)}\subseteq\{1,\ldots,m^{(i)}\}$, $m_0^{(i)}=|\Hzero^{(i)}|$, and
$m=\sum_im^{(i)}$ is the total number of p-values. We use two assumptions
on how these data arise. The first is all that the main result needs.

\begin{assumption}[Fixed design]\label{ass:fixed}
The site sizes $m^{(i)}$, the assignment of hypotheses to sites, and the
null sets $\Hzero^{(i)}$ are arbitrary and fixed (they may even be chosen
adversarially). The null p-values
$\{P_k^{(i)}:k\in\Hzero^{(i)},\,1\le i\le N\}$ are mutually independent
$\mathrm{Uniform}[0,1]$ random variables, independent of the non-null
p-values, which may have any joint distribution.
\end{assumption}

The second assumption is the random-site mixture model of
\citet{PournaderiXiang2024}, under which the cross-fit procedures of
Section~\ref{sec:alt} are analysed. It describes the whole network with a
single data-generating process and implies Assumption~\ref{ass:fixed}
conditionally on the site labels and null sets.

\begin{assumption}[Random-site model]\label{ass:model}
The $m$ hypotheses in the network are drawn i.i.d.: each hypothesis first picks
a node label $I_k\in\{1,\dots,N\}$ with $\PP(I_k=i)=q^{(i)}>0$,
$\sum_{i=1}^N q^{(i)}=1$, and then, given $I_k=i$, receives a p-value $P_k$ from
the local mixture
\[
G_i(t)\;=\;r_0^{(i)}\,U(t)+r_1^{(i)}\,F_i(t),\qquad t\in[0,1],
\]
where $U(t)=t$ is the uniform null CDF, $r_0^{(i)}\in[0,1]$ is the null rate
at node $i$, $r_1^{(i)}=1-r_0^{(i)}$, and $F_i$ is an unknown non-null CDF
(the non-null term and $F_i$ are vacuous when $r_0^{(i)}=1$). No
regularity of $F_i$ is assumed here; the bias analysis of
Section~\ref{sec:why} and Appendix~\ref{sec:gap} adds, and states, a
bounded-density condition. The CFGA theorems and the oracle BONuS-GA theorem of
Appendix~\ref{app:bonus-variants} are stated under this random-label
model, in which the node sizes $m^{(i)}$ are random; the main result,
Theorem~\ref{thm:bonus-pernode}, needs only Assumption~\ref{ass:fixed}. Assumption~\ref{ass:model} permits empty
nodes and nodes with $r_0^{(i)}=0$; we adopt the following conventions.
A node (or node-fold) with no p-values contributes no candidate intervals
and zero counts, and its Storey estimate is set to $1$. The Storey
estimator carries a $+1$ offset and is therefore strictly positive for
nonempty nodes, so every grid width is finite even when $r_0^{(i)}=0$.
In the oracle procedures a node with $\pi^{(i)}=0$ contributes zero
predicted null mass, which is exact since it hosts no nulls.
Equivalently, every p-value is an i.i.d.\ draw from the
global mixture $G=\sum_i q^{(i)}G_i$. A null p-value is therefore
\emph{exactly} $\mathrm{Uniform}[0,1]$, and because the hypotheses are
i.i.d.\ the p-values are mutually independent, in particular within each
node. The data split introduced
in Section~\ref{sec:method} is drawn independently of the p-values.
\end{assumption}

The model induces two complementary views of the data, both of which we use.
Marginally, each hypothesis is the \emph{lifted} pair
$\widetilde X_k=(I_k,P_k)\in\{1,\dots,N\}\times[0,1]$, and we write
$\pi^{(i)}:=q^{(i)}r_0^{(i)}$ for the per-node \emph{null mass}: the probability
that a hypothesis is both located at node $i$ and null. A family of per-node
regions $\Gamma^{(i)}\subseteq[0,1]$ lifts to the single region
$\widetilde\Gamma=\bigcup_{i}\{i\}\times\Gamma^{(i)}$, and the probability that a
hypothesis is \emph{both null and in $\widetilde\Gamma$} is
$\sum_i\pi^{(i)}\nu(\Gamma^{(i)})$, where $\nu$ denotes Lebesgue measure on
$[0,1]$. Treating all hypotheses as one group of lifted pairs, rather than
as $N$ separate nodes, is what allows a single network-wide FDP estimate
to be analysed in Appendix~\ref{app:lifted}. Conditionally on the labels,
node $i$ owns the batch $P^{(i)}$ of $m^{(i)}$ p-values described above.

\paragraph{False discovery rate.} A testing procedure outputs a rejection set
$\mathcal R\subseteq\{1,\ldots,m\}$ of hypotheses declared non-null. Writing
$R=|\mathcal R|$ for the number of rejections and $V$ for the number of them
that are in fact true nulls (the false discoveries), the false discovery
proportion and false discovery rate are
\[
\FDP=\frac{V}{R\vee 1},\qquad \FDR=\E[\FDP],
\]
with the convention $\FDP=0$ when $R=0$; our goal is a procedure with
$\FDR\le\alpha$ at a prescribed level $\alpha\in(0,1)$. For a Borel set
$\Gamma\subseteq[0,1]$ we write
\[
N^{(i)}(\Gamma)=\sum_{k=1}^{m^{(i)}}\1\{P_k^{(i)}\in\Gamma\},\qquad
V^{(i)}(\Gamma)=\sum_{k\in\Hzero^{(i)}}\1\{P_k^{(i)}\in\Gamma\}
\]
for the number of p-values at node $i$ falling in $\Gamma$ and the number of
those that are null, so that a procedure rejecting the lifted region
$\widetilde\Gamma$ has $R=\sum_iN^{(i)}(\Gamma^{(i)})$ and
$V=\sum_iV^{(i)}(\Gamma^{(i)})$.

\paragraph{Storey estimate, head and tail.} Fix a Storey threshold
$\lambda\in(0,1)$; we call $(0,\lambda]$ the \emph{head} and
$A:=(\lambda,1]$ the \emph{tail} of $[0,1]$. All grids in the paper use the
clipped Storey estimator \citep{StoreyTaylorSiegmund2004} of a site's null
proportion,
\begin{equation}\label{eq:storey}
\widehat r_0^{(i)} = \min\!\left(1,\;\frac{1 + N^{(i)}(A)}{(1-\lambda)m^{(i)}}\right),
\end{equation}
computed from the tail count $N^{(i)}(A)=\#\{k:P_k^{(i)}>\lambda\}$. A null
p-value falls in the tail with probability $1-\lambda$, and non-null
p-values can fall there too, so before clipping the estimate tends to
overstate $r_0^{(i)}$, which is the conservative direction; the $+1$ keeps
it strictly positive and the clip keeps it at most $1$. A site with no
p-values contributes no intervals and zero counts, and its estimate is set
to $1$.

\subsection{Greedy aggregation}\label{sec:greedy}

We now state the algorithm of \citet[Sec.~V-B]{PournaderiXiang2024}
in the notation above; this is the only description of it in the paper, and
later sections refer to it. Fix a bandwidth $\varepsilon>0$, a target
$\alpha\in(0,1)$ and a Storey threshold $\lambda\in(0,1)$.

\begin{enumerate}[leftmargin=*]
\item Each node $i$ estimates its null rate $\widehat r_0^{(i)}$ (Storey or
spacing) and its own node probability $\widehat q^{(i)}=m^{(i)}/m$---its
share of the network's $m$ p-values (a random quantity under the
random-label model of Assumption~\ref{ass:model})---then sets
$\widehat L^{(i)}=\varepsilon/(\widehat q^{(i)}\widehat r_0^{(i)})$,
$\widehat K^{(i)}=\lfloor 1/\widehat L^{(i)}\rfloor$.
\item Each node $i$ computes the empirical interval density
$
\widehat h_j^{(i)}\;=\;\frac{1}{\varepsilon m}\,N^{(i)}(\widehat I_j^{(i)}),
$
where $\widehat I_j^{(i)}=\widehat L^{(i)}(j-1,j]$ for
$j=1,\ldots,\widehat K^{(i)}$.
\item The center node greedily picks intervals in decreasing $\widehat h$
order, maintaining
$
\widehat\FDP\;=\;M\big/\sum_{\ell=1}^M\widehat h_{(\ell)},
$
and stops as soon as $\widehat\FDP>\alpha$.
\end{enumerate}

\noindent Because the $\widehat h_{(\ell)}$ are admitted in decreasing order,
$\widehat\FDP=M/\sum_{\ell\le M}\widehat h_{(\ell)}$ is nondecreasing in $M$, so
this greedy ``first-crossing'' stop coincides with the step-up cut
$\max\{M:\widehat\FDP\le\alpha\}$. The cross-fit estimator of
Section~\ref{sec:alt} loses this monotonicity, which is why CFGA states its
acceptance rule directly in step-up form.

\section{Budgeted BONuS-GA}
\label{sec:bonus-main}

\subsection{Idea}

The original greedy uses the same interval counts to choose the intervals
and to predict how many null p-values they contain. Budgeted BONuS-GA
separates the two roles by \emph{masking}, following the bag-of-null-statistics
construction of \citet{Yang2021BONuS}, which adapts the counting knockoffs
of \citet{BarberCandes2015} to multiple testing. Each site generates
synthetic p-values that are null by construction---independent
$\mathrm{Uniform}[0,1]$ draws---and mixes them with its real p-values. The
center chooses the intervals using only the \emph{pooled} counts of real and
synthetic p-values, so the choice cannot favour real over synthetic
p-values. Within a selected region, the real null p-values and the synthetic
p-values are then exchangeable, and the synthetic count is a valid estimate
of the unknown number of real nulls, up to a scale factor that the tail
counts supply. Because null p-values are exactly uniform, the synthetic
p-values are generated, not estimated.

The one new ingredient relative to \citet{Yang2021BONuS} is how the sites
are combined. Ranking is done \emph{across} sites, on the pooled counts, so
the most promising intervals in the network are admitted first whatever site
they belong to. Calibration is done \emph{within} each site: site $i$ compares
its own synthetic count in the selected region with its own share
$\alpha_i$ of the level, using its own tail counts for scale. Exchangeability
of real nulls and synthetic p-values then only needs to hold within a site,
where it holds for any number of synthetic p-values; nothing has to be
known about how the nulls are distributed across sites.

\subsection{The procedure}

\begin{algo}[Budgeted BONuS-GA]\label{alg:bonus-budget}
\emph{Inputs.} At site $i$: the p-values $P^{(i)}_1,\dots,P^{(i)}_{m^{(i)}}$
and the site size $m^{(i)}$. Globally: the target $\alpha\in(0,1)$, the
bandwidth $\varepsilon>0$, the Storey threshold $\lambda\in(0,1)$ (default
$\tfrac12$), the synthetic multiplier $c>0$ (default $1$), and a budget
rule (Step~4).
\begin{enumerate}[leftmargin=*]
\item \textbf{Synthetic p-values.} Site $i$ draws $\widetilde m^{(i)}$
independent $\mathrm{Uniform}[0,1]$ p-values, labelled synthetic, and forms
its \emph{pooled sample}, the multiset of its $m^{(i)}$ real and
$\widetilde m^{(i)}$ synthetic p-values. Any allocation
$(\widetilde m^{(i)})_i$ that does not depend on the data is allowed; the
simplest is $\widetilde m^{(i)}=\lceil c\,m^{(i)}\rceil$, which needs no
coordination. The experiments use the proportional multinomial
allocation $(\widetilde m^{(i)})_i\sim\mathrm{Multinomial}(\lceil cm\rceil,(m^{(i)}/m)_i)$,
drawn once at the center, which has
$\E[\widetilde m^{(i)}]=\lceil cm\rceil\,m^{(i)}/m$ ($=cm^{(i)}$ when $cm$
is an integer). Write $\widetilde N^{(i)}(\Gamma)$ for the number of site
$i$'s synthetic p-values in $\Gamma$.
\item \textbf{Grid.} From its pooled sample, site $i$ computes the pooled
tail counts $N^{(i)}(A)$ and $\widetilde N^{(i)}(A)$ on $A=(\lambda,1]$,
the pooled Storey estimate
\[
\widehat r_0^{(i)}=\min\Bigl\{1,\
\frac{1+N^{(i)}(A)+\widetilde N^{(i)}(A)}
{(1-\lambda)\,(m^{(i)}+\widetilde m^{(i)})}\Bigr\},
\]
its pooled share $\widehat q^{(i)}=(m^{(i)}+\widetilde m^{(i)})/(m+\widetilde m)$,
where $\widetilde m=\sum_i\widetilde m^{(i)}$, the width
$\widehat L^{(i)}=\varepsilon/(\widehat q^{(i)}\widehat r_0^{(i)})$, and the
candidate intervals $\widehat I^{(i)}_j=\widehat L^{(i)}(j-1,j]$,
$1\le j\le K^{(i)}:=\lfloor\lambda/\widehat L^{(i)}\rfloor$, in the head
$(0,\lambda]$ (any final remainder of width less than $\widehat L^{(i)}$ is
unused).
\item \textbf{Upload and ranking.} Site $i$ sends to the center, for each
candidate interval, the pair (pooled count, real count),
$\bigl(N^{(i)}(\widehat I^{(i)}_j)+\widetilde N^{(i)}(\widehat I^{(i)}_j),\
N^{(i)}(\widehat I^{(i)}_j)\bigr)$, together with its two tail counts. The
center orders all candidate intervals in the network by decreasing pooled
count (ties by site index, then interval index) and admits them one at a
time. After $M$ admissions, $\widehat\Gamma_M^{(i)}$ denotes the union of
the admitted intervals at site $i$ (empty if none), $R(M)=\sum_iN^{(i)}(\widehat\Gamma_M^{(i)})$
the number of real p-values in the admitted intervals, and
$\mathcal N(M)=\{i:\widehat\Gamma_M^{(i)}\ne\emptyset\}$ the set of
\emph{active} sites. The ranking uses pooled counts only, so the nested
family $\{\widehat\Gamma_M\}_M$ is determined by the pooled samples.
\item \textbf{Budget shares.} The center computes shares
$\alpha_1,\dots,\alpha_N\ge0$ with $\sum_i\alpha_i\le\alpha$ from the pooled
counts only. The simplest permissible rule is to split $\alpha$ in
proportion to the sites' pooled candidate mass
$\sum_j\bigl(N^{(i)}(\widehat I_{j}^{(i)})+\widetilde N^{(i)}(\widehat I_{j}^{(i)})\bigr)$.
The rule used in all our experiments gives more of the level to sites
whose candidate intervals are denser than uniform: with
$N_{\mathrm{pool}}^{(i)}=N^{(i)}+\widetilde N^{(i)}$ and $H_i$ the number
of site $i$'s pooled p-values in the head,
\[
w_i=\sum_{j}\Bigl[N_{\mathrm{pool}}^{(i)}(\widehat I_{j}^{(i)})
-H_i\,|\widehat I_{j}^{(i)}|/\lambda\Bigr]_+ ;
\]
if every $w_i$ is zero, $w_i$ is replaced by the site's pooled candidate
mass, and if these fallback weights also sum to zero the procedure makes
no rejections. Sites with $w_i<0.05\max_\ell w_\ell$ are dropped, together
with their candidate intervals, and $\alpha_i=\alpha\,w_i/\sum_{\ell\
\mathrm{retained}}w_\ell$ for the retained sites, $\alpha_i=0$ otherwise.
\item \textbf{Acceptance.} Set
\begin{equation}\label{eq:bonus-budget}
\widehat M=\max\Bigl\{M:\;
\frac{N^{(i)}(A)+1}{\widetilde N^{(i)}(A)}
\bigl(\widetilde N^{(i)}(\widehat\Gamma_M^{(i)})+1\bigr)
\le\alpha_i\,(R(M)\vee1)\ \ \text{for every }i\in\mathcal N(M)\Bigr\},
\end{equation}
with $\widehat M=0$ if no $M$ qualifies; the left side is $+\infty$ when
$\widetilde N^{(i)}(A)=0$, so such a site can never be active. The left
side is site $i$'s estimate of its number of false discoveries: the
synthetic count in its region, scaled by the ratio of real to synthetic
p-values in its tail. Non-null p-values may fall in the tail; they only
increase the ratio, so the estimate errs on the conservative side. The two
``$+1$''s are the usual offsets that keep the estimate valid when counts are
small \citep{BarberCandes2015}.
\item \textbf{Output.} The center sends each site the indices of its
intervals among the first $\widehat M$ admitted, and the rejection set is
$\widehat{\mathcal R}=\{(i,k):P^{(i)}_k\in\widehat\Gamma_{\widehat M}^{(i)}\}$,
the real p-values in the admitted intervals.
\end{enumerate}
\end{algo}

\noindent Steps~1, 2 and 6 are local to the sites; Steps~3--5 are at the
center. Formally, the quantities that may drive the ranking and the budget
shares are those measurable with respect to the \emph{masking
$\sigma$-field} generated by the pooled samples,
\begin{equation}\label{eq:F-mask-def-ext}
\F_{\mathrm{mask}}^{\mathrm{ext}}:=\sigma\Bigl(\bigl\{\,\{\!\{ P:\
P\text{ a real or synthetic p-value at site }i\}\!\}\,\bigr\}_{i=1}^N\Bigr),
\end{equation}
where $\{\!\{\cdot\}\!\}$ denotes a multiset: it records where the pooled
p-values lie at each site but not which of them are real. The grid of
Step~2, the ranking and nested family of Step~3, and the site filter and
budget shares of Step~4 are functions of the pooled samples, hence
$\F_{\mathrm{mask}}^{\mathrm{ext}}$-measurable. The real counts and the
separate tail counts transmitted in Step~3 are not; they are used only in
the acceptance condition of Step~5, which is allowed because by then the
nested family and the shares are fixed.

\subsection{Main result}

\begin{theorem}[FDR control for budgeted BONuS-GA]\label{thm:bonus-pernode}
Under Assumption~\ref{ass:fixed}, Algorithm~\ref{alg:bonus-budget}
satisfies $\FDR\le\alpha$ for every $m\ge1$, every $\varepsilon>0$, every
$\lambda\in(0,1)$, every allocation $(\widetilde m^{(i)})_i$ of synthetic
p-values that does not depend on the data, and every choice of budget
shares that is $\F_{\mathrm{mask}}^{\mathrm{ext}}$-measurable with
$\sum_i\alpha_i\le\alpha$. No null rates or other oracle input enter the
procedure. The same holds under Assumption~\ref{ass:model}. The proof is in
Appendix~\ref{app:bonus}.
\end{theorem}

\paragraph{Why it works.} Write $\FDP=V/(R\vee1)$ with $V$ the number of
real nulls in $\widehat\Gamma_{\widehat M}$. For every active site the
acceptance condition at $\widehat M$ gives
\[
\frac{V^{(i)}(\widehat\Gamma_{\widehat M}^{(i)})}{R(\widehat M)\vee1}
\;\le\;\alpha_i\,Q_i,\qquad
Q_i:=\frac{V^{(i)}(\widehat\Gamma_{\widehat M}^{(i)})}
{\widetilde N^{(i)}(\widehat\Gamma_{\widehat M}^{(i)})+1}\cdot
\frac{\widetilde N^{(i)}(A)}{N^{(i)}(A)+1},
\]
so $\FDP\le\sum_i\alpha_iQ_i$ for every realization. Condition on the
null sets, on the non-null p-values (which may be arbitrarily dependent),
and on the pooled samples; write $\mathcal M$ for this information. Given
$\mathcal M$, the real nulls and the synthetic p-values at site $i$ are
exchangeable: every assignment of the labels ``real'' and ``synthetic'' to
the site's masked (null or synthetic) p-values is equally likely, like
draws from an urn, and the urns at different sites are independent.
Revealing the labels one at a time, backward along the nested regions,
makes the ratio of real to synthetic counts in the region a backward
supermartingale, and an elementary identity for a two-cell urn
(Lemma~\ref{lem:urn}) bounds its starting value by the tail counts; this
gives $\E[Q_i\mid\mathcal M]\le1$ at the stopping index. The bound must
hold conditionally on $\mathcal M$ because the shares $\alpha_i$ are
random, computed from the pooled samples; since they are
$\mathcal M$-measurable,
\[
\FDR\;\le\;\E\Bigl[\sum_i\alpha_i\,\E[Q_i\mid\mathcal M]\Bigr]
\;\le\;\E\Bigl[\sum_i\alpha_i\Bigr]\;\le\;\alpha .
\]
The per-site budgets are the price of the proof: a site that does not use
its whole share cannot pass the rest to other sites, and every active site
carries its own $+1$ offsets. At sites with many p-values the offsets are
negligible; at sites with few p-values they cost power
(Section~\ref{sec:sim}).

\subsection{A worked realization}
\label{sec:worked}

Table~\ref{tab:toy} shows one realization of Algorithm~\ref{alg:bonus-budget}
on a small instance of the monitoring example, generated as in
Section~\ref{sec:sim-setup}: site A has $500$ p-values, each non-null
independently with probability $0.3$ and with signal shift $3$; site B has
$500$ p-values, non-null with probability $0.3$ and shift $2.5$; site C has
$200$ null p-values; $\alpha=0.2$,
$\lambda=\tfrac12$, $c=1$ with $\widetilde m^{(i)}=m^{(i)}$, and a coarse
bandwidth $\varepsilon=0.02$ so that the grid is small enough to display
($8$, $8$ and $4$ candidate intervals). The pooled tail counts give the
scale factors $(N^{(i)}(A)+1)/\widetilde N^{(i)}(A)=0.680$, $0.743$ and
$1.178$, and the budget rule of Step~4 gives $\alpha_A=0.098$,
$\alpha_B=0.090$ and $\alpha_C=0.012$; site C's candidate intervals have
much less empirical excess mass over the uniform profile than the hubs',
so it receives a small share. The first two
admitted intervals are the lowest intervals of the two hubs, and after
$M=3$ admissions the acceptance condition holds at both active sites; at
$M=4$ site B's estimated false discoveries ($38.6$) exceed its budget
($33.5$), and no larger $M$ qualifies, so $\widehat M=3$ and the procedure
rejects the $342$ real p-values in those three intervals. In this
realization $59$ of them are null, an FDP of $0.17$; the FDR guarantee is
a statement about the average of this quantity over the data and the
synthetic draws, not about any single realization.

\begin{table}[b]
\centering
\small
\caption{One realization of Algorithm~\ref{alg:bonus-budget} on the small
instance of Section~\ref{sec:worked}: the first five admitted intervals in
the center's ranking. ``Est.'' is site $i$'s estimated number of false
discoveries, $\frac{N^{(i)}(A)+1}{\widetilde N^{(i)}(A)}(\widetilde
N^{(i)}(\widehat\Gamma_M^{(i)})+1)$, and ``budget'' is $\alpha_iR(M)$;
the condition must hold at every active site. The last column, the true
number of nulls in the interval, is unknown to the procedure.}
\label{tab:toy}
\setlength{\tabcolsep}{3.5pt}
\begin{tabular}{@{}c c c r r r r l c r@{}}
\toprule
$M$ & site & interval & pool & real & syn. & $R(M)$ & est.\ vs.\ budget (active sites) & acc. & nulls \\
\midrule
1 & A & $(0,0.056]$ & 176 & 159 & 17 & 159 & A: $12.2\le15.5$ & yes & 16 \\
2 & B & $(0,0.056]$ & 170 & 146 & 24 & 305 & A: $12.2\le29.8$; B: $18.6\le27.4$ & yes & 14 \\
3 & A & $(0.056,0.112]$ & 66 & 37 & 29 & 342 & A: $31.9\le33.4$; B: $18.6\le30.8$ & yes & 29 \\
4 & B & $(0.225,0.281]$ & 58 & 31 & 27 & 373 & A: $31.9\le36.4$; B: $38.6>33.5$ & no & 26 \\
5 & C & $(0.367,0.490]$ & 58 & 24 & 34 & 397 & \dots; C: $41.2>4.9$ & no & 24 \\
\bottomrule
\end{tabular}
\end{table}

\subsection{Communication}
\label{sec:bonus-comm}

Budgeted BONuS-GA uses two communication rounds. The first is needed only
if the synthetic allocation is coordinated (the multinomial allocation of
Step~1): each site sends its size and the center returns its number of
synthetic p-values, $O(N\log m)$ bits in all. In the second round each
site sends the pairs of Step~3 and its two tail counts, and the center
returns the accepted interval indices. The network has
$\sum_iK^{(i)}\le\lambda/\varepsilon$ candidate intervals in total, because
$\sum_i\widehat q^{(i)}=1$ and $\widehat r_0^{(i)}\le1$; at the standard
bandwidth $\varepsilon\asymp m^{-1/2}$ this is $O(\sqrt m)$ intervals, each
costing $O(\log m)$ bits (a pooled count is at most $m+\lceil cm\rceil$, so
generating $c$ synthetic p-values per real one adds at most
$\lceil\log_2(2+c)\rceil$ bits per count). The total is $O((N+\sqrt m)\log m)$ bits, which is the
original greedy's $O(\sqrt m\log m)$ whenever $N=O(\sqrt m)$. Synthetic
p-values never leave the sites. Appendix~\ref{app:protocols} gives the
message formats and an interactive variant.

\subsection{Remarks and variants}

\begin{remark}[Fixed versus random site labels]\label{rem:label-free}
The proof of Theorem~\ref{thm:bonus-pernode} conditions on the null
sets and the non-null p-values, runs an urn argument \emph{within} each
site, and uses only independence \emph{across} sites; no step averages
over site labels, which is why Assumption~\ref{ass:fixed} suffices. This
contrasts with
Theorems~\ref{thm:cfga} and~\ref{thm:aug}, whose lifted-space argument
genuinely requires the random labels of Assumption~\ref{ass:model}: under
fixed labels the held-out thinning ratios become node-specific and the
single supermartingale of Lemma~\ref{lem:lifted} is no longer available
(a per-node budgeted CFGA, in the spirit of
Algorithm~\ref{alg:bonus-budget}, would restore a fixed-label guarantee
at the corresponding power cost).
\end{remark}

\begin{remark}[Split-free adaptive $\varepsilon$]\label{rem:bonus-eps}
The bandwidth $\varepsilon$ can be tuned without reintroducing a split. Because
the masked greedy is driven only by the \emph{pooled} p-values, any choice of
$\varepsilon$ that is a function of the pooled values---for instance the one
maximizing a standardized $\chi^2$ excess of the pooled counts over uniform on
$[0,\lambda]$---is measurable with respect to the extended masking
$\sigma$-field $\mathcal{F}_{\mathrm{mask}}^{\mathrm{ext}}$ of
\eqref{eq:F-mask-def-ext} (it may depend on all pooled locations, not only
on the candidate-cell counts) and leaves Theorems~\ref{thm:bonus-pernode} and~\ref{thm:bonus} intact.
Note that the rejection-maximizing sub-split cross-validation of
Section~\ref{sec:eps-adaptive} is \emph{not}
$\mathcal{F}_{\mathrm{mask}}^{\mathrm{ext}}$-measurable, and hence not available here: a
rejection count separates real from synthetic points inside the candidate
region---exactly the masked information the knockoff protects---so the split is
what licensed that recipe for CFGA. Empirically the pooled-score rule helps
precisely where the fixed default
$\varepsilon=\alpha\,m^{-1/2}$ is too conservative: in a side run of the
Experiment~5 configuration at $m_{\mathrm{hub}}=250$ (not shown in the
figures) it lifts the power of global plug-in BONuS-GA from $0.60$ to $0.77$,
but its empirical FDR rises to $\approx0.22$, marginally above the target
$\alpha=0.2$---a reminder that the global plug-in variant of
Remark~\ref{rem:oracle-only} has no established
finite-sample guarantee and that data-driven tuning consumes whatever
conservatism slack the fixed default leaves. It is therefore a targeted
tool for the small-sample regime, to be used with a validity margin in
mind, rather than the universal default that adaptive $\varepsilon$ is for
cross-fit CFGA.
\end{remark}

\begin{remark}[Oracle and global plug-in variants]\label{rem:variants}
Replacing the per-site calibration of Step~5 by a single network-wide
ratio, $\widehat\FDP(M)=\frac{N(A)+1}{\widetilde N(A)}\cdot\frac{\widetilde
N(\widehat\Gamma_M)+1}{R(M)\vee1}$ with $N(\cdot)=\sum_iN^{(i)}(\cdot)$ and
$\widetilde N(\cdot)=\sum_i\widetilde N^{(i)}(\cdot)$, gives the
\emph{global} form of BONuS-GA. It controls FDR when the synthetic
p-values are allocated to sites in proportion to the true null masses
$\pi^{(i)}$ (the \emph{oracle} allocation; Theorem~\ref{thm:bonus} in
Appendix~\ref{app:bonus-variants}, under Assumption~\ref{ass:model}). With
the Storey estimates in place of $\pi^{(i)}$ (the \emph{global plug-in}
variant) it has no finite-sample guarantee, but it is the most powerful
practical variant in our simulations, where it serves as the heuristic
reference for the budgeted procedure. Appendix~\ref{app:bonus-variants}
gives both variants and the proof of Theorem~\ref{thm:bonus}.
BONuS-GA is randomized through the synthetic draw; averaging over several
draws, or combining their e-values \citep{RenBarber2024evalues}, would
reduce this randomness, and a finite-sample analysis of such variants is
left to future work.
\end{remark}

\section{A sample-splitting alternative: the CFGA family}
\label{sec:alt}

The second way to separate selection from calibration is to split the
data. The procedures of this section are analysed under the random-site
model (Assumption~\ref{ass:model}), under which the site sizes are random;
unlike Theorem~\ref{thm:bonus-pernode}, their proofs do not cover fixed
designs (Remark~\ref{rem:label-free}). They use the clipped Storey estimate
\eqref{eq:storey} and the same head-restricted grids as
Algorithm~\ref{alg:bonus-budget}, applied to the ranking half.

\subsection{CFGA: cross-fit greedy aggregation}
\label{sec:method}

CFGA addresses the winner's curse of Section~\ref{sec:why} by separating
the role of \emph{ranking} the intervals from the role of \emph{counting}
the rejections. At each node we randomly partition the p-values into two
halves of equal expected size (by independent fair coin flips) and use
each in only one role: the \emph{ranking half}
$\D_1$ decides which intervals to reject; the \emph{inference half}
$\D_2$ supplies the count $R$ that goes into the FDP estimate. Because
the inference-half counts are independent of the ranking that selected
the regions, the $\sqrt{\log|\Aset|}$ subset-maximization inflation of
Appendix~\ref{sec:gap-bias} collapses to a single unselected
$O((\varepsilon m)^{-1/2})$ Monte Carlo deviation per admitted
region---selection-free, though not smaller in order at the default
bandwidth; the proof handles it exactly, by a martingale identity rather
than a maximal inequality. We run this twice with the roles of the
halves swapped, and union the rejections with a Bonferroni-2 correction
so that each fold operates at level $\alpha/2$.

\begin{algo}[CFGA]\label{alg:cfga}
Fix $\varepsilon>0$ and target $\alpha\in(0,1)$.
\begin{enumerate}[leftmargin=*]
\item \textbf{Split.} Independently for each hypothesis $k\in\{1,\dots,m\}$,
draw $Z_k\sim\mathrm{Bernoulli}(1/2)$ and assign $k$ to fold $b=1+Z_k$. Let
$\D_b^{(i)}=\{k:Z_k=b-1,\,I_k=i\}$ and $\D_b=\bigcup_i\D_b^{(i)}$. The
fold sizes $m_b^{(i)}:=|\D_b^{(i)}|$ and $m_b:=|\D_b|$ are random with
$\E[m_b^{(i)}]=m^{(i)}/2$, $\E[m_b]=m/2$. The Bernoulli assignment is
drawn independently of the data and uses purely local randomness.

\item \textbf{Fold 1 (rejections in $\D_2$, level $\alpha/2$).}
\begin{enumerate}[label=(1\alph*)]
\item \emph{Ranking on $\D_1$.} On the ranking half, each node $i$ forms a
Storey grid: it estimates the local null proportion
\[
\widehat r_0^{(i)}=\min\Bigl\{1,\ \frac{1+\#\{k\in\D_1^{(i)}:P_k^{(i)}>\lambda\}}{(1-\lambda)\,m_1^{(i)}}\Bigr\}
\quad(\lambda=\tfrac12\text{ by default}),
\]
the clipped Storey estimator \eqref{eq:storey} applied to the
ranking half (its $+1$ offset keeps it strictly positive, so the width
below is finite even when the fold has no tail observations; an empty
node-fold contributes no intervals),
sets the interval width
$\widehat L^{(i)}=\varepsilon/(\widehat q_1^{(i)}\widehat r_0^{(i)})$ with
the \emph{ranking-fold} proportion $\widehat q_1^{(i)}:=m_1^{(i)}/m_1$
(the node sizes are random under Assumption~\ref{ass:model}, and using
the fold proportion keeps the grid $\D_1$-measurable, as the proof
requires; this matches the reference implementation), and
places the full-width intervals
$\widehat I_j^{(i)}=\widehat L^{(i)}(j-1,j]$,
$1\le j\le\lfloor\lambda/\widehat L^{(i)}\rfloor$, in the head $(0,\lambda]$
(any final remainder, of width less than $\widehat L^{(i)}$, is unused).
Reporting only the per-interval
counts $N^{(i)}(\widehat I_j^{(i)})$, the center node ranks all intervals
across the network in decreasing count and admits them greedily; after $M$
admissions the rejection region at node $i$ is the union
$\widehat\Gamma_M^{(i,1)}\subseteq[0,\lambda]$ of its admitted intervals---a
nested family in $M$.
\item \emph{Acceptance on $\D_2$.} The grid width
$\widehat L^{(i)}=\varepsilon/(\widehat q_1^{(i)}\widehat r_0^{(i)})$ makes the
predictable null mass of a $k_i$-interval region cancel the estimated rate,
$m_2\widehat q_1^{(i)}\widehat r_0^{(i)}\nu(\widehat\Gamma_M^{(i,1)})
=m_2\widehat q_1^{(i)}\widehat r_0^{(i)}\,k_i\widehat L^{(i)}=m_2\varepsilon k_i$.
Summing over nodes, the count-based FDP estimator collapses to
\begin{equation}\label{eq:fdpA}
\widehat\FDP_2(M)\;=\;\frac{1+m_2\,\varepsilon\,M}{R_2(M)\vee 1},
\qquad M=\textstyle\sum_i k_i,
\end{equation}
where $m_2=\sum_i m_2^{(i)}$, $M$ is the total number of selected intervals,
and $R_2(M)=\sum_{i}\sum_{k\in\D_2^{(i)}}\1\{P_k^{(i)}\in\widehat\Gamma_M^{(i,1)}\}$.
Let $\widehat M^{(1)}=\max\{M:\widehat\FDP_2(M)\le\alpha/2\}$ ($0$ if empty);
$\mathcal R^{(1)}\!=\!\{k\!\in\!\D_2^{(i)}\!:\!P_k^{(i)}\!\in\!\widehat\Gamma_{\widehat M^{(1)}}^{(i,1)}\}$.
\end{enumerate}

\item \textbf{Fold 2 (rejections in $\D_1$, level $\alpha/2$).}
Repeat Step~2 with the roles of $\D_1$ and $\D_2$ swapped; produce
$\mathcal R^{(2)}\subseteq\D_1$.

\item \textbf{Output} $\mathcal R^{\mathrm{CFGA}}=\mathcal R^{(1)}\cup\mathcal R^{(2)}$.
\end{enumerate}
\end{algo}

\paragraph{Intuition.}
The numerator $m_2\varepsilon M$ is the original greedy's predicted null
count (by the grid design the plug-in model predicts $m_2\varepsilon$
null p-values of the inference half per interval; the actual conditional
expected number is $m_2\pi^{(i)}\widehat L^{(i)}$); the ``$+1$'' is the
standard Selective SeqStep$+$ offset. The
structural form is the same as the original greedy estimator; what changes
is the data on which it is evaluated: the regions $\widehat\Gamma_M^{(i,1)}$
are formed on $\D_1$, and $R_2(M)$ counts rejections on the held-out half
$\D_2$. Because the family $\{\widehat\Gamma_M^{(i,1)}\}$ is independent of
$\D_2$, the noise in $R_2(M)$ is genuine sampling variation, free of the
selection optimism (winner's curse) affecting the in-sample density
estimates of the original procedure.

\paragraph{Step-up acceptance, not greedy stopping.}
The index $\widehat M^{(1)}=\max\{M:\widehat\FDP_2(M)\le\alpha/2\}$ is a
\emph{step-up} cut on the greedy ranking, not the greedy ``first-crossing'' stop
$\min\{M:\widehat\FDP_2>\alpha/2\}-1$ of the original algorithm. The two agree
only when the FDP path is monotone: in \citet{PournaderiXiang2024} the in-sample
estimate $M/\sum_{\ell\le M}\widehat h_{(\ell)}$ is nondecreasing in $M$ (the
densities enter in decreasing order), but our held-out estimate
$\widehat\FDP_2(M)=(1+m_2\varepsilon M)/(R_2(M)\vee1)$ pairs a deterministic
increasing numerator with a \emph{random} (merely nondecreasing) denominator
$R_2(M)$, so it can dip back below $\alpha/2$ after first crossing it, and the
two indices differ in general. We adopt the step-up form because (i) it is
exactly the backward stopping time controlled by Lemma~\ref{lem:lifted}---the
forward greedy stop is not a backward stopping time, so the finite-sample bound
would not apply to it---and (ii) it rejects a superset of the forward stop, hence
is at least as powerful, reducing to the original rule whenever the path is
monotone. The greedy ranking itself is unchanged: CFGA still rejects a prefix
$\widehat\Gamma_{\widehat M^{(1)}}$ of the density-ordered intervals, and only the
cut point moves from the first crossing to the last index at or below $\alpha/2$.

\paragraph{Communication.} After a preliminary exchange of fold sizes,
CFGA needs one uplink of per-interval counts and one downlink of accepted
interval indices: $O((N+\sqrt m)\log m)$ bits in two rounds, i.e.\ the
original $O(\sqrt m\log m)$ budget when $N=O(\sqrt m)$
(Appendix~\ref{app:protocols-cfga}).

\paragraph{Finite-sample FDR control.}

Fix a fold $b\in\{1,2\}$ with ranking half $\D_b$ and inference half
$\D_{b'}$ of size $m_{b'}$ ($b':=3-b$), and condition on $\D_b$, so that the
regions $\widehat\Gamma_M^{(i,b)}$ are $\sigma(\D_b)$-measurable. With the
lifted hypothesis $\widetilde X_k=(I_k,P_k)$ from
Section~\ref{sec:setting} and the lifted region
$\widetilde\Gamma_M=\bigcup_i\{i\}\times\widehat\Gamma_M^{(i,b)}$, let
\[
V(M)=\#\{k\in\D_{b'}:k\text{ null},\,\widetilde X_k\in\widetilde\Gamma_M\}
\]
be the inference-half null count, and recall the lifted null mass
$p_M=\sum_i\pi^{(i)}\nu(\widehat\Gamma_M^{(i,b)})$ from
Section~\ref{sec:setting}, so $m_{b'}p_M=\E[V(M)\mid\D_b]$.
Algorithms~\ref{alg:cfga} and~\ref{alg:aug} use the convention that a fold
whose ranking half is empty ($m_b=0$, in which case no grid, no Storey
estimate and no margin are computed), whose candidate family is empty, or
whose inference half is empty ($m_{b'}=0$) makes no rejections; in e-CFGA
(Section~\ref{sec:ecfga}) such a fold assigns the score $0$ to every
hypothesis.

\begin{theorem}[Finite-sample FDR for CFGA under known null masses]\label{thm:cfga}
Under Assumption~\ref{ass:model}, suppose each fold of Algorithm~\ref{alg:cfga}
uses the \emph{oracle} null mass $\bar V_{\mathrm{or}}(M)=m_{b'}p_M$ (built from
the true null masses $\pi^{(i)}$) in place of the rate-free $m_{b'}\varepsilon M$,
stopping at
$\widehat M^{(b)}=\max\{M:(1+\bar V_{\mathrm{or}}(M))/(R(M)\vee1)\le\alpha/2\}$
($\widehat M^{(b)}=0$ if the set is empty; the superscript indicates the
fold, while $\widehat M^{(b)}$ depends on both $\D_b$ and $\D_{b'}$ through
the regions and the count $R(M)$).
Then $\FDR\le\alpha$ for every $m\ge1$\footnote{With the conventions:
an empty ranking fold at a node contributes no candidate intervals from
that node, an empty inference fold contributes zero counts, and an empty
candidate family makes no rejections. Each occurs with positive
probability under the Bernoulli split when $m^{(i)}$ is tiny, and the
proof applies verbatim.} and every $\varepsilon>0$. The proof---a
lifted-space adaptive lemma (Lemma~\ref{lem:lifted}) followed by a Bonferroni
union over the two folds---is deferred to Appendix~\ref{app:lifted}.
\end{theorem}

\paragraph{CFGA is a plug-in heuristic.} The oracle mass needs the unknown null
rates $\pi^{(i)}=q^{(i)}r_0^{(i)}$. CFGA supplies them \emph{implicitly}
through the Storey grid $\widehat L^{(i)}=\varepsilon/(\widehat q_b^{(i)}\widehat r_0^{(i)})$
(fold proportions, as in Algorithm~\ref{alg:cfga}):
plugging the same $\widehat q_b^{(i)}\widehat r_0^{(i)}$ into the null mass gives
$m_{b'}\,\widehat q_b^{(i)}\widehat r_0^{(i)}\,\nu(\widehat\Gamma_M^{(i,b)})
=m_{b'}\,\varepsilon\,k_i$, so the rate \emph{cancels} and the numerator
collapses to the rate-free $1+m_{b'}\varepsilon M$ of \eqref{eq:fdpA}. This
matches the original greedy estimator \citep{PournaderiXiang2024}
(scored cross-fit) up to the additive ``$+1$'' SeqStep$+$ offset, which
our finite-sample supermartingale argument uses for convenience (it can
be removed; Remark~\ref{rem:no-offset} in Appendix~\ref{app:lifted}) and which does not
appear in the asymptotic form of \citet{PournaderiXiang2024}.
Since these grid estimates may fall below the truth, the plug-in mass need not
dominate $m_{b'}p_M$ and the guarantee of Theorem~\ref{thm:cfga} is not assured
at finite $m$: CFGA is a \emph{heuristic}. Empirically it stays below
$\alpha$ in the independent-data configurations of Section~\ref{sec:sim},
but its adaptive-bandwidth version reaches $\FDR\approx0.32$ under strong
within-node dependence (Experiment~6). Section~\ref{sec:augmented} (Inflated CFGA)
restores rigor by using an inflated rate that does \emph{not} cancel.

\begin{remark}[Bonferroni-2 is loose]
The two folds reject disjoint sets of hypotheses, so the FDP of the union
is the weighted average $w_1\FDP_1+w_2\FDP_2$ of the fold FDPs, with
weights $w_b=R_b/(R_1+R_2)$. The Bonferroni-2 bound replaces this weighted
average by the sum $\FDP_1+\FDP_2$, so the argument controls each fold
at half the target level, $\alpha/2$.
Section~\ref{sec:ecfga} avoids this division: it converts each fold's
rejections into scores and combines them through the e-BH procedure
\citep{WangRamdas2022}, averaging over several random splits;
Theorem~\ref{thm:ecfga} shows that, with the oracle or the inflated null
estimate, this controls FDR at every finite $m$ with \emph{no} Bonferroni
division.
\end{remark}

\subsection{Inflated CFGA: a plug-in estimate with a guarantee}
\label{sec:augmented}

CFGA is a heuristic: the per-node null rate $\pi^{(i)}=q^{(i)}r_0^{(i)}$ enters
its FDP numerator only through the Storey grid, where it cancels and leaves the
rate-free $1+m_2\varepsilon M$ (Theorem~\ref{thm:cfga}). Because that implicit
rate can fall below the truth, the plug-in null mass need not dominate
$m_{b'}p_M$, and the oracle guarantee does not transfer. \emph{Inflated CFGA}
restores rigor by replacing the implicit rate with an \emph{explicit} estimate,
inflated by a vanishing margin so that it dominates the truth with high
probability.

CFGA already confines its rejection regions to the left tail $[0,\lambda]$
(Algorithm~\ref{alg:cfga}). On the \emph{complementary} tail $(\lambda,1]$ of
the ranking half we form a Storey estimate of the per-node null rate,
\begin{equation}\label{eq:pihat}
\widehat\pi_b^{(i)}
=\frac{\#\{k\in\D_b^{(i)}:P_k^{(i)}>\lambda\}}{m_b(1-\lambda)},
\end{equation}
and inflate it by a uniform margin
$\widehat\pi_b^{(i),+}=\widehat\pi_b^{(i)}+\delta$ with
\begin{equation}\label{eq:delta}
\delta=\frac{1}{1-\lambda}\sqrt{\frac{\log(2N/\eta)}{2m_b}} .
\end{equation}
By Hoeffding's inequality (made explicit in the proof of
Theorem~\ref{thm:aug}) the margin guarantees
$\widehat\pi_b^{(i),+}\ge\pi^{(i)}$ for all $i$ with probability at least
$1-\eta$.

\begin{algo}[Inflated CFGA]\label{alg:aug}
Fix $\varepsilon>0$, $\lambda\in(0,1)$ and $\eta\in(0,1)$ (default $\eta=1/m$ for $m\ge2$).
Run CFGA (Algorithm~\ref{alg:cfga}), but in each fold replace the rate-free
acceptance numerator $1+m_{b'}\varepsilon M$ by $1+\bar V(M)$, where
$\bar V(M)=m_{b'}\sum_i\widehat\pi_b^{(i),+}\,\nu(\widehat\Gamma_M^{(i,b)})$ uses
the inflated ranking-half estimate \eqref{eq:pihat}--\eqref{eq:delta}; stop at
$\widehat M^{(b)}=\max\{M:(1+\bar V(M))/(R(M)\vee1)\le\alpha/2\}$ and union the
two folds.
\end{algo}

\paragraph{Communication.} The only addition over CFGA is one tail count
per node per fold, $O(N\log m)$ bits in the same two rounds
(Appendix~\ref{app:protocols-aug}).

\begin{theorem}[FDR control for Inflated CFGA]\label{thm:aug}
Under Assumption~\ref{ass:model}, Algorithm~\ref{alg:aug} satisfies
$\FDR\le\alpha+\eta$ for every $m\ge1$, every $\varepsilon>0$, every
$\lambda\in(0,1)$ and every $\eta\in(0,1)$; with $\eta=1/m$ this reads
$\FDR\le\alpha+1/m$, and the inflation
$\delta=\frac{1}{1-\lambda}\sqrt{\log(2Nm)/(2m_b)}\to 0$
as $m_b\to\infty$ with $\log(Nm)=o(m_b)$, so the penalty vanishes as
the ranking half grows. The proof is in Appendix~\ref{app:lifted}.
\end{theorem}

\paragraph{The margin is the price of rigor.} The inflation
$\delta=O(\sqrt{\log(N/\eta)/m_b})$ enters the FDP numerator as
$m_{b'}\delta\sum_i\nu(\widehat\Gamma_M^{(i,b)})$, which grows with the
number of nodes carrying selected regions; in networks with hundreds of
nodes and moderate per-node samples it is large enough that the
\emph{fixed}-bandwidth Inflated procedure can make no rejections at all
(Section~\ref{sec:sim}, Experiments~1 and~4). Inflated CFGA should
therefore be run with the adaptive bandwidth of Section~\ref{sec:eps-adaptive},
which widens the intervals until the signal mass clears the inflated null
mass; this is the configuration used throughout the experiments.

\paragraph{What the inflated method needs from the estimate.}
The heuristic forms no explicit rate estimate (the rate cancels in
$1+m_{b'}\varepsilon M$), so it imposes no conditions; the inflated
estimate $\widehat\pi_b^{(i),+}$ does. The proof of
Theorem~\ref{thm:aug} relies on two properties:
\emph{(R1) Measurability.} $\widehat\pi_b^{(i),+}$ must be
$\sigma(\D_b)$-measurable, i.e.\ computed from the ranking half only,
so that it is fixed before the inference half enters
Lemma~\ref{lem:lifted}. Estimators that also use inference-half data are
\emph{not} covered --- not even ones built only from the inference-half
tail, which is disjoint from the head $\D_{b'}\cap[0,\lambda]$ that
drives $V(M)$: conditioning on that tail also fixes the inference-half
head count, and controlling the resulting conditional head mass requires
a separate argument (Remark~\ref{rem:pool-tail} in Appendix~\ref{app:lifted}).
\emph{(R2) Dominance.} $\widehat\pi_b^{(i),+}\ge\pi^{(i)}$ on a
high-probability event. Any region $A\subseteq[0,1]$ gives a valid
estimator $\#\{P\in A\}/(|A|\,m_b)$ since signal contamination only
raises the proportion; the Storey choice $A=(\lambda,1]$ is often
useful for \emph{power} when signal p-values concentrate near zero, since
then few of them contaminate the tail. Any estimator satisfying (R1) and (R2)
inherits Theorem~\ref{thm:aug} with $\delta$ replaced by the
corresponding Hoeffding margin.

\subsection{Choosing the bandwidth}\label{sec:eps-adaptive}

The greedy aggregation algorithm and CFGA both take a bandwidth
$\varepsilon>0$, which sets the interval widths and hence the total number of
candidate intervals, $O(1/\varepsilon)$; the original paper uses
$\varepsilon\propto m^{-1/2}$. The following rule chooses $\varepsilon$ from the
data without affecting the CFGA-family guarantees; its practical details
(candidate values, use in both folds, computational cost, and a caveat under
within-node dependence) are in Appendix~\ref{app:eps}.

We propose an adaptive procedure that selects $\varepsilon$ via sub-split
cross-validation on $\D_1$. Writing $U$ for the auxiliary sub-split
randomness, the selected value is $\sigma(\D_1,U)$-measurable, and $U$ is
independent of the data, so finite-sample FDR control is preserved.

\begin{algo}[Adaptive $\varepsilon$ via sub-split CV on $\D_1$]\label{alg:adaptive-eps}
Given $\D_1$, target $\alpha$, and a candidate grid
$\mathcal B=\{\varepsilon_1<\cdots<\varepsilon_T\}$:
\begin{enumerate}[leftmargin=*]
\item Sub-split $\D_1$ uniformly into $\D_{1a}\sqcup\D_{1b}$.
\item For each $\varepsilon_t\in\mathcal B$:
\begin{itemize}
\item Run greedy aggregation on $\D_{1a}$ at spacing $\varepsilon_t$ to
completion. Let $\{\widehat\Gamma_M^{(i,t)}\}$ be the resulting nested
rejection regions.
\item Apply the CFGA acceptance rule \eqref{eq:fdpA} on $\D_{1b}$ at
target $\alpha$. Record the resulting rejection count $R_{1b}^{(t)}$.
\end{itemize}
\item Output $\widehat\varepsilon=\varepsilon_{t^\star}$, where
$t^\star=\min\arg\max_{t\in\{1,\ldots,T\}}R_{1b}^{(t)}$ (ties are broken
toward the smallest bandwidth, as in the reference implementation).
\end{enumerate}
The downstream CFGA (or Inflated CFGA) is then run at $\widehat\varepsilon$ on
the full $\D_1$ for selection and the original $\D_2$ for acceptance.
\end{algo}

\begin{proposition}[FDR with adaptive $\varepsilon$]
\label{prop:adaptive-FDR}
Run with $\widehat\varepsilon$ chosen by Algorithm~\ref{alg:adaptive-eps},
the oracle CFGA of Theorem~\ref{thm:cfga} satisfies $\FDR\le\alpha$ and
Inflated CFGA (Theorem~\ref{thm:aug}) satisfies $\FDR\le\alpha+\eta$, for every
$m\ge1$ and every candidate grid $\mathcal B$.
\end{proposition}

\begin{proof}
$\widehat\varepsilon$ is a deterministic function of $\D_1$ and the
sub-split randomness $U$, so $\widehat\varepsilon\in\sigma(\D_1,U)$. Since
$U$ is independent of $(\D_1,\D_2)$ (and of the split), conditioning on
$(\D_1,U)$ leaves the inference-half hypotheses exactly as they are
conditionally on $\D_1$ alone. The proofs of Theorems~\ref{thm:cfga}
and~\ref{thm:aug} condition on the ranking half throughout and use no
property of $\varepsilon$ beyond measurability with respect to the
conditioning $\sigma$-field and nestedness of the regions; running them
conditionally on $(\D_1,U)$ instead of $\D_1$, both bounds carry over.
\end{proof}

\subsection{e-CFGA: removing the Bonferroni factor}
\label{sec:ecfga}

CFGA carries two residual costs. The first is the Bonferroni-2 split: each
fold is stopped at level $\alpha/2$ and the per-fold rejections are unioned,
so a problem with signal in both folds effectively runs the procedure at
half the target level. The second is the dependence of the output on a
single random split: a different Bernoulli$(\tfrac12)$ draw produces a
different, correlated rejection set, which is undesirable when the
practitioner wants a reproducible decision.

Aggregating across folds through \emph{generalized e-value scores}
\citep{WangRamdas2022} instead of a union bound removes the first cost
and reduces the second (with finitely many splits the output remains
random, but averaging makes it less sensitive to any single split).
This section turns each fold's stopping decision into scores, averages
the scores across folds and across $S$ independent splits, and finishes
with a single e-BH cut at the global level $\alpha$. The procedure inherits
finite-sample $\FDR\le\alpha$ from a fold-level bound proved by
Lemma~\ref{lem:lifted}.

Appendix~\ref{app:ecfga-details} gives the construction. In brief, each
fold $(s,b)$ of split $s$ stops at level $\gamma\in(0,\alpha]$ instead of
$\alpha/2$, and every hypothesis in its inference half that falls in the
fold's rejection region receives the score
$e_k^{(s,b)}=m_{b'}^s/(1+\bar V(\widehat M_{s,b}))$, where $m_{b'}^s$ is the
inference-half size and $\bar V(\widehat M_{s,b})$ the fold's estimate of
the number of nulls in its region; all other hypotheses receive $0$. The
scores are averaged over the $S$ splits, $\bar e_k=S^{-1}\sum_s\sum_b
e_k^{(s,b)}$, and the e-BH procedure of \citet{WangRamdas2022} is applied
once at level $\alpha$: order $\bar e_{(1)}\ge\bar e_{(2)}\ge\cdots$ and
reject the $k^\star=\max\{k:k\,\bar e_{(k)}\ge m/\alpha\}$ largest.
Lemma~\ref{lem:lifted} (Appendix~\ref{app:lifted}) implies that the scores
of each fold have expected null total at most $m^s_{b'}$, so the averaged
scores have expected null total at most $m$, which is all that e-BH needs
\citep{WangRamdas2022}.

\begin{theorem}[FDR control for e-CFGA]\label{thm:ecfga}
Under Assumption~\ref{ass:model}, with the oracle null mass
$\bar V(M)=m^s_{b'}p_M$, e-CFGA (Algorithm~\ref{alg:ecfga} in Appendix~\ref{app:ecfga-details}) satisfies
$\FDR\le\alpha$ for every $m\ge1$, every $S\ge1$, every
$\gamma\in(0,\alpha]$ and every $\varepsilon>0$. With the inflated
plug-in of Section~\ref{sec:augmented} run at confidence $\eta/S$ per
split, $\FDR\le\alpha+\eta$. The proof is in Appendix~\ref{app:ecfga}.
\end{theorem}

The default $\gamma=0.75\alpha$ and the number of splits $S=10$ are
empirical choices discussed in Appendix~\ref{app:ecfga-details}, which
also covers the communication cost, $O((NS+S\sqrt m+\bar R)\log m)$ bits
with $\bar R\le m$ the number of hypotheses with a nonzero score, and the
quantization of the scores.

\section{Why the original greedy fails at finite samples}
\label{sec:why}

Greedy aggregation (Section~\ref{sec:greedy}) is a plug-in procedure: it
ranks the intervals by their counts and stops when the plug-in FDP estimate
$\widehat\FDP(M)=M/\sum_{\ell\le M}\widehat h_{(\ell)}$ first exceeds
$\alpha$, where $\widehat h_{(1)}\ge\widehat h_{(2)}\ge\cdots$ are the
density statistics of the admitted intervals. The same counts therefore
rank the intervals and calibrate the stop. This creates a winner's curse:
the admitted intervals are those whose counts happen to be largest, and
their average density statistic $\bar h_M=M^{-1}\sum_{\ell\le M}\widehat
h_{(\ell)}$ exceeds, in expectation, the average of the same intervals'
conditional means.

Appendix~\ref{sec:gap} quantifies this for a head-restricted variant of the
greedy whose candidate intervals lie in $(0,\lambda]$, the grids used by
CFGA and BONuS-GA. Let $\Delta(M\mid\mathcal G)$ be the conditional
expected selection bias of $\bar h_M$ given the site sizes and tail counts
(Proposition~\ref{prop:wc-bias-body}). Assume, in addition to
Assumption~\ref{ass:model}, that every site has a positive null rate,
$\min_ir_0^{(i)}>0$, with $r_{\min}\in(0,\min_ir_0^{(i)})$ fixed, and that
each non-null CDF $F_i$ has a bounded density $f_i$. These conditions are
stronger than what Theorems~\ref{thm:bonus-pernode}, \ref{thm:cfga},
\ref{thm:aug} and~\ref{thm:ecfga} need and enter only here; the
bounded-density condition excludes, for example, the one-sided
Gaussian-shift p-values of Section~\ref{sec:sim}, whose density is
unbounded at $0$, so the propositions below describe the mechanism rather
than the exact experimental model. Then, with $|\Aset|$ the number of
candidate intervals and constants $C_1,C_2$ depending only on
$\sup_i\|f_i\|_\infty$, $r_{\min}$ and $\lambda$, on the event
$\{\min_i\widehat r_0^{(i)}\ge r_{\min}\}$, whose probability is
exponentially close to one,
\[
\Delta(M\mid\mathcal G)\;\le\;\sqrt{\frac{C_1\log(e|\Aset|/M)}{\varepsilon m}}
+\frac{C_2\log(e|\Aset|/M)}{\varepsilon m}
\qquad\text{for every }M ,
\]
which is $O(m^{-1/4}\sqrt{\log m})$ at the standard bandwidth
$\varepsilon\asymp m^{-1/2}$; and under the global null this order is
attained for $M$ in a range growing with $m$
(Proposition~\ref{prop:wc-lower}). For comparison, the density statistic of
a single \emph{unselected} interval fluctuates around its conditional mean
at scale $(\varepsilon m)^{-1/2}$, which is $m^{-1/4}$ at the standard
bandwidth (Appendix~\ref{sec:gap-conc}); selection adds the logarithmic
factor and, unlike that fluctuation, is systematically one-sided. Two
qualifications are
important. The propositions concern the density statistics at a fixed
$M$, not the FDP estimate (a reciprocal) or the FDR at the data-chosen
stopping index, and they do not prove that this mechanism explains the FDR
excess observed in Section~\ref{sec:sim}; they show that a bias of this
order is present. CFGA removes it by computing the counts that enter
$\widehat\FDP$ on a half of the data that played no role in the ranking;
BONuS-GA removes it because, given the pooled counts, selection inflates
real-null and synthetic counts alike.

\section{Numerical illustration}
\label{sec:sim}

The experiments address three questions in turn: why finite-sample
control fails (Section~\ref{sec:sim-why}), what guaranteed validity costs
(Section~\ref{sec:sim-cost}), and which method to use in which regime
(Section~\ref{sec:sim-regime}).

\subsection{Design}
\label{sec:sim-setup}

We return to the monitoring network of Section~\ref{sec:example}. The
network has $N=302$ nodes: two signal-rich \emph{hubs}, each with
$m^{(i)}=1000$ p-values and non-null proportion $r_1^{(i)}=0.3$, and $300$
\emph{quiet} pure-null nodes with $m^{(i)}=20$ p-values each. A test
statistic is $X_k=Z_k+B_k\mu_k$ with $Z_k\sim\mathcal N(0,1)$, a non-null
indicator $B_k\sim\mathrm{Bernoulli}(r_1^{(i)})$ and, at a hub, a shift
$\mu_k\sim\mathrm{Unif}[\mu_{\mathrm{hub}}-0.5,\mu_{\mathrm{hub}}+0.5]$ with
$\mu_{\mathrm{hub}}=3$; the indicators, shifts and noise variables are
mutually independent across $k$ and across sites (except in
Experiment~6, where the noise is autocorrelated), so the number of
non-nulls at a site is random. The p-values are one-sided,
$P_k=1-F_0(X_k)$ with $F_0$ the null CDF of $Z_k$ (standard normal, or
standard Cauchy in Experiment~4), so that null p-values are exactly
uniform. The target is $\alpha=0.2$,
the bandwidth $\varepsilon=\alpha\,m^{-1/2}$ with $m=\sum_im^{(i)}$, and each
point averages $100$ Monte Carlo trials. FDR is the average over trials
of $V/(R\vee1)$, and power is the average over trials of $(R-V)/(n_1\vee1)$,
where $n_1$ is the realized number of non-nulls in the trial (the ratio
is $0$ when $n_1=0$). Each experiment varies one factor of this base
configuration.

We compare twelve procedures: the centralized \textbf{Pooled BH};
\textbf{Local BH ($\alpha$)} (each node runs BH at level $\alpha$ and
the rejections are combined), \textbf{Local BH ($\alpha/N$)} (each node
runs BH at level $\alpha/N$, which controls the global FDR by a union
bound) and \textbf{Local BH ($\alpha m^{(i)}/m$)} (each node runs BH at
the size-weighted level $\alpha m^{(i)}/m$, which also controls the global
FDR since $\sum_i\alpha m^{(i)}/m=\alpha$, and which is the size-weighted
baseline of \citealp{PournaderiXiang2022Sample}); the \textbf{original greedy};
\textbf{CFGA$+$Storey} with fixed and adaptive $\varepsilon$;
\textbf{Inflated CFGA} with fixed and adaptive $\varepsilon$;
\textbf{BONuS-GA} in its global plug-in form (synthetic p-values
allocated in proportion to
$w_i=m^{(i)}\max\{\widehat r_0^{(i)},1/m^{(i)}\}$, that is, with the
Storey estimate floored at $1/m^{(i)}$, and $\tilde m=m$ synthetic
p-values in total); and the two provable implementable procedures,
\textbf{inflated e-CFGA} ($S=10$ splits, fold level $\gamma=0.75\alpha$,
adaptive $\varepsilon$ in each fold, $\eta=1/m$, with confidence $\eta/S$
in each split) and \textbf{budgeted BONuS-GA} (default budget shares
of Algorithm~\ref{alg:bonus-budget}; proportional allocation
$(\widetilde m^{(i)})_i\sim\mathrm{Multinomial}(m,(m^{(i)}/m)_i)$, so
$\widetilde m^{(i)}=m^{(i)}$ in expectation). Each main figure shows the
methods needed for that experiment's findings; Appendix~\ref{app:full-figs}
shows all twelve. In the main figures, error bars ($\pm2$ Monte Carlo SE)
are drawn on FDR curves that cross the target; on the other curves the
$\pm2$SE half-width is at most $0.010$ for the baselines and CFGA variants and at most $0.02$
for BONuS-GA and the provable procedures.

\emph{Implementation notes.} All procedures compute grid widths with
$\max\{\widehat r_0^{(i)},0.01\}$ in place of $\widehat r_0^{(i)}$, a numerical
floor that prevents excessively wide intervals when the estimated null
proportion is very small. It changes only the grid, which is computed
from the ranking half (or from the pooled sample), and none of the
guarantees depends on how the grid is chosen. CFGA is reported with the
Storey null-rate estimate throughout. The full Python source ---
the method library, a driver for every experiment and companion run, and
the stored per-experiment results with a script that redraws every figure
from them --- is provided as a supplementary code package and is publicly
available at
\begin{center}\small
\url{https://github.com/mehrdad-pournaderi/On-Large-Scale-Multiple-Testing-Over-Networks-A-Non-Asymptotic-Approach}
\end{center}

\emph{Label model.} The experiments fix each node's sample size, whereas
the CFGA-family theorems (Theorems~\ref{thm:cfga}, \ref{thm:aug}
and~\ref{thm:ecfga}) are stated under the random-label model of
Assumption~\ref{ass:model}. With fixed node sizes, the results for
Inflated CFGA and e-CFGA are therefore empirical evidence, not
consequences of the theorems; the guarantee for budgeted BONuS-GA does
cover fixed node sizes (Assumption~\ref{ass:fixed}, Remark~\ref{rem:label-free}). A companion run
draws the node sizes from $\mathrm{Multinomial}(m,q)$, with $q$
proportional to the fixed sizes, so that the CFGA theorems apply
directly. It uses the base ($m_{\mathrm{hub}}=1000$) and small-hub
($m_{\mathrm{hub}}=250$) configurations, with $100$ trials each.
FDR/power: inflated e-CFGA $0.050/0.797$ and $0.018/0.640$; Inflated
CFGA with adaptive $\varepsilon$ $0.042/0.768$ and $0.010/0.228$; budgeted
BONuS-GA $0.149/0.901$ and $0.081/0.476$; and the heuristic CFGA$+$Storey
with adaptive $\varepsilon$ $0.090/0.850$ and $0.103/0.762$ --- in each case
within Monte Carlo error of the fixed-size values.

\subsection{Why does finite-sample control fail?}
\label{sec:sim-why}

\subsubsection{Experiment 1: quiet nodes}
\label{sec:sim-localbh}

Experiment~1 fixes the two hubs and grows the number of quiet nodes from
$0$ to $600$ ($m$ from $2{,}000$ to $14{,}000$; Figure~\ref{fig:mn-break}).
\emph{Principal finding: combining the rejections of BH run separately
at each node does not control the global FDR.} Local BH$(\alpha)$'s FDR rises from $0.14$ with no quiet nodes to
$0.33$ with $600$, crossing $\alpha$ near $150$: each quiet node's false
rejections survive within-node BH and accumulate in the global FDP.
Local BH$(\alpha/N)$ is valid, but its power falls below $0.15$; the
size-weighted Local BH$(\alpha m^{(i)}/m)$ is valid and much better (FDR
$0.07\to0.01$), but its power still declines from $0.85$ to $0.58$ as
the quiet nodes dilute the hubs' share of the level. The original greedy
exceeds $\alpha$ throughout ($0.21$--$0.26$). CFGA$+$Storey
with adaptive $\varepsilon$ stays near $0.09$ with power about $0.85$, and
BONuS-GA stays at $0.18$--$0.19$ with the highest power among the
methods whose FDR stayed below the target in this experiment
($0.90$--$0.92$).

\begin{figure}[t]
\centering
\includegraphics[width=\textwidth]{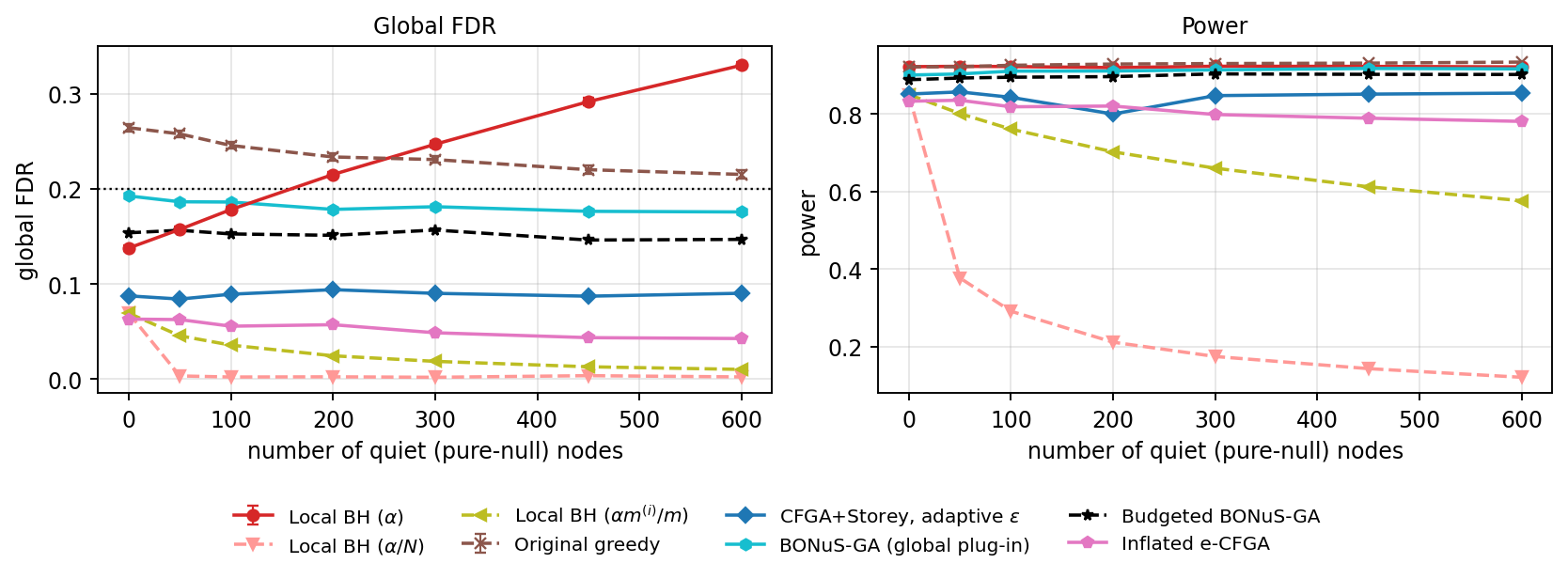}
\caption{Experiment 1: global FDR and power versus the number of quiet
pure-null nodes. The dotted line is $\alpha=0.2$.}
\label{fig:mn-break}
\end{figure}

\subsubsection{Experiment 2: signal strength}
\label{sec:sim-mu}

Experiment~2 varies the hub signal
$\mu_{\mathrm{hub}}\in\{1.5,2.0,2.5,3.0,3.5\}$ with $300$ quiet nodes
(Figure~\ref{fig:mn-mu}). \emph{Principal finding: the original greedy's
FDR exceeds $\alpha$ once the signal is strong enough to be selected, and
the excess grows with the signal} ($0.18$--$0.25$, above target for
$\mu_{\mathrm{hub}}\gtrsim2.5$), while the cross-fit methods stay at or
below $0.09$. The excess persists at the larger tested sample sizes and is
consistent with the winner's-curse bias analyzed in
Appendix~\ref{sec:gap}. The experiment also shows why the bandwidth
matters: at $\mu_{\mathrm{hub}}=2.0$ fixed-$\varepsilon$ CFGA$+$Storey has
power $0.06$ and its adaptive version $0.42$.

\begin{figure}[t]
\centering
\includegraphics[width=\textwidth]{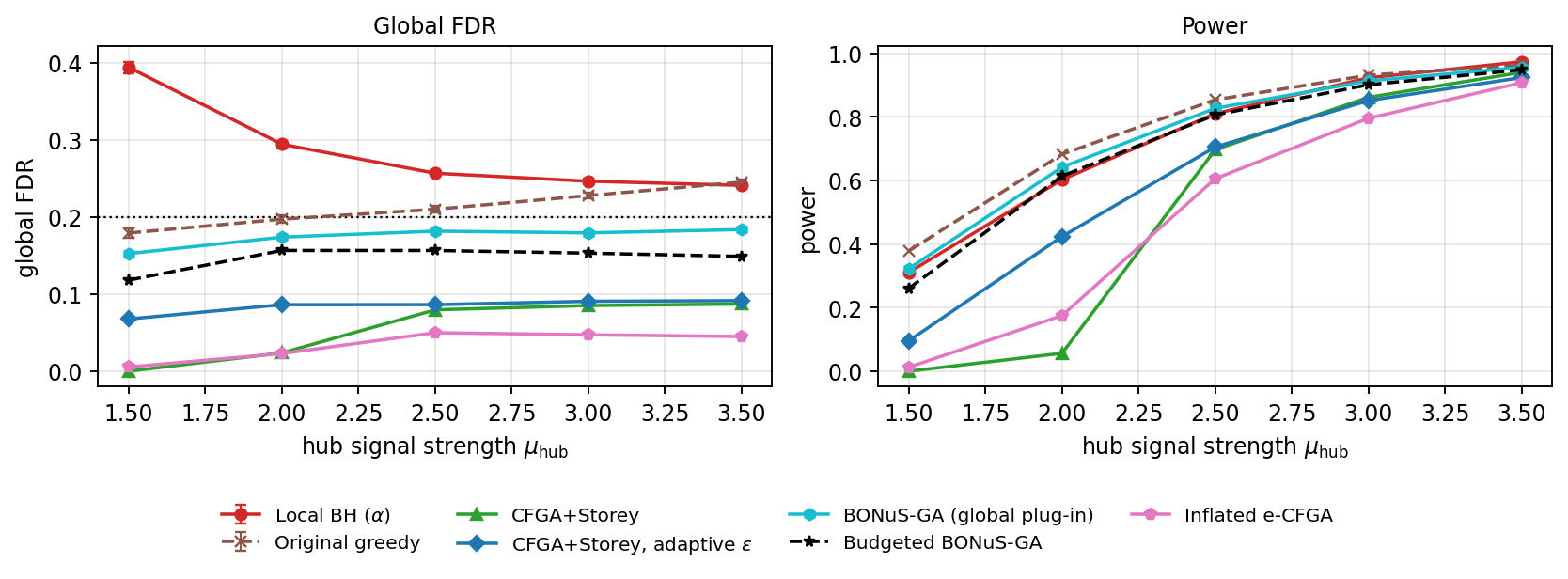}
\caption{Experiment 2: global FDR and power versus hub signal strength.}
\label{fig:mn-mu}
\end{figure}

\subsection{What does guaranteed validity cost?}
\label{sec:sim-cost}

The two provable implementable procedures compare with different
baselines. Budgeted BONuS-GA usually stays close to its heuristic
counterpart: it gives up at most $0.06$ power against global plug-in
BONuS-GA in every configuration but one --- the smallest heavy-tailed hub,
where the gap is $0.14$ --- always at lower FDR. Inflated e-CFGA is not
close to the uninflated CFGA heuristics; what it improves on is the
single-split Inflated CFGA with adaptive $\varepsilon$, whose power it
exceeds at every operating point except $\mu_{\mathrm{hub}}=2$
(Experiment~2) and the smallest heavy-tailed hub, with FDR at most $0.075$
everywhere. At that hub ($m_{\mathrm{hub}}=500$, Experiment~4) inflated
e-CFGA has zero power, while CFGA$+$Storey with adaptive $\varepsilon$
reaches $0.75$. The inflation margin has a real price for the
single-split Inflated CFGA: with adaptive $\varepsilon$ its power is
$0.08$--$0.10$ below CFGA$+$Storey at the base configuration, and with a
fixed $\varepsilon$ it often makes no rejections at all, because its margin
grows with the number of nodes. For the regimes studied we therefore
recommend running Inflated CFGA with the adaptive bandwidth. The two
experiments below show where the cost concentrates.

\subsubsection{Experiment 3: target level}
\label{sec:sim-alpha}

Experiment~3 varies $\alpha\in\{0.05,0.10,0.15,0.20\}$
(Figure~\ref{fig:mn-alpha}). \emph{Principal finding: the FDR of the provable
procedures stays below $\alpha$ at every level} (budgeted BONuS-GA from
$0.035$ to $0.148$; inflated e-CFGA from $0.010$ to $0.050$), and inflated
e-CFGA has higher power than Bonferroni-2 Inflated CFGA at every level
($0.57$ vs.\ $0.53$ at $\alpha=0.05$, $0.80$ vs.\ $0.77$ at $0.20$).
Adaptive $\varepsilon$ gives its largest gain at small $\alpha$ ($0.17\to0.64$
at $\alpha=0.05$), where the default bandwidth is farthest from optimal.

\begin{figure}[t]
\centering
\includegraphics[width=\textwidth]{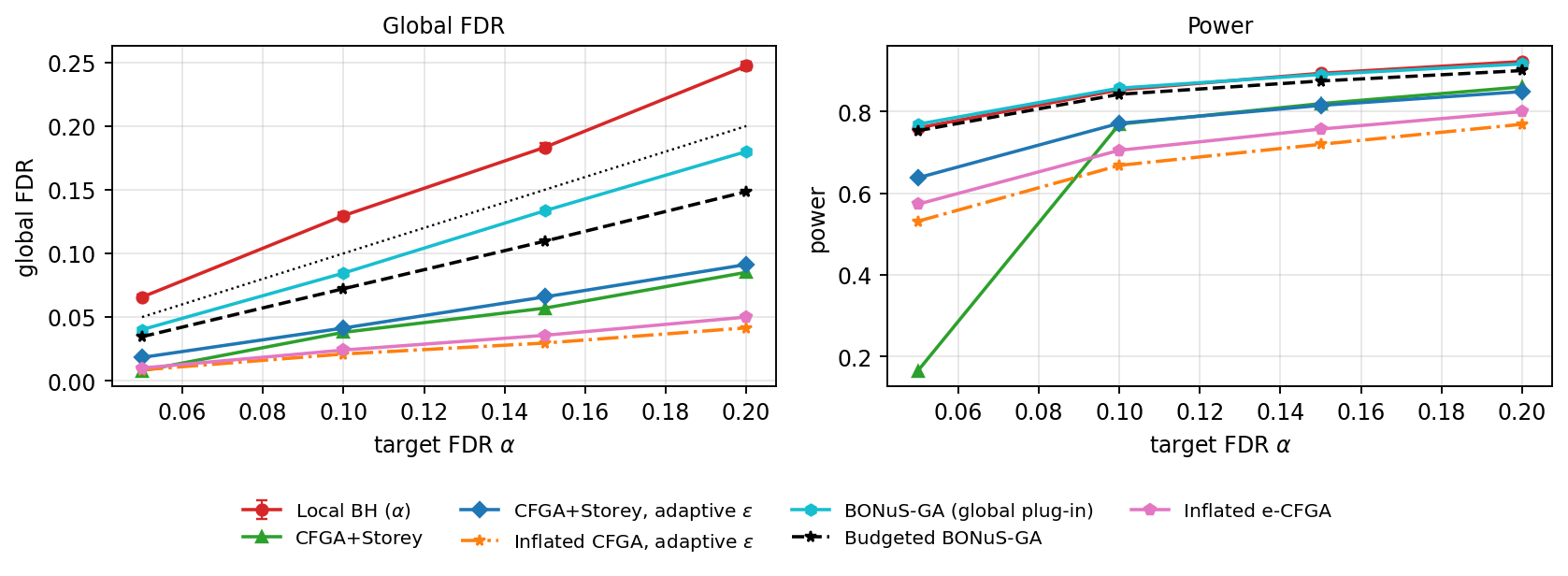}
\caption{Experiment 3: global FDR and power versus the target $\alpha$.
The dotted diagonal is $y=\alpha$.}
\label{fig:mn-alpha}
\end{figure}

\subsubsection{Experiment 4: heavy tails}
\label{sec:sim-cauchy}

Experiment~4 replaces the hub Gaussian statistics by $\mathrm{Cauchy}(0,1)$
noise with a larger shift $\mu_{\mathrm{hub}}=8$ and varies the hub size
(Figure~\ref{fig:mn-cauchy}). \emph{Principal finding: heavy tails are
where the inflated procedures pay most.} Inflated CFGA with adaptive
$\varepsilon$ reaches power $0.37$ at $m_{\mathrm{hub}}=1000$ and $0.67$ at
$2000$, and inflated e-CFGA $0.58$ and $0.78$; at the smallest hub
($500$) e-CFGA makes no discoveries and Inflated CFGA essentially none
(power $0.005$). CFGA$+$Storey with adaptive $\varepsilon$ keeps power
$0.75$--$0.82$ and BONuS-GA $0.87$--$0.90$. Pooled BH controls FDR, but
because the hub signals are diluted among the many quiet-node nulls it
makes essentially no true discoveries, while its FDR is close to the BH
bound. The FDR averages over all trials, and the FDP is nonzero only in
the few trials with rejections.

\begin{figure}[t]
\centering
\includegraphics[width=\textwidth]{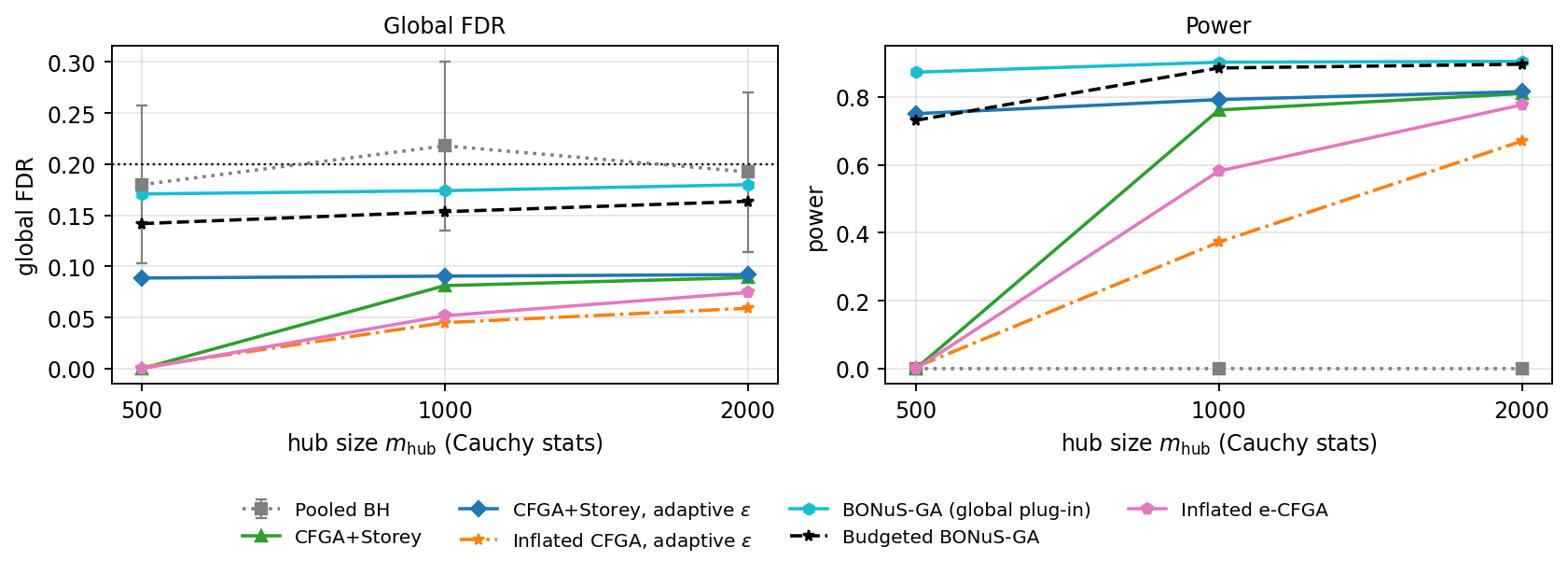}
\caption{Experiment 4: global FDR and power under $\mathrm{Cauchy}(0,1)$
hub statistics, versus hub size.}
\label{fig:mn-cauchy}
\end{figure}

\subsection{Which method in which regime?}
\label{sec:sim-regime}

\subsubsection{Experiment 5: hub size}
\label{sec:sim-hubsize}

Experiment~5 varies $m_{\mathrm{hub}}\in\{250,500,1000,2000\}$ with $300$
quiet nodes (Figure~\ref{fig:mn-n}). \emph{Principal finding: the best
method depends on the per-node sample size.} At $m_{\mathrm{hub}}=250$,
CFGA$+$Storey with adaptive $\varepsilon$ reaches power $0.78$ at FDR
$0.10$, against $0.60$ at FDR $0.13$ for BONuS-GA, while both
fixed-$\varepsilon$ Storey methods make no rejections; from
$m_{\mathrm{hub}}=500$ on, BONuS-GA has the highest power ($0.92$). Among
the provable procedures, inflated e-CFGA has power $0.64$ at FDR $0.02$
at the smallest hub, against $0.19$ for Bonferroni-2 Inflated CFGA. The
size-weighted Local BH$(\alpha m^{(i)}/m)$ baseline, valid without any
communication beyond the site sizes, has power $0.47$, $0.58$, $0.66$ and
$0.73$ at the four hub sizes (FDR at most $0.03$); budgeted BONuS-GA and
inflated e-CFGA exceed it at every hub size, and Inflated CFGA with
adaptive $\varepsilon$ at all but the smallest.

The paired ablation of Table~\ref{tab:bonus-ablation} explains why
BONuS-GA is weaker at small hubs. Each of the $200$ trials runs all
variants (arms) on the same data. Arms with the same allocation rule and
the same number $c$ of synthetic p-values per real one also use the same
synthetic p-values, because the random-number state is restored before
each arm. Comparing arms with the same allocation therefore isolates the
effect of the calibration rule (a global or a per-node ratio), and
comparing allocations isolates the effect of the allocation. With the
synthetic p-values held fixed, per-node calibration gains $0.24\pm0.07$
power over the global ratio under the Storey allocation, and
$0.49\pm0.06$ under the proportional allocation (paired $2$SE). With the
global ratio, switching from the Storey to the proportional allocation
costs $0.25\pm0.06$ power; with per-node calibration the same switch
changes power by $0.00\pm0.06$. The oracle allocation gains
$0.09\pm0.04$, and four times as many synthetic p-values gain
$0.16\pm0.08$. In this design the global ratio therefore limits power,
by an amount that depends on the allocation. Budgeted BONuS-GA, which
calibrates per node, has a finite-sample guarantee, but at this hub size
it has less power than the heuristic per-node calibration.

\begin{figure}[t]
\centering
\includegraphics[width=\textwidth]{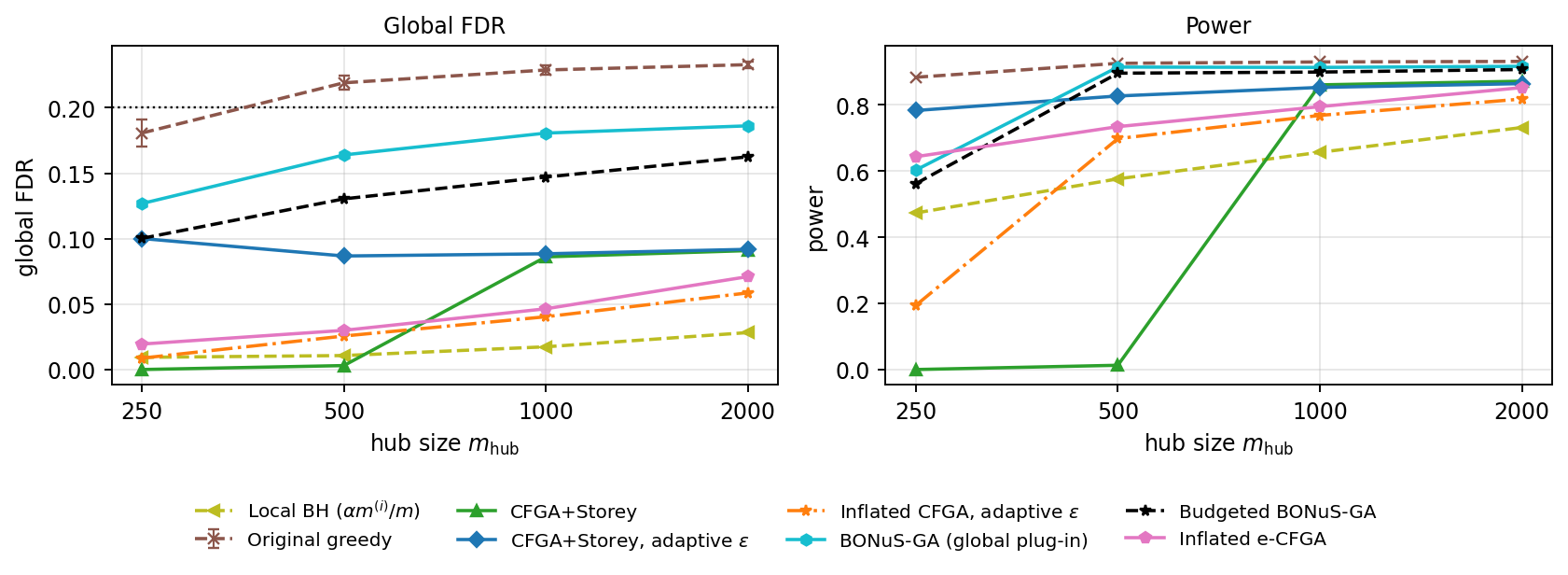}
\caption{Experiment 5: global FDR and power versus hub sample size.}
\label{fig:mn-n}
\end{figure}

\begin{table}[t]
\centering
\caption{Paired ablation at $m_{\mathrm{hub}}=250$ ($200$ trials,
$\alpha=0.2$): mean FDR and power $\pm2$ Monte Carlo SE. Rows with the same
allocation and the same number $c$ of synthetic p-values per real p-value
use the same synthetic p-values in each trial.}
\label{tab:bonus-ablation}
\small
\begin{tabular}{@{}l l c c c@{}}
\toprule
Calibration & Allocation & $c$ & FDR & Power \\
\midrule
Global ratio & Storey plug-in & 1 & $0.112\pm.013$ & $0.546\pm.063$ \\
Global ratio & oracle $r_0$ & 1 & $0.127\pm.012$ & $0.637\pm.061$ \\
Global ratio & Storey plug-in & 4 & $0.140\pm.011$ & $0.705\pm.055$ \\
Per-node ratio & Storey plug-in & 1 & $0.152\pm.010$ & $0.787\pm.046$ \\
Global ratio & proportional & 1 & $0.058\pm.011$ & $0.294\pm.058$ \\
Per-node ratio & proportional & 1 & $0.148\pm.010$ & $0.789\pm.047$ \\
Per-node budgets (Thm.~\ref{thm:bonus-pernode}) & proportional & 1 & $0.092\pm.012$ & $0.528\pm.065$ \\
\multicolumn{3}{@{}l}{CFGA$+$Storey, adaptive $\varepsilon$ (reference)} & $0.102\pm.005$ & $0.771\pm.012$ \\
\bottomrule
\end{tabular}
\end{table}

\subsubsection{Experiment 6: dependence}
\label{sec:sim-rho}

Experiment~6 makes the Gaussian noise a stationary AR(1) process at every
node, $Z_k=\rho Z_{k-1}+\sqrt{1-\rho^2}\,\epsilon_k$ with i.i.d.\ standard
normal $\epsilon_k$, and adds the independent signal shifts $B_k\mu_k$ as
before, varying the AR(1) noise correlation
$\rho\in\{0,0.2,0.4,0.6,0.8\}$ (Figure~\ref{fig:mn-rho}). The statistics
themselves are therefore not AR(1) with correlation $\rho$ at the hubs:
there the adjacent-statistic correlation is $\rho/\operatorname{Var}(X_k)
=\rho/2.915$ ($0.274$ at $\rho=0.8$), whereas at the quiet nodes it is
$\rho$. \emph{Principal finding, empirical and specific to these AR(1)
experiments: an un-inflated data-driven bandwidth can over-reject under
strong dependence.} CFGA$+$Storey with adaptive $\varepsilon$ is controlled
up to $\rho\approx0.6$ but reaches FDR $0.32$ at $\rho=0.8$: the
cross-validated $\varepsilon$ takes advantage of the Storey estimate being
too small under dependence (Section~\ref{sec:eps-adaptive}). In a
companion run that confines the correlation to the hubs (quiet nulls
i.i.d., same seeds) its FDR stays at $0.087$--$0.094$ ($\pm0.005$), so in
these experiments the over-rejection is driven by the dependent
quiet-node data. The other methods stay flat in these experiments:
fixed-$\varepsilon$ CFGA$+$Storey at $0.085$--$0.087$, Inflated CFGA with
adaptive $\varepsilon$ near $0.04$, BONuS-GA at $0.17$--$0.19$, budgeted
BONuS-GA at $0.15$--$0.17$, and inflated e-CFGA nearly flat at
$0.048$--$0.049$ (power $0.79$--$0.80$) despite its per-fold data-driven
$\varepsilon$. In these AR(1) experiments the inflation margin is large
enough to absorb the effect of dependence; it is derived under
independence, however, and is not designed to correct for dependence.

\begin{figure}[t]
\centering
\includegraphics[width=\textwidth]{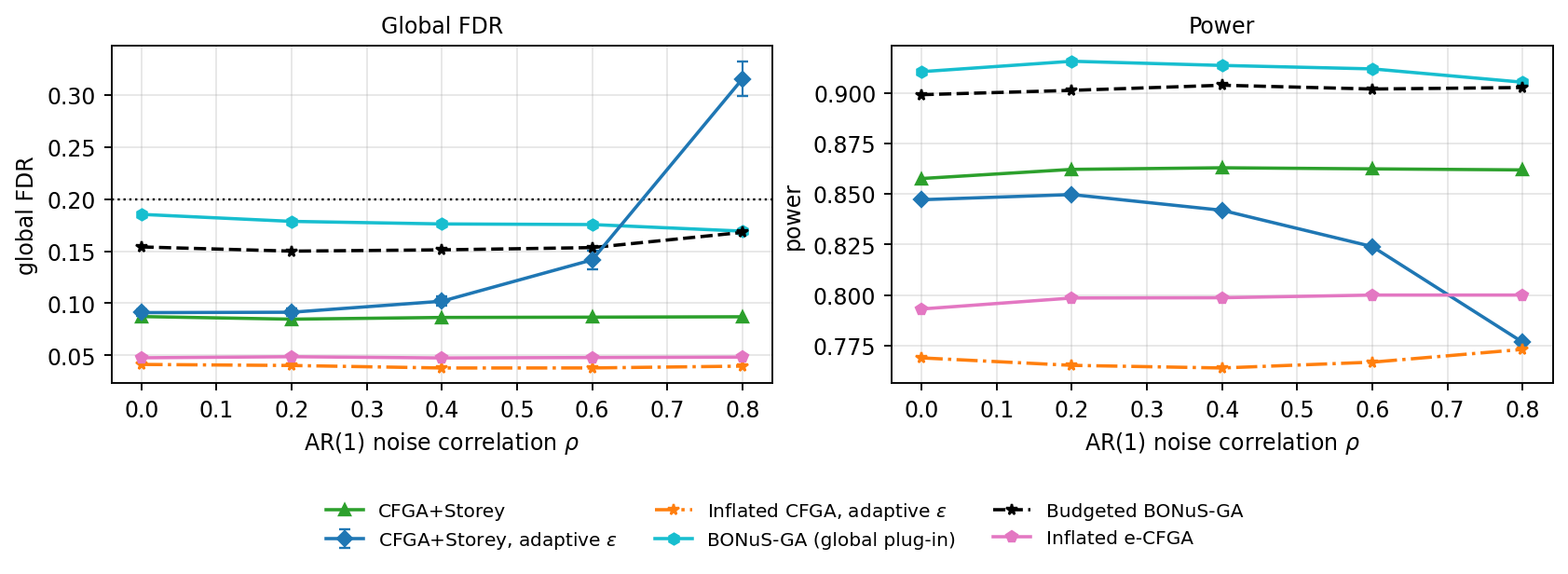}
\caption{Experiment 6: global FDR and power versus the AR(1) noise
correlation $\rho$, applied at every node.}
\label{fig:mn-rho}
\end{figure}

\subsubsection{Summary}

Table~\ref{tab:guide} condenses the experiments into a practitioner's
guide.

\begin{table}[t]
\centering
\caption{Which method when. Representative FDR/power at $\alpha=0.2$. The
proven guarantees are $\FDR\le\alpha+1/m$ for inflated e-CFGA (under
random labels) and $\FDR\le\alpha$ for budgeted BONuS-GA (under fixed or
random labels).}
\label{tab:guide}
\small
\begin{tabular}{@{}p{0.40\textwidth} p{0.56\textwidth}@{}}
\toprule
Scenario & Recommendation \\
\midrule
Moderate-to-large hubs, best power & BONuS-GA, global plug-in (heuristic;
$0.18/0.91$ at $m_{\mathrm{hub}}{=}1000$) \\[2pt]
Moderate-to-large hubs, guarantee required & Budgeted BONuS-GA
(Thm.~\ref{thm:bonus-pernode}; $0.15/0.90$) \\[2pt]
Small Gaussian hubs ($m^{(i)}\lesssim$ a few hundred) & CFGA$+$Storey with adaptive
$\varepsilon$ (heuristic; $0.10/0.78$ at $m_{\mathrm{hub}}{=}250$); with a
guarantee: inflated e-CFGA ($0.02/0.64$) \\[2pt]
Heavy-tailed hubs, guarantee required & Budgeted BONuS-GA ($0.14/0.73$
at $m_{\mathrm{hub}}{=}500$, where inflated e-CFGA has zero power) \\[2pt]
Strong within-node dependence suspected & Empirically, in our AR(1)
experiments: the inflated adaptive procedures (flat FDR up to $\rho=0.8$;
no dependence guarantee); avoid un-inflated adaptive $\varepsilon$ \\[2pt]
Reduced sensitivity to the split & e-CFGA (Thm.~\ref{thm:ecfga}),
$S$ splits averaged \\
\bottomrule
\end{tabular}
\end{table}

\FloatBarrier
\section{Discussion}
\label{sec:discuss}

Budgeted BONuS-GA (Algorithm~\ref{alg:bonus-budget}) gives finite-sample
FDR control for greedy aggregation over a network with no oracle input,
for fixed site assignments, using every p-value for both selection and
inference, in two communication rounds and within the original
$O(\sqrt m\log m)$ budget whenever $N=O(\sqrt m)$. Its proof is a
per-site urn argument that must hold conditionally on the pooled data,
because the budget shares are computed from them. The price of the
per-site budgets is visible at sites with few p-values, where the
heuristic global plug-in variant, which has no guarantee, is more
powerful. Cross-fitting (CFGA) is the sample-splitting alternative: it is
a thin wrapper around the existing procedure, its proof reduces on the
lifted (site, p-value) space to a supermartingale lemma for adaptive BH,
and the inflated variant extends the guarantee to estimated null rates at
a slack $\eta=1/m$, at a conservatism cost
$\delta=O\bigl(\sqrt{\log(2Nm)/m_b}\bigr)$ that is real in networks with
many nodes, so we pair it with the adaptive bandwidth in practice; its
guarantees require the random-site model. Empirically the two families
are complementary: BONuS-GA has the highest power at moderate-to-large
per-node sample sizes, while CFGA with adaptive $\varepsilon$ is the
higher-power choice when individual hubs are small
(Section~\ref{sec:sim}, Experiment~5).

Four directions remain open.

\paragraph{Within-node dependence.} Both our
finite-sample proofs and the asymptotic analysis of
\citet{PournaderiXiang2024} require within-node independence.
Classical BH controls FDR under positive
regression dependence \citep{BY2001}; extending CFGA's and BONuS-GA's
guarantees to PRDS or weakly-dependent sequences would close that gap.
Empirically the fixed-$\varepsilon$ methods are remarkably robust to AR(1)
correlation up to $\rho=0.8$ (Experiment~6), but a proof under dependence
is open, and the inflated margin's resilience there should be read as an
empirical observation rather than a dependence correction.
\paragraph{Offsets.} The ``$+1$'' in the CFGA stopping rule can be dropped
(Remark~\ref{rem:no-offset}); the offsets in the BONuS-GA calibration, whose
calibrator is a synthetic count, remain, and whether they can be reduced
is open.

\paragraph{Sharper e-CFGA.} The e-value aggregation of
Section~\ref{sec:ecfga} removes the Bonferroni-2 factor at a fixed fold
level $\gamma$; characterizing the power-optimal $\gamma$ (our
boundary calculation of Appendix~\ref{app:ecfga-details} shows only that $\gamma=\alpha$
leaves no slack in the unanimous regime; $\gamma=0.75\alpha$ is an
empirical default)
and extending the e-value route to the small-hub regime, where it
currently trails CFGA$+$adaptive $\varepsilon$, are natural refinements.

\paragraph{Closing the remaining BONuS-GA gap.} Theorem~\ref{thm:bonus-pernode} now gives a fully implementable,
split-free procedure with exact finite-sample FDR control and no oracle
input, and at moderate-to-large hubs it essentially matches the
global-ratio heuristic. What remains open is the small-hub regime: the
\emph{scalar-sum} per-node rule of Remark~\ref{rem:bonus} (which lets
active nodes pool their slack and a single SeqStep$+$ offset, reaching
power $0.787$ at $m_{\mathrm{hub}}=250$ against the budgeted rule's
$0.528$) still lacks a proof; a finite-sample bound for it --- or a
budget-sharing scheme that provably recovers the pooled slack --- would
absorb both regimes under one analysis, as the ablation of
Table~\ref{tab:bonus-ablation} suggests.

\newpage
\appendix

\section{BONuS-GA: variants and proofs}
\label{app:bonus-variants}

This appendix gives the global form of BONuS-GA, whose oracle version is
the procedure analysed by Theorem~\ref{thm:bonus}, and the proofs of
Theorems~\ref{thm:bonus} and~\ref{thm:bonus-pernode}.

\subsection{Global BONuS-GA: oracle and plug-in allocations}

Like budgeted BONuS-GA, the global form described here achieves the
decoupling by masking rather than by splitting. CFGA and Inflated CFGA both buy the independence of selection from inference by
\emph{splitting} the data: the inference half pays for the guarantee by
contributing only $m/2$ p-values to the rejection count. We now describe a
variant that achieves the same decoupling by \emph{masking}, following the
bag-of-null-statistics (BONuS) construction of \citet{Yang2021BONuS}, and uses
all $m$ p-values for both selection and inference. The idea, due to
\citet{BarberCandes2015} and adapted to multiple testing by
\citet{Yang2021BONuS}, is to generate a bag of \emph{synthetic} null
statistics --- ``counting knockoffs'' --- known to be null by construction,
hide their real-vs-synthetic identities from the procedure during region
selection, and then use the synthetic counts as a calibrator of the unknown
null count inside the selected region: $\widetilde N(\Gamma)$ proxies the
unobserved $V(\Gamma)$. The exactly-uniform null (Assumption~\ref{ass:model})
makes this the most favourable possible setting for BONuS: the null
distribution is known in closed form, so the ``bag'' is \emph{generated}
rather than estimated.

\paragraph{Three tiers, two with proofs.} We present BONuS-GA in three
tiers. The \emph{oracle} tier uses per-node bag-size weights equal to
the true null masses $\pi^{(i)}=q^{(i)}r_0^{(i)}$; for it we give a
finite-sample FDR proof (Theorem~\ref{thm:bonus}). The
\emph{global plug-in} (global-ratio) tier replaces the unknown $r_0^{(i)}$ by
Storey's plug-in $\widehat r_0^{(i)}$; it does \emph{not} inherit a
finite-sample proof from Theorem~\ref{thm:bonus} (see
Remark~\ref{rem:oracle-only}) and is the highest-power practical variant
in the simulations of Section~\ref{sec:sim}. The \emph{per-node budgeted}
tier (Algorithm~\ref{alg:bonus-budget} of Section~\ref{sec:bonus-main}) replaces the global tail ratio by
per-node ratios under per-node budget constraints; it requires no oracle
input at all and carries a finite-sample proof
(Theorem~\ref{thm:bonus-pernode}), at a power cost that is negligible at
moderate-to-large hubs and visible at small ones. A proof for the
global-ratio tier itself remains open.

\begin{algo}[BONuS-GA]\label{alg:bonus}
Fix $\varepsilon>0$, target $\alpha\in(0,1)$, Storey threshold $\lambda\in(0,1)$,
and a per-node weighting rule (oracle or global plug-in;
Appendix~\ref{app:bonus-variants}).
\begin{enumerate}[leftmargin=*]
\item \textbf{Per-node weight $w^{(i)}\ge0$.}
\begin{itemize}
\item \emph{Oracle (Theorem~\ref{thm:bonus}):}
$w^{(i)}=\pi^{(i)}=q^{(i)}r_0^{(i)}$, using the true per-node null masses.
In this tier the grid width of Step~3 is likewise set from the oracle
quantities, $\widehat L^{(i)}=\varepsilon/(q^{(i)}r_0^{(i)})$, so the
candidate intervals are deterministic and the basic masking
$\sigma$-field \eqref{eq:F-mask-def} below suffices. Edge cases: a node
with $\pi^{(i)}=0$ hosts no nulls and receives no synthetic points, and
its grid width is set to $\varepsilon/q^{(i)}$ (any fixed grid is
admissible there, since rejections at such a node cannot be false); if
$\pi_0=\sum_i\pi^{(i)}=0$ there are no nulls at all, every procedure has
$\FDR=0$, and by convention the bag is then allocated proportionally to
$q^{(i)}$.
\item \emph{Global plug-in (recommended heuristic).} The synthetic bag in Step~2 is
drawn with allocation weights
$w^{(i)}=m^{(i)}\,\widehat r_{0,\mathrm{init}}^{(i)}$, where
$\widehat r_{0,\mathrm{init}}^{(i)}$ is the real-only Storey estimator from
Section~\ref{sec:greedy}; once the bag is drawn, the \emph{pooled} Storey
estimator
\[
\widehat r_0^{(i)}=\min\Bigl\{1,\
\frac{1+N^{(i)}(A)+\widetilde N^{(i)}(A)}
{(1-\lambda)\,(m^{(i)}+\widetilde m^{(i)})}\Bigr\},
\quad A=(\lambda,1],
\]
is used, with the pooled share
$\widehat q^{(i)}=(m^{(i)}+\widetilde m^{(i)})/(m+\tilde m)$, to set the grid
width $\widehat L^{(i)}=\varepsilon/(\widehat q^{(i)}\widehat r_0^{(i)})$
in Step~3 (the same clipped $+1$ estimator as in the other procedures,
applied to the pooled sample). The pooled form makes $\widehat L^{(i)}$ (and hence
$\{\widehat\Gamma_M\}$) $\F_{\mathrm{mask}}$-measurable via the extended
definition \eqref{eq:F-mask-def-ext} below, which is what step~(b) of the
proof of Theorem~\ref{thm:bonus} requires. The two-stage construction
breaks the apparent circularity between bag allocation and grid spacing.
\end{itemize}

\item \textbf{Synthetic bag of nulls.} Fix a synthetic budget
$\tilde m\in\N$ (typically $\tilde m=c\,m$ for some scale $c>0$). Draw
$\tilde m$ i.i.d.\ synthetic null lifted hypotheses
$\widetilde X_\ell^{(\mathrm{syn})}=(\widetilde I_\ell,\widetilde P_\ell)$,
$\ell=1,\dots,\tilde m$, with $\widetilde I_\ell$ drawn from the categorical
distribution
$\PP(\widetilde I_\ell=i)=w^{(i)}/\sum_j w^{(j)}$
(\emph{i.i.d.\ multinomial labels}) and $\widetilde P_\ell$ an independent
$\mathrm{Uniform}[0,1]$. The realized per-node count
$\tilde m^{(i)}:=\#\{\ell:\widetilde I_\ell=i\}$ is therefore
multinomial$(\tilde m,\{w^{(i)}/\sum_j w^{(j)}\})$ rather than deterministic;
this is what makes the synthetic bag exchangeable with the real-null lift in
the proof of Theorem~\ref{thm:bonus}.

\item \textbf{Masked greedy.} Pool the real and synthetic p-values at each node
and rank intervals on the Storey grid $[0,\lambda]$ by their \emph{pooled}
counts, producing the nested family
$\widehat\Gamma_1\subset\cdots\subset\widehat\Gamma_{M_{\max}}$ measurable with
respect to the masking $\sigma$-field $\mathcal{F}_{\mathrm{mask}}$ (pooled
counts do not reveal real-vs-synthetic identities).

\item \textbf{Counting-knockoff FDP.} With correction set $A=(\lambda,1]$ and
the Storey-BONuS estimator
\begin{equation}\label{eq:bonus-fdp}
\widehat\FDP(M)=\frac{N(A)+1}{\widetilde N(A)}\cdot
\frac{\widetilde N(\widehat\Gamma_M)+1}{R(M)\vee 1},
\end{equation}
where $N(\cdot)=\sum_iN^{(i)}(\cdot)$, $\widetilde
N(\cdot)=\sum_i\widetilde N^{(i)}(\cdot)$, and $R(M)=N(\widehat\Gamma_M)$;
if $\widetilde N(A)=0$ the estimator is $+\infty$ by convention and the
procedure makes \emph{no} rejections (substituting a positive denominator
here can severely violate FDR when the bag is tiny).

\item \textbf{Reject.} Set $\widehat M=\max\{M:\widehat\FDP(M)\le\alpha\}$
($\widehat M=0$ if empty) and reject the real p-values lying in
$\widehat\Gamma_{\widehat M}$.
\end{enumerate}
\end{algo}

\paragraph{Reading the algorithm.} Step~1 supplies the per-node weights from
which the synthetic bag is sized in Step~2; the two weight choices (oracle
and global plug-in) produce the two procedure variants whose FDR scope is
discussed in Appendix~\ref{app:bonus-variants}.

\emph{What the masking $\sigma$-field $\F_{\mathrm{mask}}$ contains.} At
runtime each node knows which of its p-values are real and which are
synthetic (it generated the latter), so nothing is literally hidden at the
node level. The masking framework of \citet{Yang2021BONuS} is a
\emph{proof-side restriction} on how the algorithm is allowed to use that
information: the family $\{\widehat\Gamma_M\}$ must be constructed using only
the pooled per-interval counts
$N^{(i)}(\widehat I_j^{(i)}) + \widetilde N^{(i)}(\widehat I_j^{(i)})$,
without exploiting the real-vs-synthetic split of any pooled count inside
the candidate region. Formally,
\begin{equation}\label{eq:F-mask-def}
\F_{\mathrm{mask}}:=\sigma\!\Bigl(\bigl\{N^{(i)}(\widehat I_j^{(i)})+
\widetilde N^{(i)}(\widehat I_j^{(i)}):(i,j)\in\Aset\bigr\}\Bigr) .
\end{equation}
For the global plug-in variant (and for budgeted BONuS-GA in
Section~\ref{sec:bonus-main}, which uses the same pooled grid) the grid
itself is data-driven: it uses the
pooled node total $n_i:=m^{(i)}+\widetilde m^{(i)}$ and the pooled tail count
$T_i:=N^{(i)}(A)+\widetilde N^{(i)}(A)$ on $A=(\lambda,1]$, through
$\widehat q^{(i)}\widehat r_0^{(i)}=\min\{n_i,(1+T_i)/(1-\lambda)\}/(m+\tilde m)$,
and these are not determined by the candidate-cell counts (part of the
head can lie outside the candidate cells). We therefore work with the
enrichment generated by each node's \emph{unordered} pooled sample,
\begin{equation*}
\F_{\mathrm{mask}}^{\mathrm{ext}}:=\sigma\Bigl(\bigl\{\,\{\!\{ P:\
P\text{ a real or synthetic p-value at node }i\}\!\}\,\bigr\}_{i=1}^N\Bigr),
\end{equation*}
where $\{\!\{\cdot\}\!\}$ denotes a multiset: it records where the
pooled points lie at each node but not which of them are real. It
determines $n_i$, $T_i$, the grid, and every pooled cell count, so
$\F_{\mathrm{mask}}\subseteq\F_{\mathrm{mask}}^{\mathrm{ext}}$, and it is
contained in the conditioning $\sigma$-field of Appendix~\ref{app:bonus}.
The pooled counts are by construction $\F_{\mathrm{mask}}^{\mathrm{ext}}$-measurable;
the per-interval splits into $N^{(i)}(\widehat I_j^{(i)})$ and
$\widetilde N^{(i)}(\widehat I_j^{(i)})$ inside the candidate region are
not, which is what the masking argument protects. The FDR proof then unmasks
identities one at a time along a backward filtration and bounds
$\E[V/(R\vee 1)]$ by exchangeability of the masked nulls.

Step~3 obeys this restriction: it ranks intervals by their pooled counts
only. Step~4 is the standard Storey-BONuS estimator from
\citet{Yang2021BONuS}: the synthetic count $\widetilde N(\widehat\Gamma_M)$
is a counting-knockoff proxy for the unknown false-rejection count, the
global tail ratio $N(A)/\widetilde N(A)$ supplies the null-to-synthetic
scaling, and the two ``$+1$''s are the standard SeqStep+ offsets that
guard against the empty regime. Step~5 is the step-up cut. The stopping
rule of Step~4--5 is allowed to use the runtime real-vs-synthetic split
because by that point the family $\{\widehat\Gamma_M\}$ is already fixed by
Step~3 and the backward-martingale argument handles the unmasking.

\subsection{Finite-sample FDR control under the oracle allocation}
Exactness asks the synthetic nulls to be exchangeable with the real nulls in the
lifted space. A real null lands at node $i$ with probability $\pi^{(i)}/\pi_0$,
where $\pi_0=\sum_j\pi^{(j)}$; matching this is the \emph{oracle allocation}.

\begin{theorem}[FDR control for oracle BONuS-GA]\label{thm:bonus}
Under Assumption~\ref{ass:model}, suppose the synthetic bag is drawn from the
true lifted null law---each synthetic null is placed at node $i$ with
probability $\pi^{(i)}/\pi_0$ and given a $\mathrm{Uniform}[0,1]$ p-value. Then
for every $m\ge1$, every $\tilde m\ge1$ and every $\lambda\in(0,1)$,
BONuS-GA satisfies $\FDR\le\alpha$, with no sample splitting. The proof is in
Appendix~\ref{app:bonus}.
\end{theorem}

\begin{remark}[Scope: oracle weighting only]\label{rem:oracle-only}
Theorem~\ref{thm:bonus} controls FDR under the \emph{oracle} weighting
$w^{(i)}=\pi^{(i)}=q^{(i)}r_0^{(i)}$, which requires knowledge of the
unobserved per-node null proportions $r_0^{(i)}$. The global plug-in
allocation $w^{(i)}=m^{(i)}\widehat r_0^{(i)}$ used in
Appendix~\ref{app:bonus-variants} is a plug-in approximation and does
\emph{not} inherit a finite-sample bound from Theorem~\ref{thm:bonus};
it is in the same relationship to oracle BONuS-GA as Storey-CFGA is to
Theorem~\ref{thm:cfga}. The per-node \emph{budgeted} variant of
Section~\ref{sec:bonus-main} closes this gap by a different route ---
per-node calibration removes the allocation question entirely --- while a
proof for the global-ratio plug-in itself remains open.
\end{remark}

The argument is a counting-knockoff backward martingale in the spirit of
\citet{BarberCandes2015,Yang2021BONuS}, proved self-contained in
Appendix~\ref{app:bonus}: under the oracle allocation the real nulls and
the $\tilde m$ synthetic points are i.i.d.\ from one common lifted law,
so conditional on their \emph{locations} the real-vs-synthetic identity
assignment is a single uniform urn; the pooled-count greedy is
identity-blind; unmasking the urn backward along the nested family makes
$V(\widehat\Gamma_M)/(\widetilde N(\widehat\Gamma_M)+1)$ a backward
supermartingale (an exact martingale while the synthetic count is
positive), and a two-cell urn identity (Lemma~\ref{lem:urn}) anchors it
against the tail calibrator. The martingale plays the role of
Lemma~\ref{lem:lifted}, with the \emph{generated}
synthetic count replacing the \emph{predictable} null mass $m_{b'}p_M$---which
is precisely why no held-out half is required.


\subsection{Oracle and global plug-in weights}\label{sec:bonus-tiers}

\paragraph{Oracle.} With $w^{(i)}=\pi^{(i)}=q^{(i)}r_0^{(i)}$ in
Step~1 of Algorithm~\ref{alg:bonus}, the synthetic bag is drawn from the true
lifted null law and Theorem~\ref{thm:bonus} above gives $\FDR\le\alpha$ at
every finite $m$.

\paragraph{Global plug-in (recommended heuristic).} The oracle weights are
unknown in practice. Replacing them by the Storey plug-in
$\widehat\pi^{(i)}=(m^{(i)}/m)\,\widehat r_0^{(i)}$ produces the
global plug-in weight $w^{(i)}=m^{(i)}\,\widehat r_0^{(i)}$ in Step~1. This is
exactly Algorithm~\ref{alg:bonus} with plug-in estimates of the oracle
quantities --- in the same sense in which practical CFGA$+$Storey relates to
the oracle CFGA of Theorem~\ref{thm:cfga}: it inherits no finite-sample
guarantee on its own. Its gap to the oracle allocation is not negligible
at small hubs: at $m_{\mathrm{hub}}=250$ the oracle allocation gains
$0.09\pm0.04$ power in the paired ablation of
Table~\ref{tab:bonus-ablation}. We adopt it as the recommended
heuristic at moderate-to-large hubs. A rigorous
finite-sample guarantee for this \emph{global plug-in ratio} variant
itself --- whether through a per-node calibration that removes the
allocation dependence (Remark~\ref{rem:bonus}) or through an inflation
argument that absorbs the Storey estimation error --- is an interesting
open problem; the budgeted tier of Theorem~\ref{thm:bonus-pernode} below
sidesteps it, obtaining a fully implementable guarantee through per-node
budgets rather than a global ratio.

\begin{remark}[Per-node calibration]\label{rem:bonus}
Replacing the single global ratio by per-node ratios
$N^{(i)}(A)/\widetilde N^{(i)}(A)$ in \eqref{eq:bonus-fdp} removes the allocation
dependence entirely: $\sum_i\frac{N^{(i)}(A)}{\widetilde
N^{(i)}(A)}\widetilde N^{(i)}(\widehat\Gamma_M)$ targets
$\sum_in_0^{(i)}\nu(\widehat\Gamma_M^{(i)})$ for \emph{any} bag, at the cost of
noisier estimates at the smallest nodes. This scalar-sum rule remains a
heuristic, but Section~\ref{sec:bonus-main} shows that its
\emph{budgeted} form---per-node SeqStep$+$ constraints with
masking-measurable budget shares---admits a finite-sample proof with
\emph{no} oracle knowledge. BONuS-GA is also \emph{randomized}
through the bag; averaging several bags or aggregating their $e$-values
\citet{RenBarber2024evalues} derandomize it; a finite-sample analysis of
the derandomized variants is left to future work.
\end{remark}

\subsection{Proofs}
\label{app:bonus}

Both proofs in this appendix are self-contained specializations of one
mechanism --- a backward supermartingale (an exact martingale away from
empty synthetic counts) over the progressive unmasking
of real-vs-synthetic identities, anchored by an elementary two-cell urn
identity --- applied to a \emph{single global} exchangeable pool for
Theorem~\ref{thm:bonus} and to $N$ \emph{per-node} pools for
Theorem~\ref{thm:bonus-pernode}. (The mechanism is the counting-knockoff
principle of \citet{BarberCandes2015,Yang2021BONuS}; we prove what we need
directly rather than importing their statements.) We isolate the urn
identity first, because it must hold \emph{conditionally on the masked
p-value locations}: in the budgeted variant the budget shares $\alpha_i$
are computed from pooled counts, so the expectation bound has to survive
that conditioning, under which the tail/head identity split is
hypergeometric rather than binomial.

\begin{lemma}[Two-cell urn identity]\label{lem:urn}
Partition an urn of $R$ red and $W$ white balls uniformly at random into a
cell $A$ of size $C_A$ and a cell $\Gamma$ of size $R+W-C_A$, and write
$V_A,W_A$ (resp.\ $V_\Gamma,W_\Gamma$) for the red/white counts in $A$
(resp.\ $\Gamma$). Then, with $v_{\min}:=\max(0,\,C_A-W)$,
\[
\E\!\left[\frac{W_A\,V_\Gamma}{(V_A+1)(W_\Gamma+1)}\right]
\;=\;1-\PP(V_A=v_{\min})\;\le\;1 .
\]
\end{lemma}

\begin{proof}
With $p_v:=\binom{R}{v}\binom{W}{C_A-v}/\binom{R+W}{C_A}$ the
hypergeometric pmf, $W_A=C_A-v$, $V_\Gamma=R-v$ and
$W_\Gamma=W-C_A+v$, so the summand weight is
\[
\binom{R}{v}\binom{W}{C_A-v}\,\frac{(C_A-v)(R-v)}{(v+1)(W-C_A+v+1)}
=\binom{R}{v+1}\binom{W}{C_A-v-1},
\]
using $\binom{R}{v+1}=\binom{R}{v}\frac{R-v}{v+1}$ and
$\binom{W}{C_A-v-1}=\binom{W}{C_A-v}\frac{C_A-v}{W-C_A+v+1}$. Summing
over the support $v\ge v_{\min}$ and substituting $j=v+1$,
\[
\E[\cdot]=\frac{1}{\binom{R+W}{C_A}}\sum_{j\ge v_{\min}+1}
\binom{R}{j}\binom{W}{C_A-j}
=1-\PP(V_A=v_{\min})
\]
by Vandermonde's identity, since the omitted terms $j\le v_{\min}$ reduce
to the single nonzero term $j=v_{\min}$.
\end{proof}

\begin{proof}[Proof of Theorem~\ref{thm:bonus}]
\emph{(a) One global urn.} Lift each real hypothesis to
$\widetilde X_k=(I_k,P_k)$. Under Assumption~\ref{ass:model} a real null
has lifted law $G_0$ --- the distribution on $\{1,\dots,N\}\times[0,1]$
that places mass $\pi^{(i)}/\pi_0$ on label $i$ and is
$\mathrm{Uniform}[0,1]$ in the p-value --- while a real non-null has an
arbitrary law on the lifted space. Let
$H_0:=\{k:k\text{ is a real null}\}$ and $n_0:=|H_0|$ (not to be confused
with the per-node $\Hzero^{(i)}$ of Section~\ref{sec:setting}). By
construction the oracle synthetic bag (Step~2 of
Algorithm~\ref{alg:bonus} with $w^{(i)}=\pi^{(i)}$) is i.i.d.\ from
$G_0$, so the $n_0+\tilde m$ masked lifted points (real nulls $+$
synthetics) are i.i.d.\ from the \emph{common} law $G_0$ given $H_0$ and
the non-null points. Condition throughout on the identity-blind
$\sigma$-field
\[
\mathcal M:=\sigma\bigl(H_0,\ \{\widetilde X_k\}_{k\notin H_0},\
\text{all masked lifted locations}\bigr):
\]
given $\mathcal M$, the identity assignment is a single uniform urn ---
$n_0$ ``real'' labels among the $n_0+\tilde m$ masked locations. The
nested family $\{\widehat\Gamma_M\}$ is a function of the pooled
per-interval counts, which are location functions, hence
$\mathcal M$-measurable and identity-blind (under the oracle tier the
grid $\widehat L^{(i)}=\varepsilon/(q^{(i)}r_0^{(i)})$ is deterministic,
so \eqref{eq:F-mask-def} suffices; with the pooled-Storey grid one uses
\eqref{eq:F-mask-def-ext}, likewise identity-blind and
$\mathcal M$-measurable).

\emph{(b) Algebraic identity.} With the lifted tail
$\widetilde A:=\bigcup_i\{i\}\times(\lambda,1]$ (we keep writing $A$),
plug \eqref{eq:bonus-fdp} into
\[
\frac{V(\widehat M)}{R(\widehat M)\vee 1}
=\widehat\FDP(\widehat M)\cdot
\underbrace{\frac{V(\widehat M)}{\widetilde N(\widehat\Gamma_{\widehat M})+1}\cdot
\frac{\widetilde N(A)}{N(A)+1}}_{=:Q}.
\]
At stopping $\widehat\FDP(\widehat M)\le\alpha$, so
$\FDP\le\alpha\,Q$ on $\{\widehat M\ge1\}$, and $\FDP=0$ on
$\{\widehat M=0\}$; hence $\FDR\le\alpha\,\E[Q]$ with the convention
$Q:=0$ on $\{\widehat M=0\}$.

\emph{(c) Backward martingale and the urn bound.} Define the backward
filtration $(\F_M)$ containing $\mathcal M$, the tail identity split
$(V(A),\widetilde N(A))$, and the region splits
$(V(\widehat\Gamma_{M'}),\widetilde N(\widehat\Gamma_{M'}))$ for
$M'\ge M$. The step-up index $\widehat M$ is a function of the
$\F_M$-measurable quantities $\widetilde N(\widehat\Gamma_{M'})$,
$R(M')=N(\widehat\Gamma_{M'})$ (the non-null contribution is
$\mathcal M$-known and the null contribution is the revealed
$V(M')$), and the tail counts --- hence a backward stopping time of
$(\F_M)$. Conditional on $\F_{M+1}$ the identities inside
$\widehat\Gamma_{M+1}$ remain a uniform urn, so peeling one
$\mathcal M$-measurable masked location at a time shows that
$U_M:=V(\widehat\Gamma_M)/(\widetilde N(\widehat\Gamma_M)+1)$ is a
backward supermartingale, and an exact martingale whenever
$\widetilde N\ge1$: if the peeled point is real with
conditional probability $q=V/(V+\widetilde N)$,
\[
q\,\frac{V-1}{\widetilde N+1}+(1-q)\,\frac{V}{\widetilde N}
=\frac{V\bigl[(V-1)+(\widetilde N+1)\bigr]}{(V+\widetilde N)(\widetilde N+1)}
=\frac{V}{\widetilde N+1}\qquad(\widetilde N\ge1).
\]
When $\widetilde N=0$ every remaining masked point is real ($q=1$) and the
value moves from $V$ to $V-1$, a strict decrease; the supermartingale
inequality thus also holds at this boundary, in the direction that only
lowers the stopped expectation.
Optional stopping for the bounded backward stopping time
$\widehat M\vee1$ gives, with the tail factor
$w:=\widetilde N(A)/(N(A)+1)$ revealed at backward time $M_{\max}$
already and $N(A)\ge V(A)$,
\[
\E[Q\mid\mathcal M]
\le\E\bigl[w\,U_{M_{\max}}\bigm|\mathcal M\bigr]
\le\E\!\left[\frac{V(\widehat\Gamma_{M_{\max}})}
{\widetilde N(\widehat\Gamma_{M_{\max}})+1}\cdot
\frac{\widetilde N(A)}{V(A)+1}\Bigm|\mathcal M\right]\le 1 ,
\]
the first step by optional stopping for the backward supermartingale (an
equality whenever the synthetic count stays positive along the path), and
the last step by Lemma~\ref{lem:urn}: conditional on $\mathcal M$ and on
the masked totals inside $\widehat\Gamma_{M_{\max}}\cup A$ (the two sets
are disjoint since $\widehat\Gamma_{M_{\max}}\subseteq[0,\lambda]$ at
every node), the identity split between the two cells is exactly the
uniform urn of the lemma --- the exchangeable assignment marginalizes
over the remaining masked locations --- and the lemma's bound holds for
every value of the totals, hence after the tower property. Combining
with (b), $\FDR\le\alpha\,\E[\E[Q\mid\mathcal M]]\le\alpha$.
\end{proof}

\begin{remark}[Per-node decomposition is not required]
A natural concern is that ``a high-$q^{(i)}$ node decides the global tail
ratio $N(A)/\widetilde N(A)$ without contributing to $\widehat\Gamma_M$.''
This does not break step (b), which uses no consistency property of the
tail ratio: $N(A)/\widetilde N(A)$ is not in general consistent for
$\pi_0$ (with $\tilde m=cm$ it tends to $1/c$ under the global null, and
non-null tail mass adds a further, upward, discrepancy). What step (b)
uses is exchangeability: under the oracle allocation the synthetic points
at every node are exchangeable with the real nulls there, and the
hypergeometric step requires only this \emph{combined} exchangeability
of the masked nulls, not per-node decomposition, even when
$\widehat\Gamma_M$ is geographically concentrated.
\end{remark}

\begin{proof}[Proof of Theorem~\ref{thm:bonus-pernode}]
\emph{(a) Per-node exchangeability, for any allocation, given the masked
data.} Let
\[
\mathcal M:=\sigma\Bigl(H_0,\ \{\widetilde X_k\}_{k\notin H_0},\
\{\tilde m^{(i)}\},\ \text{all pooled p-value locations}\Bigr)
\]
be the identity-blind $\sigma$-field: it reveals which hypotheses are
null, the non-null points, the bag sizes (data-independent by
assumption), and the \emph{locations} of all masked (real-null and
synthetic) p-values at every node---but not which masked point is which.
The nested family $\{\widehat\Gamma_M\}$, the budget shares
$\{\alpha_i\}$, and the candidate-node filter are all
$\mathcal M$-measurable (they are functions of pooled counts).
Conditional on $\mathcal M$, the identity assignment at node $i$ is a
uniform urn---$n_0^{(i)}$ real labels among the $n_0^{(i)}+\tilde m^{(i)}$
masked locations, because real nulls and synthetics are i.i.d.\
$\mathrm{Uniform}[0,1]$ whatever the bag sizes are---and the urns are
independent across nodes. No cross-node structure is needed---this is why
the allocation is free.

\emph{(b) Pointwise budget decomposition.} Write
$\widehat W^{(i)}(M):=\frac{N^{(i)}(A)+1}{\widetilde N^{(i)}(A)}
\bigl(\widetilde N^{(i)}(\widehat\Gamma_M^{(i)})+1\bigr)$ and
\[
Q_i:=\frac{V^{(i)}(\widehat\Gamma_{\widehat M}^{(i)})}
{\widetilde N^{(i)}(\widehat\Gamma_{\widehat M}^{(i)})+1}\cdot
\frac{\widetilde N^{(i)}(A)}{N^{(i)}(A)+1},
\]
with $Q_i:=0$ when $\widehat\Gamma_{\widehat M}^{(i)}=\emptyset$ or
$\widetilde N^{(i)}(A)=0$. On $\{\widehat M\ge1\}$ the stopping rule
\eqref{eq:bonus-budget} gives
$\widehat W^{(i)}(\widehat M)\le\alpha_i(R(\widehat M)\vee1)$ for every
active node, whence
\[
\FDP
=\sum_{i\in\mathcal N(\widehat M)}\frac{V^{(i)}(\widehat\Gamma_{\widehat M}^{(i)})}{R(\widehat M)\vee1}
=\sum_{i\in\mathcal N(\widehat M)}\frac{\widehat W^{(i)}(\widehat M)}{R(\widehat M)\vee1}\,Q_i
\le\sum_{i=1}^N\alpha_i\,Q_i ,
\]
and the bound extends trivially to $\{\widehat M=0\}$, where $\FDP=0$.

\emph{(c) Per-node optional stopping, conditional on $\mathcal M$.} Fix
$i$ and define the joint backward filtration $(\F_M)$ that contains
$\mathcal M$, the tail identity splits
$(V^{(j)}(A),\widetilde N^{(j)}(A))$ at every node, and, for $M'\ge M$,
the identity splits inside $\widehat\Gamma_{M'}^{(j)}$ at every node $j$.
The index $\widehat M$ of \eqref{eq:bonus-budget} is a function of
$\{R(M')\}_{M'\ge M}$, the per-node region splits at $M'\ge M$, the tail
splits, and the $\mathcal M$-measurable shares---all
$\F_M$-measurable---hence a backward stopping time of $(\F_M)$.
Conditional on $\F_{M+1}$, node $i$'s identities inside
$\widehat\Gamma_{M+1}^{(i)}$ remain a uniform urn (the other nodes'
reveals are independent of it), so peeling one $\mathcal M$-measurable
point at a time shows that
$U_M^{(i)}:=V^{(i)}(\widehat\Gamma_M^{(i)})/
(\widetilde N^{(i)}(\widehat\Gamma_M^{(i)})+1)$ is a backward
supermartingale, and an exact martingale whenever $\widetilde N\ge1$:
if the peeled point is real with conditional
probability $q=V/(V+\widetilde N)$,
\[
q\,\frac{V-1}{\widetilde N+1}+(1-q)\,\frac{V}{\widetilde N}
=\frac{V\bigl[(V-1)+(\widetilde N+1)\bigr]}{(V+\widetilde N)(\widetilde N+1)}
=\frac{V}{\widetilde N+1}\qquad(\widetilde N\ge1).
\]
When $\widetilde N=0$ every remaining masked point is real ($q=1$) and the
value moves from $V$ to $V-1$, a strict decrease; the supermartingale
inequality thus also holds at this boundary, in the direction that only
lowers the stopped expectation.
Optional stopping for the bounded backward stopping time $\widehat M$
gives, with the tail factor
$w_i:=\widetilde N^{(i)}(A)/(N^{(i)}(A)+1)$ revealed at backward time
$M_{\max}$ already,
\[
\E\bigl[Q_i\bigm|\mathcal M\bigr]
\le\E\Bigl[w_i\,U^{(i)}_{M_{\max}}\Bigm|\mathcal M\Bigr]
\le\E\!\left[\frac{V^{(i)}(\widehat\Gamma^{(i)}_{M_{\max}})}
{\widetilde N^{(i)}(\widehat\Gamma^{(i)}_{M_{\max}})+1}\cdot
\frac{\widetilde N^{(i)}(A)}{V^{(i)}(A)+1}\Bigm|\mathcal M\right]
\le 1 ,
\]
where the first step is optional stopping for the backward supermartingale
(an equality whenever the synthetic count stays positive along the path),
the middle inequality uses $N^{(i)}(A)\ge V^{(i)}(A)$, and the last
is Lemma~\ref{lem:urn}: conditional on $\mathcal M$ and on the masked
totals inside $\widehat\Gamma^{(i)}_{M_{\max}}\cup A$, the identity split
between the two cells is exactly the uniform urn of the lemma (the
exchangeable assignment marginalizes over the remaining masked points),
and the lemma's bound holds for every value of the totals, hence after
the tower property.

\emph{(d) Combine.} The shares $\alpha_i$ are $\mathcal M$-measurable, so
taking $\E[\,\cdot\mid\mathcal M\,]$ in (b) before integrating,
\[
\FDR\le\E\Bigl[\sum_i\alpha_i\,\E[Q_i\mid\mathcal M]\Bigr]
\le\E\Bigl[\sum_i\alpha_i\Bigr]\le\alpha .
\]
The conditional route is what licenses \emph{data-dependent}
(masking-measurable) budget shares: an unconditional bound
$\E[Q_i]\le1$ would not suffice, since $\alpha_i$ and $Q_i$ are
correlated through the pooled counts.
\end{proof}

\section{Lifted-space supermartingale; proofs of Theorems~\ref{thm:cfga} and~\ref{thm:aug}}
\label{app:lifted}

\begin{lemma}[Lifted-space adaptive control]\label{lem:lifted}
Work under Assumption~\ref{ass:model} and fix a fold $b$ with ranking half
$\D_b$ and inference half $\D_{b'}$ of size $m_{b'}$ (Section~\ref{sec:method}).
Running the greedy on $\D_b$ yields a \emph{nested, increasing} family of lifted
regions $\widetilde\Gamma_1\subseteq\cdots\subseteq\widetilde\Gamma_{M_{\max}}$,
where $\widetilde\Gamma_M=\bigcup_i\{i\}\times\widehat\Gamma_M^{(i,b)}$ is
$\sigma(\D_b)$-measurable. With $\widetilde X_k=(I_k,P_k)$ and
$\pi^{(i)}=q^{(i)}r_0^{(i)}$, set
\[
V(M)=\#\{k\in\D_{b'}:k\ \text{null},\,\widetilde X_k\in\widetilde\Gamma_M\},\qquad
p_M=\sum_i\pi^{(i)}\,\nu(\widehat\Gamma_M^{(i,b)}),
\]
so that $p_M$ is the probability that a single hypothesis is \emph{both} null
and inside $\widetilde\Gamma_M$, and $\E[V(M)\mid\D_b]=m_{b'}p_M$. Let
$R(M)=\#\{k\in\D_{b'}:\widetilde X_k\in\widetilde\Gamma_M\}$ be the
held-out rejection count, and define the backward filtration
\begin{equation}\label{eq:lem-lifted-filt}
\F_M:=\sigma\Bigl(\sigma(\D_b),\bigl\{\1\{\widetilde X_k\in\widetilde\Gamma_{M'}\},\,\1\{k\text{ null}\}:M'\ge M,\,k\in\D_{b'}\bigr\}\Bigr),
\end{equation}
which shrinks as $M$ grows ($\F_1\supseteq\cdots\supseteq\F_{M_{\max}}$).
Then, conditionally on $\D_b$,
\[
\E\!\left[\frac{V(M^\dagger)}{1+m_{b'}\,p_{M^\dagger}}\,\Big|\,\D_b\right]\;\le\;1
\]
for every $(\F_M)$ backward stopping time $M^\dagger$. In particular this
applies to the data-driven CFGA index $\widehat M^{(b)}$, because
$\widehat M^{(b)}$ is a function of the cumulative counts $R(M)$ (each
$\F_M$-measurable) and the $\sigma(\D_b)$-measurable masses $p_M$.
\end{lemma}

\noindent The lemma is the engine behind Theorems~\ref{thm:cfga} and
\ref{thm:aug}: it bounds the held-out null count $V(M^\dagger)$ against its
predicted mean $m_{b'}p_{M^\dagger}$ for an \emph{arbitrary} data-driven stopping
index---exactly the guarantee an adaptive procedure needs, since the index it
stops at is itself random.


\begin{proof}
\emph{Setup.} Condition on $\D_b$ throughout; since $\D_b$ identifies which
hypotheses fell in fold $b$, this conditioning determines the Bernoulli
fold assignment $(Z_k)$ of Algorithm~\ref{alg:cfga}, Step~1, and in
particular the inference-half size $m_{b'}=m-m_b$. Because the hypotheses
are i.i.d.\ (Assumption~\ref{ass:model}) and the fold assignment is drawn
independently of the data, conditional on $\D_b$ the inference-half
hypotheses $(\widetilde X_k)_{k\in\D_{b'}}$ are i.i.d.\ from the model and
independent of the now-deterministic regions $\widetilde\Gamma_M$. (This is
precisely where the i.i.d.\ fold assignment is needed: an exact per-node
half-split would make the inference-half label counts essentially
deterministic given $\D_b$, the labels would no longer be i.i.d., and the
single thinning ratio in \eqref{eq:lifted-cond} below would not exist.)
By the lemma's definition of $p_M$, the joint event
``$k$ null and $\widetilde X_k\in\widetilde\Gamma_M$'' has probability $p_M$
for each $k\in\D_{b'}$, so summing the indicators of this event over $k$ gives
the mean identity $\E[V(M)\mid\D_b]=m_{b'}p_M$ asserted in the lemma.

\emph{Reverse filtration.} The greedy admits one interval at a time, so
the regions grow and $V(M),p_M$ are nondecreasing in $M$. We index $M$
\emph{backward}, from the full region $M_{\max}$ down to $1$, under the
filtration $(\F_M)$ of \eqref{eq:lem-lifted-filt}. Moving from $\F_{M+1}$
to $\F_M$ reveals, for each held-out point, whether it also lies in the
smaller region $\widetilde\Gamma_M$.

\emph{Degenerate cases.} If the candidate family is empty ($M_{\max}=0$)
the procedure rejects nothing and $\FDP=0$. If $\pi_0=0$ there are no nulls,
$V\equiv0$, and the claim is trivial. If $p_{M+1}=0$ for some $M$, then
$p_M=0$ by nestedness and $V(M+1)=V(M)=0$ almost surely, so
$U_M=U_{M+1}=0$ and the step below holds trivially. We therefore assume
$\pi_0>0$ and $p_{M+1}>0$ wherever we divide by them.

\emph{One-step thinning.} Fix $M<M_{\max}$. Conditional on $\F_{M+1}$, the
distribution of $V(M)$ depends on $\F_{M+1}$ only through $V(M+1)$: the
surviving set $S_{M+1}=\{k\in\D_{b'}:k\text{ null},\,\widetilde X_k\in\widetilde\Gamma_{M+1}\}$
has size $V(M+1)$, and the only remaining randomness for $V(M)$ is whether
each survivor is also in the smaller region $\widetilde\Gamma_M$. By Bayes
and nestedness, this per-hypothesis probability is the same for every
$k\in S_{M+1}$:
\begin{equation}\label{eq:lifted-cond}
\PP\bigl(\widetilde X_k\in\widetilde\Gamma_M\mid k\text{ null},\,
\widetilde X_k\in\widetilde\Gamma_{M+1}\bigr)
=\frac{\PP(\widetilde X_k\in\widetilde\Gamma_M\mid k\text{ null})}
{\PP(\widetilde X_k\in\widetilde\Gamma_{M+1}\mid k\text{ null})}
=\frac{p_M}{p_{M+1}},
\end{equation}
where the last equality uses
\begin{equation}\label{eq:lifted-marg}
\PP(\widetilde X_k\in\widetilde\Gamma_M\mid k\text{ null})
=\sum_i\frac{\pi^{(i)}}{\pi_0}\,\nu(\widehat\Gamma_M^{(i,b)})=\frac{p_M}{\pi_0}
\end{equation}
(a null hypothesis carries label $i$ with posterior $\pi^{(i)}/\pi_0$ and a
Uniform[0,1] p-value given the label). The inference-half hypotheses are
i.i.d.\ given $\sigma(\D_b)$, so these per-hypothesis indicators are
independent across $k\in S_{M+1}$. Summing,
\[
V(M)\mid\F_{M+1}\;\sim\;\mathrm{Binomial}\bigl(V(M+1),\,p_M/p_{M+1}\bigr),
\]
and $\E[V(M)\mid\F_{M+1}]=V(M+1)\,p_M/p_{M+1}$.

\emph{Backward super-martingale.} Put $U_M=V(M)/(1+m_{b'}p_M)$. Therefore
\[
\begin{aligned}
\E[U_M\mid\F_{M+1}]
&=\frac{V(M+1)\,p_M/p_{M+1}}{1+m_{b'}p_M}\\
&\le\frac{V(M+1)}{1+m_{b'}p_{M+1}}=U_{M+1},
\end{aligned}
\]
the inequality being $p_M(1+m_{b'}p_{M+1})\le p_{M+1}(1+m_{b'}p_M)$, which
simplifies to $p_M\le p_{M+1}$ and holds by nestedness. Thus $(U_M,\F_M)$ is a
backward super-martingale.

\emph{Optional stopping.} At the largest region the predicted mass equals
the conditional expected null count, $\E[V(M_{\max})\mid\D_b]=m_{b'}p_{M_{\max}}$, so
\[
\E[U_{M_{\max}}\mid\D_b]
=\frac{m_{b'}p_{M_{\max}}}{1+m_{b'}p_{M_{\max}}}\le1.
\]
Since $\sigma(\D_b)\subseteq\F_M$ for every $M$, the backward
supermartingale property holds under the conditional measure
$\PP[\cdot\mid\D_b]$, and optional stopping for backward supermartingales
applied to the bounded backward stopping time $M^\dagger$ gives
$\E[U_{M^\dagger}\mid\D_b]\le\E[U_{M_{\max}}\mid\D_b]\le1$, as claimed.
\end{proof}

\begin{proof}[Proof of Theorem~\ref{thm:cfga}]
Fix fold $b\in\{1,2\}$ with ranking half $\D_b$ and inference half
$\D_{b'}$ (where $b':=3-b$)
of size $m_{b'}$; set $\beta=\alpha/2$ and condition on $\D_b$ throughout.
Let $\widehat M^{(b)}$ be the stopping index of fold $b$ from
Algorithm~\ref{alg:cfga} (a random function of both halves), with
$\widehat M^{(b)}=0$ when no $M$ meets the stopping rule. Set
$M^\dagger:=\widehat M^{(b)}\vee 1$, and define
$U_M:=V(M)/(1+m_{b'}p_M)$ as in Lemma~\ref{lem:lifted}.

Let
\[
\Phi := \frac{V(\widehat M^{(b)})}{R(\widehat M^{(b)})\vee 1}
\]
denote the fold-$b$ false-discovery proportion, with the convention
$V(0)=R(0)=0$ (no rejections) so that $\Phi=0$ on
$\{\widehat M^{(b)}=0\}$. On $\{\widehat M^{(b)}\ge 1\}$, the oracle
stopping rule $(1+m_{b'}p_M)/(R(M)\vee 1)\le\beta$ at
$M=\widehat M^{(b)}$ gives
\[
\Phi
=\frac{V(\widehat M^{(b)})}{1+m_{b'}\,p_{\widehat M^{(b)}}}\cdot
\frac{1+m_{b'}\,p_{\widehat M^{(b)}}}{R(\widehat M^{(b)})\vee 1}
\le \beta\,U_{\widehat M^{(b)}}.
\]
Combining the two events, $\Phi\le\beta\,U_{M^\dagger}$ everywhere (the
bound is trivial on $\{\widehat M^{(b)}=0\}$). Applying Lemma~\ref{lem:lifted}
to the backward stopping time $M^\dagger$,
\[
\E[\Phi\mid\D_b]\;\le\;\beta\,\E[U_{M^\dagger}\mid\D_b]\;\le\;\beta.
\]

The two folds reject disjoint subsets
($\mathcal R^{(1)}\subseteq\D_2$, $\mathcal R^{(2)}\subseteq\D_1$), so by the
elementary inequality
$\frac{a_1+a_2}{(b_1+b_2)\vee 1}\le\frac{a_1}{b_1\vee 1}+\frac{a_2}{b_2\vee 1}$,
$\FDR\le 2\beta = \alpha$.
\end{proof}

\begin{proof}[Proof of Theorem~\ref{thm:aug}]
Fix fold $b\in\{1,2\}$ with ranking half $\D_b$ (size $m_b$) and inference
half $\D_{b'}$ (size $m_{b'}$, $b':=3-b$), set $\beta=\alpha/2$, and
condition on $\D_b$ throughout. Let
$E=\{\widehat\pi_b^{(i),+}\ge\pi^{(i)}\ \forall i\}$, which is
$\sigma(\D_b)$-measurable. With the convention $V(0)=R(0)=0$, set
$V:=V(\widehat M^{(b)})$, $R:=R(\widehat M^{(b)})$ and
$\FDP_b:=V/(R\vee 1)$, which equals $0$ on $\{\widehat M^{(b)}=0\}$. Let
$U_M=V(M)/(1+m_{b'}p_M)$ as in Lemma~\ref{lem:lifted}, and write
$\bar V_{\mathrm{or}}(M):=m_{b'}p_M$ for the oracle null mass (the same
quantity as in Theorem~\ref{thm:cfga}). On $E$,
\[
\bar V(M):=m_{b'}\sum_i\widehat\pi_b^{(i),+}\nu(\widehat\Gamma_M^{(i,b)})
\;\ge\;\bar V_{\mathrm{or}}(M)\quad\text{for every }M.
\]

\emph{(i) On $E$.} On $\{\widehat M^{(b)}\ge 1\}$ the stopping rule gives
$(1+\bar V(\widehat M^{(b)}))/(R\vee 1)\le\beta$; by the dominance above,
$(1+\bar V_{\mathrm{or}}(\widehat M^{(b)}))/(R\vee 1)\le\beta$, and
\[
\FDP_b
=\frac{V}{1+\bar V_{\mathrm{or}}(\widehat M^{(b)})}\cdot
\frac{1+\bar V_{\mathrm{or}}(\widehat M^{(b)})}{R\vee 1}
\le\beta\,U_{\widehat M^{(b)}} .
\]
On $\{\widehat M^{(b)}=0\}\cap E$, $\FDP_b=0\le\beta\,U_{M^\dagger}$
trivially. Setting $M^\dagger:=\widehat M^{(b)}\vee 1$, the bound
$\FDP_b\,\1_E\le\beta\,U_{M^\dagger}$ holds everywhere, and
Lemma~\ref{lem:lifted} gives
\[
\E[\FDP_b\,\1_E\mid\D_b]\;\le\;\beta\,\E[U_{M^\dagger}\mid\D_b]\;\le\;\beta.
\]

\emph{(ii) Off $E$.} $\FDP_b\le 1$. Combining (i) and (ii) via
$\E[\FDP_b]=\E[\FDP_b\1_E]+\E[\FDP_b\1_{E^c}]$, the tower property
($E\in\sigma(\D_b)$) gives $\E[\FDP_b\1_E]\le\beta$ from (i) and
$\E[\FDP_b\1_{E^c}]\le\PP(E^c)$, so $\E[\FDP_b]\le\beta+\PP(E^c)$.

\emph{Concentration (the Hoeffding step).} It remains to bound $\PP(E^c)$.
Fix a node $i$ and write the ranking-half upper-tail count as
$X^{(i)}:=\#\{k\in\D_b:I_k=i,\,P_k>\lambda\}=\sum_{k\in\D_b}Y_k$
with $Y_k:=\1\{I_k=i,\,P_k>\lambda\}$, so
$\widehat\pi_b^{(i)}=X^{(i)}/(m_b(1-\lambda))$. Throughout this step we
work conditionally on the fold assignment $(Z_k)$ of
Algorithm~\ref{alg:cfga}, which fixes $m_b$ and hence the margin $\delta$
of \eqref{eq:delta} (both are assignment-measurable). As in the proof of
Lemma~\ref{lem:lifted}, the assignment is independent of the data, so
given $(Z_k)$ the $m_b$ hypotheses in $\D_b$ are i.i.d.\ from the model of
Assumption~\ref{ass:model}; the $Y_k$ are then i.i.d.\ $\{0,1\}$-valued,
and
\begin{align*}
\E[Y_k]\ge\PP\bigl(I_k=i,\ k\text{ null at node }i,\ P_k>\lambda\bigr)
=q^{(i)}r_0^{(i)}(1-\lambda)=\pi^{(i)}(1-\lambda),
\end{align*}
where the inequality drops the (nonnegative) non-null contribution and the
last line uses the label model with exactly uniform nulls. Hence
$\E[X^{(i)}]\ge m_b(1-\lambda)\pi^{(i)}$ and
$\E[\widehat\pi_b^{(i)}]\ge\pi^{(i)}$.

Because $\widehat\pi_b^{(i),+}=\widehat\pi_b^{(i)}+\delta$, the bad event
for node $i$ satisfies, with $t:=m_b(1-\lambda)\delta$,
\begin{align*}
\{\widehat\pi_b^{(i),+}<\pi^{(i)}\}
=\{X^{(i)}<m_b(1-\lambda)(\pi^{(i)}-\delta)\}
\subseteq\{X^{(i)}<\E[X^{(i)}]-t\},
\end{align*}
the inclusion using $m_b(1-\lambda)\pi^{(i)}\le\E[X^{(i)}]$. The summands
$Y_k$ are independent and take values in $[0,1]$, so Hoeffding's
lower-tail inequality gives
\[
\PP\big(X^{(i)}<\E[X^{(i)}]-t\big)\le e^{-2t^2/m_b}=e^{-2m_b(1-\lambda)^2\delta^2}.
\]
Substituting $\delta=(1-\lambda)^{-1}\sqrt{\log(2N/\eta)/(2m_b)}$ from
\eqref{eq:delta} makes the exponent exactly $-\log(2N/\eta)$, so
$\PP(\widehat\pi_b^{(i),+}<\pi^{(i)})\le\eta/(2N)$. A union bound over the
$N$ nodes gives $\PP(E^c\mid(Z_k))\le\eta/2$ for every realization of the
fold assignment, hence $\PP(E^c)\le\eta/2$ unconditionally, whence
$\E[\FDP_b]\le\beta+\eta/2$ per fold.

\emph{Fold combination.} The two folds reject within $\D_2$ and $\D_1$,
which are disjoint, so by the elementary inequality from the proof of
Theorem~\ref{thm:cfga}, $\FDP\le\FDP_1+\FDP_2$, hence
$\E[\FDP]\le\E[\FDP_1]+\E[\FDP_2]\le 2(\beta+\eta/2)=\alpha+\eta$.
\end{proof}

\begin{remark}[Tighter inflation using both halves' tails]\label{rem:pool-tail}
The estimator of Theorem~\ref{thm:aug} uses only the $\D_b$ tail,
leaving the inference-half tail $\D_{b'}\cap(\lambda,1]$ unused even
though it is disjoint from the inference-half head that drives $V(M)$;
this suggests pooling, but disjointness from the head is not by itself
sufficient. Pooling the two tails into
\begin{equation*}
\widehat\pi_{\rm pool}^{(i)}=\frac{\#\{k\in\D^{(i)}:P_k>\lambda\}}{(1-\lambda)\,m},
\qquad
\delta_{\rm pool}=\frac{1}{1-\lambda}\sqrt{\frac{\log(2N/\eta)}{2m}} ,
\end{equation*}
where $\D^{(i)}=\D_b^{(i)}\cup\D_{b'}^{(i)}$, yields a
$\sigma(\D_b,\,\D_{b'}\cap(\lambda,1])$-measurable estimator that still
dominates $\pi^{(i)}$ with probability $\ge 1-\eta/N$ by Hoeffding on
$m$ summands. One subtlety prevents Theorem~\ref{thm:aug} from carrying
over \emph{verbatim}: conditioning on the inference-half tail also
freezes the inference-half \emph{head count}
$m_H:=\#\{k\in\D_{b'}:P_k\le\lambda\}$, so the conditional null mass of
the head is $m_H\,p_M/\PP(P\le\lambda)$ rather than $m_{b'}p_M$; since
$m_H$ exceeds its mean $m_{b'}\PP(P\le\lambda)$ by $O(\sqrt{m_{b'}})$
with constant probability, the optional-stopping anchor of
Lemma~\ref{lem:lifted} requires a \emph{second} concentration event
controlling $m_H$, and we have \emph{not} shown that the combined
argument yields any net margin improvement over \eqref{eq:delta}. A
simple check counsels caution: with one node under the global null,
$\lambda=\tfrac12$ and two folds of size $n=10^4$, the discrepancy
between the head-conditional null rate and the pooled estimate is
$D=(3H-X-n)/n$ for independent $X,H\sim\mathrm{Bin}(n,\tfrac12)$, and at
$\eta=10^{-3}$ exact binomial summation gives
$\PP(D>\delta)\approx0.0069>\eta$ even for the \emph{original} margin
$\delta\approx0.039$, so the naive pooled construction does not deliver the
required conditional dominance. We therefore use the single-half
estimator of Theorem~\ref{thm:aug} throughout and record pooling only as
an unproved possible refinement.
\end{remark}

\begin{remark}[The offset can be removed from the CFGA stopping rule]\label{rem:no-offset}
The ``$+1$'' in the numerators of Theorems~\ref{thm:cfga} and~\ref{thm:aug}
is not needed. Keep the conventions of Theorem~\ref{thm:cfga}: an empty
candidate family or an empty inference fold ($m_{b'}=0$) makes no
rejections, and $\widehat M^{(b)}=0$ when no index qualifies, with
$\FDP_b=0$ in each case. With the notation of Lemma~\ref{lem:lifted},
define $Y_M:=V(M)/(m_{b'}p_M)$ when $m_{b'}p_M>0$ and $Y_M:=0$ when
$m_{b'}p_M=0$. The one-step thinning identity gives
$\E[Y_M\mid\F_{M+1}]=Y_{M+1}$ whenever $m_{b'}p_M>0$, an exact backward
martingale step; when $p_M=0$ we have $V(M)=0$ almost surely, so
$Y_M=0\le Y_{M+1}$; and $\E[Y_{M_{\max}}\mid\D_b]=1$ when
$m_{b'}p_{M_{\max}}>0$ (and $0$ otherwise).
Optional stopping for the bounded backward supermartingale $(Y_M)$ gives
$\E[Y_{M^\dagger}\mid\D_b]\le1$ for every backward stopping time
$M^\dagger$. Running fold $b$ with the stopping rule
$\widehat M^{(b)}=\max\{M:m_{b'}p_M/(R(M)\vee1)\le\alpha/2\}$ therefore
gives $\FDP_b\le(\alpha/2)\,Y_{\widehat M^{(b)}}$ on $\{p_{\widehat M^{(b)}}>0\}$
and $\FDP_b=0$ otherwise, and the fold combination of the proof of
Theorem~\ref{thm:cfga} yields $\FDR\le\alpha$ without the offset. The
inflated version follows from the same dominance argument as in the proof
of Theorem~\ref{thm:aug}, with $\bar V(M)\ge m_{b'}p_M$ on the good event.
The experiments retain the offset, which is how the procedures are
implemented; the observation does not apply to the BONuS-GA offsets, whose calibrator is a
synthetic count rather than a predicted mass, nor to the e-CFGA scores,
whose denominator must stay positive.
\end{remark}

\section{e-CFGA: construction, choices, and proof of Theorem~\ref{thm:ecfga}}
\label{app:ecfga-details}

\subsection{E-value background}
\label{sec:ecfga-evalue}

A nonnegative random variable $e$ is an \emph{e-value} for a null hypothesis
$H_0$ if $\E_{H_0}[e]\le 1$. The e-BH procedure of
\citet{WangRamdas2022} takes per-hypothesis scores $e_1,\dots,e_m\ge0$,
orders them in decreasing magnitude $e_{(1)}\ge e_{(2)}\ge\cdots$, and
rejects the $k^\star$ largest with $k^\star=\max\{k:k\,e_{(k)}\ge m/\alpha\}$.
Its FDR proof uses only the \emph{aggregate} null-mass bound
$\E\bigl[\sum_{k\in\Hzero}e_k\bigr]\le m$, where $\Hzero$ is the set of
true nulls across the network. Individual e-values satisfy this bound,
but it is weaker than requiring $\E[e_k]\le1$ for each null; following
\citet{WangRamdas2022} we call nonnegative scores satisfying the aggregate
bound \emph{generalized e-value scores}. e-BH on such scores controls
$\FDR\le\alpha$ under \emph{arbitrary} dependence among them. Two facts
make this calculus convenient for our setting:

\emph{(a) Averaging preserves the aggregate bound.} If
$e^{(1)},\dots,e^{(S)}$ are families of generalized e-value scores, then
$\bar e=S^{-1}\sum_s e^{(s)}$ is one too, by linearity of expectation, no
matter how the families are correlated. We can therefore average across
arbitrary dependent splits without a Bonferroni union over them (with
finitely many splits the output remains random; averaging reduces the
split-to-split variability rather than eliminating it).

\emph{(b) With the oracle null estimate, each CFGA fold supplies
generalized e-value scores.}
Fix a split index $s\in\{1,\dots,S\}$ and a fold label $b\in\{1,2\}$ with
ranking half $\D_b^s$ and inference half $\D_{b'}^s$ of size $m_{b'}^s$
($b':=3-b$). Lemma~\ref{lem:lifted} (the backward super-martingale
underlying Theorem~\ref{thm:cfga}), applied to the fold's stopping
index $\widehat M_{s,b}$ --- a backward stopping time of the fold's
lifted filtration $(\F_M)$ defined in~\eqref{eq:lem-lifted-filt}, not
$\sigma(\D_b^s)$-measurable since it depends on the inference-half
counts $R(M)$ --- with the oracle null mass
$\bar V(M)=m^s_{b'}p_M$, or with any plug-in that dominates the oracle
mass \emph{pathwise on a high-probability event} (the inflated
estimator of Section~\ref{sec:augmented}, handled exactly as in
Theorem~\ref{thm:aug}), gives
\begin{equation}\label{eq:efold-bound}
\E\!\left[\frac{V^{s,b}(\widehat M_{s,b})}{1+\bar V(\widehat M_{s,b})}\,\bigg|\,\D_b^s\right]
\;\le\;1 ,
\end{equation}
where
$V^{s,b}(\widehat M_{s,b})=\#\{k\in\D_{b'}^s:k\text{ null and rejected}\}$
is the inference-half null count of fold $(s,b)$. For the oracle mass,
\eqref{eq:efold-bound} holds for every ranking-half realization; for the
inflated plug-in it holds only \emph{on the dominance event} of
Theorem~\ref{thm:aug} (per-split confidence $\eta/S$), and the
complementary event is absorbed into the level by the good-event
aggregation of Appendix~\ref{app:ecfga}, built on the fold-level argument
of Appendix~\ref{app:lifted}, yielding the $\alpha+\eta$ guarantee of
Theorem~\ref{thm:ecfga}. Multiplying
\eqref{eq:efold-bound} by $m_{b'}^s$ shows that assigning the constant
\begin{equation}\label{eq:efold-evalue}
e_k^{(s,b)}\;=\;\frac{m_{b'}^s\,\1\{k\text{ rejected by fold }(s,b)\}}{1+\bar V(\widehat M_{s,b})}
\end{equation}
to every hypothesis $k\in\D_{b'}^s$ (and $e_k^{(s,b)}=0$ to every
$k\notin\D_{b'}^s$) makes
$\E\bigl[\sum_{k\in\Hzero\cap\D_{b'}^s}e_k^{(s,b)}\bigm|\D_b^s\bigr]\le m_{b'}^s$.
The two folds of split $s$ contribute at most $m_1^s+m_2^s=m$, and the
average over $S$ splits preserves the total:
$\E\bigl[\sum_{k\in\Hzero}\bar e_k\bigr]\le m$. Feeding $\bar e_k$ to e-BH at
level $\alpha$ then controls FDR by the generalized-e-value calculus
of \citet{WangRamdas2022}.

\subsection{Algorithm}

\begin{algo}[e-CFGA]\label{alg:ecfga}
Fix $\varepsilon>0$, $\alpha\in(0,1)$, a fold level $\gamma\in(0,\alpha]$
(default $\gamma=0.75\,\alpha$; see \S\ref{sec:ecfga-gamma}), and a number
of splits $S\ge1$.
\begin{enumerate}[leftmargin=*]
\item For $s=1,\dots,S$: draw an i.i.d.\ Bernoulli$(1/2)$ split as in
Algorithm~\ref{alg:cfga}; for each fold $b\in\{1,2\}$ run the greedy on
the ranking half and stop on the inference half at level $\gamma$,
$\widehat M_{s,b}=\max\{M:(1+\bar V(M))/(R(M)\vee1)\le\gamma\}$, where
$\bar V(M)$ is the fold's null-mass numerator (oracle $m_{b'}p_M$, or its
Storey/inflated plug-ins). Assign to every hypothesis $k$ in the
\emph{inference} half the score~\eqref{eq:efold-evalue},
and $e_k^{(s,b)}=0$ otherwise.
\item Average: $\bar e_k=\frac1S\sum_{s=1}^S\sum_{b}e_k^{(s,b)}$ (each
hypothesis is in exactly one inference half per split, so at most one
of $e_k^{(s,1)},e_k^{(s,2)}$ is nonzero).
\item Reject by e-BH at level $\alpha$: order
$\bar e_{(1)}\ge\bar e_{(2)}\ge\cdots$ and reject the $k^\star$ largest,
$k^\star=\max\{k:k\,\bar e_{(k)}\ge m/\alpha\}$.
\end{enumerate}
\end{algo}

\paragraph{Reading the algorithm.}
Step~1 runs $S$ independent CFGA-style selections, each on a fresh
Bernoulli$(\tfrac12)$ split, and stops each fold at the \emph{same target
level} $\gamma$ (not $\alpha/2$). Every hypothesis in the fold's
rejection set $\widetilde\Gamma_{\widehat M_{s,b}}$ then receives the
score $m_{b'}^s/(1+\bar V(\widehat M_{s,b}))$; hypotheses
outside the rejection set keep $e=0$. Step~2 averages the per-split
scores into one
number per hypothesis; because a hypothesis is in fold $1$ of split $s$
\emph{or} fold $2$, but not both, the inner sum has at most one nonzero
term per split. Step~3 applies e-BH at the \emph{global} level $\alpha$,
without any Bonferroni penalty across folds or across splits.

\subsection{Finite-sample FDR control}
\label{sec:ecfga-fdr}

Theorem~\ref{thm:ecfga} (Section~\ref{sec:ecfga}) states the guarantee.

The validity statement is uniform in $\gamma$: any choice
$\gamma\in(0,\alpha]$ produces a valid procedure, regardless of $S$.
The choice of $\gamma$ governs only the algorithm's power, through the
feasibility of the e-BH cut. Likewise, each fold may choose its bandwidth
adaptively by sub-split cross-validation on its own ranking half
(Algorithm~\ref{alg:adaptive-eps}): the choice is
measurable with respect to the ranking half and its independent
sub-split randomness (Proposition~\ref{prop:adaptive-FDR}), so the
fold-level bound
\eqref{eq:efold-bound} --- and with it Theorem~\ref{thm:ecfga} --- is
unaffected. This per-fold adaptive-$\varepsilon$ configuration is the one
used in the experiments of Section~\ref{sec:sim}.

\subsection{Choosing the fold level $\gamma$}
\label{sec:ecfga-gamma}

The e-BH cut needs $k^\star\bar e_{(k^\star)}\ge m/\alpha$ for at least
one $k^\star\ge 1$; if no $k$ satisfies this, e-CFGA makes no
rejection. One boundary calculation shows that $\gamma=\alpha$ leaves no
slack in the \emph{unanimous} regime. Consider a stylized setting in
which each fold stops with $R$ rejections on its inference half (so
$\bar V(\widehat M_{s,b})\approx\gamma R$ from the stopping rule) and a
hypothesis is rejected by its inference fold in every split. Its average
score is $\bar e\approx m^s_{b'}/(1+\gamma R)\approx m/(2\gamma R)$,
and taking $k^\star$ equal to the $2R$ such hypotheses across the two
folds gives $k^\star\bar e_{(k^\star)}\approx m/\gamma$: the e-BH
condition $m/\gamma\ge m/\alpha$ holds \emph{exactly at} $\gamma=\alpha$,
with nothing to spare. This calculation does not by itself imply that
$\gamma=\alpha$ must fail under non-unanimous profiles: averaging can
also \emph{expand} the support, and a profile whose splits reject
disjoint sets doubles $k^\star$ while halving $\bar e$, exactly
preserving $k^\star\bar e_{(k^\star)}$. Whether the cut fires at
$\gamma=\alpha$ therefore depends on the joint rejection profile across
splits; in our configurations it collapsed to zero rejections, and we do
not attempt to model the required slack.
$\gamma=0.75\alpha$ is an \emph{empirical default}: across the
experiments of Section~\ref{sec:sim} it kept FDR at or below nominal and
made rejections in all but the hardest heavy-tailed small-hub
configuration (Experiment~4, where the inflated tier returns no
rejections at $m_{\mathrm{hub}}=500$). Smaller $\gamma$ widens the safety
margin at a power cost on each individual fold.

\subsection{Derandomization}

Each Bernoulli$(\tfrac12)$ split produces a different fold-level
realization $(\widehat M_{s,b}, e_k^{(s,b)})$. Given the data, the $S$
split-level scores of a hypothesis are i.i.d.\ and bounded, so the
average $\bar e_k$ concentrates around its split-conditional mean at the
$1/\sqrt{S}$ rate. This rate concerns the averaged \emph{scores}; the
rejection set is a discontinuous function of them and stabilizes only for
hypotheses whose e-BH condition is not near a tie, so we do not claim a
$1/\sqrt S$ rate for the rejection set itself. We use $S=10$ throughout.
With finitely many splits the output remains random; averaging reduces,
but does not remove, its dependence on the split draws --- dependence
that a single Bonferroni-2 split carries in full.

\subsection{Communication}
\label{sec:ecfga-comm}

Each split runs a CFGA round at fold level $\gamma$, and one extra round
reports each node's nonzero averaged scores for the final e-BH cut:
$O((NS+S\sqrt m+\bar R)\log m)$ bits in three pipelined rounds for a fixed
bandwidth, where $\bar R\le m$ is the number of hypotheses with a nonzero
score. The message formats, the reason the center cannot form the scores
itself, and the downward quantization of the scores are in
Appendix~\ref{app:protocols-ecfga}.

\subsection{Proof of Theorem~\ref{thm:ecfga}}
\label{app:ecfga}

\begin{proof}[Proof of Theorem~\ref{thm:ecfga}]
\emph{Null e-mass.} Fix a split $s$ and fold $b$, and condition on the
split assignment $Z^s$ and the ranking half $\D^s_b$. The fold's stopping
index $\widehat M_{s,b}$ is the level-$\gamma$ step-up on the inference
counts, hence a backward stopping time of the fold's filtration
(Lemma~\ref{lem:lifted} with $\beta$ replaced by $\gamma$; the lemma is
level-free). Summing the scores over the nulls of the inference half,
\[
\sum_{k\in\Hzero} e_k^{(s,b)}
=m^s_{b'}\,\frac{V^{s,b}(\widehat M_{s,b})}{1+m^s_{b'}\,p_{\widehat M_{s,b}}},
\qquad
\E\Bigl[\sum_{k\in\Hzero} e_k^{(s,b)}\Bigm| Z^s,\D^s_b\Bigr]\le m^s_{b'}
\]
by Lemma~\ref{lem:lifted}. Taking expectations and summing the two folds
of split $s$ gives at most $m^s_1+m^s_2=m$; averaging over the $S$ splits
(linearity---no independence across splits is used),
\begin{equation}\label{eq:ecfga-mass}
\E\Bigl[\sum_{k\in\Hzero}\bar e_k\Bigr]\;\le\;m .
\end{equation}

\emph{e-BH.} Let $R=k^\star$ be the number of e-BH rejections. Every
rejected $k$ satisfies $\bar e_k\ge\bar e_{(k^\star)}\ge m/(\alpha R)$,
i.e.\ $\1\{k\ \mathrm{rejected}\}\le\alpha R\,\bar e_k/m$, so
\[
\FDP=\sum_{k\in\Hzero}\frac{\1\{k\ \mathrm{rejected}\}}{R\vee1}
\le\frac{\alpha}{m}\sum_{k\in\Hzero}\bar e_k ,
\]
and \eqref{eq:ecfga-mass} gives $\FDR\le\alpha$. (This is the
generalized-e-value e-BH argument of \citet{WangRamdas2022}; it holds
under arbitrary dependence among the $\bar e_k$.)

\emph{Inflated plug-in.} Replace $m^s_{b'}p_M$ by the inflated mass
$\bar V(M)$ of Algorithm~\ref{alg:aug} computed at confidence level
$\eta/S$, and let $E_s$ be the event that the inflated rates dominate the
truth in both folds of split $s$; as in the proof of
Theorem~\ref{thm:aug}, $\PP(E_s^c)\le\eta/S$. Define the dominating
family at the \emph{same} stopping index selected by the inflated
procedure, replacing only the denominator:
$e_k^{(s,b),\mathrm{dom}}=m^s_{b'}\1\{k\text{ rejected by fold }(s,b)\}
/(1+m^s_{b'}p_{\widehat M_{s,b}})$. The inflated stopping index is
determined by the inference-half counts $R(M)$ and the
$\sigma(\D^s_b)$-measurable inflated masses ($\sigma(\D^s_b,U)$-measurable
when the fold bandwidth is chosen adaptively with auxiliary randomness
$U$, in which case the filtration below is enlarged by $U$ as in
Proposition~\ref{prop:adaptive-FDR}), hence is a backward stopping
time of the filtration of Lemma~\ref{lem:lifted}, so
\eqref{eq:efold-bound} with the oracle mass holds at it and
\eqref{eq:ecfga-mass} holds for the dominating family on every
realization. On $\bigcap_sE_s$ the inflated masses dominate the oracle
ones, so $e_k^{(s,b)}\le e_k^{(s,b),\mathrm{dom}}$ and hence
$\bar e_k\le\bar e_k^{\mathrm{dom}}$ for every $k$; the e-BH bound
$\FDP\le(\alpha/m)\sum_{k\in\Hzero}\bar e_k$ then gives
$\FDP\le(\alpha/m)\sum_{k\in\Hzero}\bar e_k^{\mathrm{dom}}$ on that
event, and off it we bound $\FDP\le1$. Hence
$\FDR\le\alpha+\PP\bigl(\textstyle\bigcup_sE_s^c\bigr)\le\alpha+\eta$.
\end{proof}

\section{Winner's-curse bias of the selected interval densities}
\label{sec:gap}
\label{sec:gap-wc}

Greedy aggregation \citep{PournaderiXiang2024} is a plug-in procedure: at
each finite $m$ it estimates a per-node grid of interval densities, ranks
the intervals across the network by their empirical densities, and stops
when its plug-in FDP estimate first crosses $\alpha$. The asymptotic FDR
guarantee of \citet[Prop.~4]{PournaderiXiang2024} holds in the limit, when
the empirical density estimates concentrate at their population
counterparts and the plug-in is consistent. At any finite $m$ those
estimates carry noise, and crucially the procedure reuses the same counts
both to rank the intervals and to form the FDP estimate. This creates a
\emph{winner's curse}: the greedy admits the intervals whose empirical
counts are largest, and the nulls inside those intervals contribute to
exactly the same counts the procedure trusts as its calibrator.

\paragraph{What the analysis measures.}
By the choice of grid width (Section~\ref{sec:greedy}), each
admitted interval has the same predicted number of nulls, $m\varepsilon$:
the estimated null count at node $i$ times the interval width is
$m^{(i)}\widehat r_0^{(i)}\cdot\widehat L^{(i)}
=m^{(i)}\widehat r_0^{(i)}\cdot\varepsilon/(\widehat q^{(i)}\widehat r_0^{(i)})
=m\varepsilon$. With $R(M)$ the number of p-values in the first $M$
admitted intervals, the greedy's FDP estimate is
\begin{equation}\label{eq:wc-num-view-intro}
\widehat\FDP(M)\;=\;\frac{M\,m\,\varepsilon}{R(M)\vee 1}
\;=\;\bigl[\bar h_M\vee(Mm\varepsilon)^{-1}\bigr]^{-1},\qquad M\ge1,
\end{equation}
where $\bar h_M:=R(M)/(Mm\varepsilon)=M^{-1}\sum_{\ell=1}^M\widehat h_{(\ell)}$
is the average of the density statistics $\widehat h_{(\ell)}$ of the
admitted intervals (in decreasing order); \eqref{eq:wc-num-view-intro}
equals $1/\bar h_M$ whenever $R(M)>0$.
The results of this section quantify the \emph{selection bias of these
density statistics}: the amount by which the average $\bar h_M$ over the
selected intervals exceeds the conditional mean density of the same
intervals. This bias can make the FDP estimate too small, but two further
steps would be needed to turn it into a statement about $\widehat\FDP$ or
about the number of null p-values, and we do not claim either in
general. First, an upward bias of $\bar h_M$ does not by itself imply a
downward bias of its reciprocal: by Jensen's inequality
$\E[1/\bar h_M]\ge1/\E[\bar h_M]$, so the two effects can offset. Second,
under the global null, when all candidate intervals have the same
conditional mean, the expected number of nulls in the selected intervals
differs from the predicted number $Mm\varepsilon$ by
\[
\E[V(\widehat\Gamma_M)\mid\mathcal G]-M\,m\,\varepsilon
=M\,m\,\varepsilon\,\bigl\{\Delta(M\mid\mathcal G)+h_M-1\bigr\},
\]
where $\Delta$ is the selection bias of
Proposition~\ref{prop:wc-bias-body} and $h_M$ is the common conditional
mean density of the intervals, that is, the expected count of an interval
divided by the predicted null count $m\varepsilon$. The term $h_M-1$ comes
from error in the estimated null rate, not from selection, and the
propositions below do not control it. Both terms are \emph{relative}: the
difference in counts is $M\,m\,\varepsilon$ times their sum.

This section formalizes the selection bias, which is what a
concentration argument controls directly. Under a \emph{head-restricted}
variant of the algorithm (candidate intervals in $(0,\lambda]$; see
\S\ref{sec:gap-setup}) that admits a clean binomial description for each
interval count (Lemma~\ref{lem:binomial}, \S\ref{sec:gap-binom}), each
density estimate is sub-gamma at scale $O((\varepsilon m)^{-1/2})$
around its conditional mean (\S\ref{sec:gap-conc}), and the top-$M$
selection bias is bounded by a sub-gamma maximal inequality over
$\binom{|\Aset|}{M}\le(e|\Aset|/M)^M$ subset sums
(Proposition~\ref{prop:wc-bias-body}, \S\ref{sec:gap-bias}). A matching
lower bound under the global null, obtained from Slud's Gaussian lower
tail for each binomial cell together with a negative-association Chernoff
bound on the number of exceedances
(Proposition~\ref{prop:wc-lower}, \S\ref{sec:gap-lower}), shows that the
upper bound is tight up to constants: at the standard bandwidth
$\varepsilon\asymp m^{-1/2}$ the selection bias is at most of order
$m^{-1/4}\sqrt{\log m}$ in general and exactly of this order under the
global null, over the stated range of $M$ --- strictly slower than the
$m^{-1/2}$ parametric rate.

\paragraph{Scope of the two propositions.}
Propositions~\ref{prop:wc-bias-body}--\ref{prop:wc-lower} quantify the
selection bias of the density statistics at a \emph{fixed} number $M$
of admitted intervals; they are not statements about $\FDR-\alpha$ at the
algorithm's random stopping index. Indeed, under the global null with a
fixed $\alpha<1$ and the default bandwidth, the maximal cell density
concentrates near $1$ while any rejection requires density at least
$1/\alpha>1$, so the global-null FDR of the full procedure tends to
\emph{zero} even though \eqref{eq:wc-lower} holds. The finite-sample FDR
violations documented in Section~\ref{sec:sim} occur in configurations
with signal, where the bias inflates the counts of the intervals selected
up to the stopping index. We therefore present the winner's curse as a
mechanism that is consistent with the observed excess---quantified for
the density statistics by the two propositions---rather than as a lower
bound on the FDR gap itself.

\subsection{Setup}\label{sec:gap-setup}
Throughout this appendix we assume, in addition to
Assumption~\ref{ass:model}, that every non-null CDF $F_i$ has a bounded
density $f_i:=F_i'$ and that $\min_ir_0^{(i)}>0$; neither condition is
used anywhere else in the paper. We write
$f^{(i)}(t):=r_0^{(i)}+r_1^{(i)}f_i(t)$ for the density of the local
mixture $G_i$ and $\|f^{(i)}\|_\infty:=\sup_{t\in[0,1]}|f^{(i)}(t)|$ for
its uniform bound.
Fix a bandwidth $\varepsilon>0$ and a Storey threshold $\lambda\in(0,1)$,
and analyze the \emph{head-restricted} variant of the greedy of
Section~\ref{sec:greedy}, whose candidate intervals are those in the head
$(0,\lambda]$ (the original greedy uses $1\le j\le\lfloor1/\widehat
L^{(i)}\rfloor$; CFGA and BONuS-GA also restrict their candidates to the
head). At each node $i$, the (clipped, $+1$-offset) Storey estimate of the local
null proportion is
$\widehat r_0^{(i)}=\min\bigl\{1,\,(1+\#\{k:P_k^{(i)}>\lambda\})/((1-\lambda)m^{(i)})\bigr\}$,
as in Section~\ref{sec:greedy}, which depends only on the \emph{tail} count
$N_>^{(i)}:=\#\{k:P_k^{(i)}>\lambda\}$. The greedy grid width
\begin{equation}
\widehat L^{(i)}=\varepsilon/(\widehat q^{(i)}\widehat r_0^{(i)}),
\quad \widehat q^{(i)}:=m^{(i)}/m ,
\label{eq:body-wc-width}
\end{equation}
which lies in $[\varepsilon m/m^{(i)},\ \varepsilon m/(m^{(i)}r_{\min})]$
on the event $\mathcal E$ of \S\ref{sec:gap-conc} (since
$r_{\min}\le\widehat r_0^{(i)}\le1$ there), hence is of order $\varepsilon$
whenever $m^{(i)}/m$ is bounded away from zero. It places
$K^{(i)}=\lfloor\lambda/\widehat L^{(i)}\rfloor$ full-width intervals
$\widehat I_j^{(i)}=((j-1)\widehat L^{(i)},\,j\widehat L^{(i)}]$ in the
head $(0,\lambda]$ (any final remainder $(K^{(i)}\widehat
L^{(i)},\lambda]$, of width less than $\widehat L^{(i)}$, is unused), and the
greedy density estimates are
\begin{equation}
\widehat h_j^{(i)}=N^{(i)}(\widehat I_j^{(i)})/(\varepsilon m),\quad
N^{(i)}(I):=\#\{k:P_k^{(i)}\in I\} .
\label{eq:body-wc-h}
\end{equation}
The candidate set is
$\Aset:=\{(i,j):1\le i\le N,\,1\le j\le K^{(i)}\}$, with
\[
|\Aset|\;\le\;\sum_i\frac{\lambda}{\widehat L^{(i)}}
\;=\;\frac{\lambda}{\varepsilon}\sum_i\widehat q^{(i)}\widehat r_0^{(i)}
\;\le\;\frac{\lambda}{\varepsilon},
\]
where the last inequality uses
$\sum_i\widehat q^{(i)} = 1$ and the truncation $\widehat r_0^{(i)}\le 1$
from Section~\ref{sec:greedy}. All conditional
expectations below are conditional on the node sizes and tail counts
$\mathcal G:=\sigma\bigl(\{(m^{(i)},N_>^{(i)})\}_{i=1}^N\bigr)$; under
the hierarchical model, conditionally on the node labels the observations
at node $i$ are i.i.d.\ $G_i$, so conditioning on the node sizes leaves
the within-node p-values i.i.d.\ $G_i$. Each grid width
$\widehat L^{(i)}$ is a function of $(m^{(i)},N_>^{(i)})$, so $\widehat L^{(i)}$,
$\widehat I_j^{(i)}$, and $\Aset$ are all $\mathcal G$-measurable, hence
deterministic given $\mathcal G$. (We condition on the counts rather
than on the grid itself because the truncation $\widehat r_0^{(i)}\le1$
makes the map $N_>^{(i)}\mapsto\widehat L^{(i)}$ many-to-one on the
clipped branch, so the grid generates a strictly coarser $\sigma$-field.)

\subsection{Each interval count is conditionally binomial}\label{sec:gap-binom}
The Storey design makes the grid width $\widehat L^{(i)}$ a function of the
tail count $N_>^{(i)}$ and the candidate intervals subsets of the head
$(0,\lambda]$. Because the head and the tail of an i.i.d.\ sample are
informationally separated, the count $N^{(i)}(\widehat I_j^{(i)})$ inherits
a binomial structure once we condition on the grid.

\begin{lemma}[Conditional binomial form]\label{lem:binomial}
Assume $r_0^{(i)}>0$, so that $\PP(P^{(i)}\le\lambda)\ge r_0^{(i)}\lambda>0$.
For each $i$, conditionally on $(m^{(i)},N_>^{(i)})$ (which
determines $\widehat L^{(i)}$: without truncation
$\widehat L^{(i)}=\varepsilon m^{(i)}(1-\lambda)/(\widehat q^{(i)}(1+N_>^{(i)}))$,
and the truncated estimator is this expression clipped, so
$\widehat L^{(i)}$ is a---possibly many-to-one---function of
$N_>^{(i)}$), the head p-values
$\{P_k^{(i)}\}_{k:\,P_k^{(i)}\le\lambda}$ are i.i.d.\ from
$f^{(i)}\!\restriction_{[0,\lambda]}/\PP(P^{(i)}\le\lambda)$, independently
of $N_>^{(i)}$, and
\begin{equation}
N^{(i)}(\widehat I_j^{(i)})\,\big|\,\bigl(m^{(i)},N_>^{(i)}\bigr)\;\sim\;
\mathrm{Bin}\!\bigl(m^{(i)}-N_>^{(i)},\,\widetilde p_j^{(i)}\bigr)
\label{eq:body-wc-binom}
\end{equation}
with per-trial success probability
$\widetilde p_j^{(i)}=\bigl(\int_{\widehat I_j^{(i)}}f^{(i)}\bigr)/\PP(P^{(i)}\le\lambda)$,
which satisfies $r_0^{(i)}\widehat L^{(i)}\le\widetilde p_j^{(i)}\le
\|f^{(i)}\|_\infty\widehat L^{(i)}/\PP(P^{(i)}\le\lambda)$ and is therefore
of order $\varepsilon$ under the regularity conditions just noted.
\end{lemma}

\begin{proof}
Work conditionally on $m^{(i)}$ throughout: under the hierarchical model
the observations at node $i$ are, conditionally on the node labels,
i.i.d.\ $G_i$, so given $m^{(i)}$ the node's p-values are
i.i.d.\ $G_i$. Decompose each p-value $P_k^{(i)}$ via its tail indicator
$\xi_k=\1\{P_k^{(i)}>\lambda\}$ and two auxiliary draws
$Y_k\sim f^{(i)}\!\restriction_{(\lambda,1]}/\PP(P^{(i)}>\lambda)$ and
$Z_k\sim f^{(i)}\!\restriction_{[0,\lambda]}/\PP(P^{(i)}\le\lambda)$,
where $\{\xi_k\}$, $\{Y_k\}$, $\{Z_k\}$ are jointly independent across $k$
and $P_k^{(i)}=\xi_k Y_k+(1-\xi_k)Z_k$ in distribution. The Storey estimate
$\widehat r_0^{(i)}=\min\{1,(1+\sum_k\xi_k)/((1-\lambda)m^{(i)})\}$ is a function of
$\{\xi_k\}$ alone, hence so is $\widehat L^{(i)}$. Conditioning further on
$\sum_k\xi_k=N_>^{(i)}$ (which determines $\widehat L^{(i)}$) leaves the
collection $\{Z_k\}_{k:\,\xi_k=0}$ i.i.d.\ from
$f^{(i)}\!\restriction_{[0,\lambda]}/\PP(P^{(i)}\le\lambda)$ and
independent of $\widehat L^{(i)}$. Since $\widehat I_j^{(i)}\subseteq(0,\lambda]$,
the indicators
$\bigl\{\1\{P_k^{(i)}\in\widehat I_j^{(i)}\}\bigr\}_{k:\,\xi_k=0}$ are
i.i.d.\ Bernoulli$(\widetilde p_j^{(i)})$, summing to
\eqref{eq:body-wc-binom}.
\end{proof}

\subsection{Sub-gamma concentration of a single estimate}\label{sec:gap-conc}
Choose the population analogue $h_j^{(i)}$ of $\widehat h_j^{(i)}$ to be the
conditional expectation
\begin{equation}
h_j^{(i)}\;:=\;\E\!\bigl[\widehat h_j^{(i)}\,\bigm|\,\mathcal G\bigr]
\;=\;\frac{(m^{(i)}-N_>^{(i)})\,\widetilde p_j^{(i)}}{\varepsilon m},
\label{eq:body-wc-h-pop}
\end{equation}
so the mean identity is exact by construction. By
Lemma~\ref{lem:binomial}, $\widehat h_j^{(i)}-h_j^{(i)}$ is, conditional on
$\mathcal G$, a centered sum of independent rescaled indicators
$X_k=\1\{P_k^{(i)}\in\widehat I_j^{(i)}\}/(\varepsilon m)$, each in
$[0,1/(\varepsilon m)]$. Fix $r_{\min}$ with $0<r_{\min}<\min_ir_0^{(i)}$
and let $\mathcal E:=\{\min_i\widehat r_0^{(i)}\ge r_{\min}\}$, a
$\mathcal G$-measurable event. On $\mathcal E$ the total variance obeys
\begin{equation}
\operatorname{Var}\!\bigl(\widehat h_j^{(i)}\,\bigm|\,\mathcal G\bigr)
=\frac{(m^{(i)}-N_>^{(i)})\,\widetilde p_j^{(i)}(1-\widetilde p_j^{(i)})}{(\varepsilon m)^2}
\;\le\;\frac{\|f^{(i)}\|_\infty}{r_{\min}^{2}\,\lambda\,\varepsilon m},
\label{eq:body-wc-var}
\end{equation}
using $m^{(i)}-N_>^{(i)}\le m^{(i)}$,
$\widetilde p_j^{(i)}\le\|f^{(i)}\|_\infty\widehat L^{(i)}/\PP(P^{(i)}\le\lambda)$,
$\widehat L^{(i)}=\varepsilon m/(m^{(i)}\widehat r_0^{(i)})\le\varepsilon m/(m^{(i)}r_{\min})$
on $\mathcal E$, and $\PP(P^{(i)}\le\lambda)\ge r_0^{(i)}\lambda>r_{\min}\lambda$.
This is only an upper bound: the variance can be far smaller, and is zero
when the node's head sample is empty.

We apply Bernstein's inequality
\citep[Thm.~2.10]{BoucheronLugosiMassart2013}: if $X_1,\dots,X_n$ are
independent with $|X_k-\E X_k|\le b$ a.s.\ and total variance
$V=\sum_k\operatorname{Var}(X_k)$, then the centered sum
$S=\sum_k(X_k-\E X_k)$ is \emph{sub-gamma on the right tail} with
parameters $(V,b/3)$ --- written $S\in\Gamma_+(V,b/3)$ --- meaning
$\PP(S\ge t)\le\exp(-t^2/(2(V+bt/3)))$ for all $t\ge0$. With
$b=1/(\varepsilon m)$ and $V\le v$ from \eqref{eq:body-wc-var}, where
\begin{equation}
\widehat h_j^{(i)}-h_j^{(i)}\;\in\;\Gamma_+\!\bigl(v,\,c\bigr)\;\;
\text{cond.\ on }\mathcal G,\ \text{on }\mathcal E,\;\;
v=\frac{\max_i\|f^{(i)}\|_\infty}{r_{\min}^{2}\lambda\,\varepsilon m},\;c=\tfrac{1}{3\varepsilon m}.
\label{eq:body-wc-subgamma}
\end{equation}
The standard deviation of a single estimate around its conditional mean is
therefore at most $\sqrt v=O((\varepsilon m)^{-1/2})$ on $\mathcal E$;
it is of exactly this order only on a further nondegeneracy event on
which each node's head sample is a constant fraction of $m^{(i)}$.

\subsection{The top-$M$ selection bias}\label{sec:gap-bias}
The greedy reuses the same counts in the rejection denominator: it ranks
the $|\Aset|\le 1/\varepsilon$ estimates and forms
\begin{equation}
\widehat\FDP(M)=\bigl[\bar h_M\vee(Mm\varepsilon)^{-1}\bigr]^{-1},\quad
\bar h_M=\frac{1}{M}\sum_{\ell=1}^M\widehat h_{(\ell)},\qquad M\ge1,
\label{eq:wc-fdp-body}
\end{equation}
which equals $1/\bar h_M$ whenever $R(M)>0$,
so the top-$M$ \emph{order statistics} $\widehat h_{(\ell)}$ ---  not the
individual estimates --- determine the stopping index. The empirical
top-$M$ estimates lie above their population counterparts on average; the
following proposition quantifies the gap.

\begin{proposition}[Top-$M$ selection bias]\label{prop:wc-bias-body}
Let $\widehat S(M)\subseteq\Aset$ be the size-$M$ subset selected by the
empirical ranking, and define the selection bias
\[
\Delta(M\mid\mathcal G)
:=\E\!\left[\frac{1}{M}\sum_{(i,j)\in\widehat S(M)}\bigl(\widehat h_j^{(i)}-h_j^{(i)}\bigr)\,\Big|\,\mathcal G\right] .
\]
Fix $r_{\min}$ with $0<r_{\min}<\min_i r_0^{(i)}$. Then on the event
$\mathcal E=\{\inf_i\widehat r_0^{(i)}\ge r_{\min}\}$,
\begin{equation}
\Delta(M\mid\mathcal G)
\le\sqrt{\frac{C_1\log(e|\Aset|/M)}{\varepsilon m}}
+\frac{C_2\log(e|\Aset|/M)}{\varepsilon m},
\label{eq:wc-Delta-body}
\end{equation}
with $C_1,C_2>0$ depending only on $\sup_i\|f^{(i)}\|_\infty$, $r_{\min}$,
and $\lambda$. Conditionally on the node sizes,
$\PP(\mathcal E^c\mid\{m^{(i)}\})\le N\exp(-c_\star\min_i m^{(i)})$ with
$c_\star=2(\inf_i r_0^{(i)}-r_{\min})^2(1-\lambda)^2$, via Hoeffding on
the Storey tail counts; averaging over the multinomial node sizes yields
the unconditional bound
$\PP(\mathcal E^c)\le N\exp(-c_\star q_{\min}m/2)+N\exp(-q_{\min}m/8)$,
$q_{\min}:=\min_i q^{(i)}$.
\end{proposition}

\paragraph{The default bandwidth.}
At $\varepsilon=\alpha/\sqrt m$, $|\Aset|\asymp m^{1/2}$ and
$\log|\Aset|=\tfrac12\log m+O(1)$; substituting into~\eqref{eq:wc-Delta-body},
\[
\Delta(M\mid\mathcal G)\le O\!\bigl(m^{-1/4}\sqrt{\log m}\bigr)
\quad\text{uniformly in }M\le|\Aset|.
\]

\paragraph{Proof strategy.} The selection bias is at most the maximum
of $\binom{|\Aset|}{M}$ centered subset sums $U_S$ over $|S|=M$.
Disjointness of the within-node grid recasts each $U_S$ as a sum of
independent indicators of magnitude $(\varepsilon m)^{-1}$ with
variance at most $Mv$ on $\mathcal E$, making it sub-gamma
at scale $\sqrt{M/(\varepsilon m)}$ by Bernstein. A sub-gamma maximal
inequality \citep[Cor.~2.6]{BoucheronLugosiMassart2013} over the
$\binom{|\Aset|}{M}\le(e|\Aset|/M)^M$ subsets supplies the
$\sqrt{\log(e|\Aset|/M)}$ factor in \eqref{eq:wc-Delta-body}. To pass to
the unconditional statement, condition first on the node sizes (which are
random under Assumption~\ref{ass:model}): Hoeffding's one-sided
inequality on the per-node Storey tail counts
$\sum_k\1\{P_k^{(i)}>\lambda\}$ combined with a union bound across $N$
nodes gives
$\PP(\mathcal E^c\mid\{m^{(i)}\})\le N\exp(-c_\star\min_i m^{(i)})$ with
$c_\star=2(\inf_i r_0^{(i)}-r_{\min})^2(1-\lambda)^2$. Averaging over the
multinomial node sizes and using the binomial lower tail
$\PP\bigl(m^{(i)}<q^{(i)}m/2\bigr)\le\exp(-q^{(i)}m/8)$ yields
$\PP(\mathcal E^c)\le N\exp(-c_\star\,q_{\min}m/2)+N\exp(-q_{\min}m/8)$
with $q_{\min}:=\min_iq^{(i)}$, which is exponentially small whenever
$\log N=o(q_{\min}m)$. The full
argument is in Appendix~\ref{app:wc-bias-body}.

\subsection{A matching lower bound under the global null}
\label{sec:gap-lower}

The selection bias is not an artifact of a loose union bound: under the
global null the order-statistics optimism is genuinely of the order that
Proposition~\ref{prop:wc-bias-body} permits.

\begin{proposition}[Tightness of the winner's-curse rate]\label{prop:wc-lower}
Let $N=1$ and $r_0^{(1)}=1$ (global null), fix $\lambda\in(0,1)$, and run
the Storey greedy of \S\ref{sec:gap-setup}. There exist constants
$c,C>0$ depending only on $\lambda$ such that, whenever
$\varepsilon\le\lambda/8$ and $\varepsilon m\ge C\log(\lambda/\varepsilon)$, the following holds
conditionally on the grid $\mathcal G$, on a grid event of probability at
least $1-2e^{-c m}$ on which $K:=K^{(1)}\ge\lambda/(4\varepsilon)$:
\begin{equation}\label{eq:wc-lower}
\Delta(M\mid\mathcal G)\;\ge\;c\,\sqrt{\frac{\log(K/M)}{\varepsilon m}}
\qquad\text{for all } C\log m\le M\le K/C .
\end{equation}
In particular, at the default bandwidth $\varepsilon=\alpha\,m^{-1/2}$,
\[
\Delta(M\mid\mathcal G)\;\ge\;c\,m^{-1/4}\sqrt{\log m}
\qquad\text{for all } C\log m\le M\le\sqrt K,
\]
matching the upper bound \eqref{eq:wc-Delta-body} up to constants. The
proof---a Gaussian lower tail for each binomial count (Slud's inequality,
\citealp{Slud1977}) combined with a negative-association Chernoff bound on
the number of exceedances \citep{JoagDevProschan1983,DubhashiRanjan1998}---is
in Appendix~\ref{app:wc-lower}.
\end{proposition}

Together, Propositions~\ref{prop:wc-bias-body} and \ref{prop:wc-lower}
show that at the default bandwidth, and for $M$ in the stated range, the
selection bias is at most of order $m^{-1/4}\sqrt{\log m}$ in general and
exactly of this order under the global null.

\paragraph{Two consequences.}
\emph{(i) The rate is strictly sub-parametric.} By
Proposition~\ref{prop:wc-lower} the upper bound is attained under the
global null, so the optimism is genuinely $\omega(m^{-1/2})$ --- an
asymptotic statement that by itself says nothing about any particular
finite $m$; empirically, the excess of the original greedy's FDR over
$\alpha$ does not shrink as the hub size grows in our experiments
(Section~\ref{sec:sim}).

\emph{(ii) The grid spacing $\varepsilon$ couples validity and power.}
As the grid is refined ($\varepsilon\downarrow 0$), both the lower
bound~\eqref{eq:wc-lower} --- through $1/(\varepsilon m)$ and
$\log(K/M)$, with $K\ge\lambda/(4\varepsilon)$ --- and the upper
bound~\eqref{eq:wc-Delta-body} grow at the matching rate, so the
guaranteed order of the fixed-$M$ optimism worsens. These are bounds at a
fixed number $M$ of selected cells: they do not show that the actual bias
is monotone in $\varepsilon$, nor that coarsening the grid reduces the FDR
at the algorithm's random stopping index, although it is consistent with
that intuition. Coarsening also has a power cost: with
$|\Aset|\asymp 1/\varepsilon$ candidate intervals, coarse cells dilute
localized signals and limit the spatial resolution of the rejection
region. The grid spacing is thus a single scalar through which the
original greedy ties validity to power; CFGA
(Section~\ref{sec:alt}) and BONuS-GA (Section~\ref{sec:bonus-main})
decouple validity from $\varepsilon$, freeing it to be tuned for power
alone.

\paragraph{Numerator-centric reading.} Via the bridge $\bar h_M=R(M)/(M\,m\,\varepsilon)$
(equation \eqref{eq:wc-num-view-intro} in the section opening), the same
bound on $\Delta(M\mid\mathcal G)$ quantifies the under-count of the
predicted null mass $M\,m\,\varepsilon$ relative to the actual null count
$\E[V(\widehat\Gamma_M)]$ inside the rejected region, up to a
conditional-mean factor: in the equal-mean global-null setting,
$\E[V(\widehat\Gamma_M)\mid\mathcal G]-M\,m\,\varepsilon
=M\,m\,\varepsilon\,\{\Delta(M\mid\mathcal G)+h_M-1\}$, where $h_M$
denotes the ratio of the selected cells' conditional-mean mass to the
predicted mass $m\varepsilon$. The term $h_M-1$ is a
null-rate--estimation discrepancy distinct from selection optimism, and
the identity makes plain that both are \emph{relative} statements: the raw
count discrepancy carries the multiplier $M\,m\,\varepsilon$. Whether
$h_M-1$ is negligible relative to $\Delta(M\mid\mathcal G)$ is a
null-rate--estimation question that the propositions do not address; when
it is, the relative under-count is of the same
$O(m^{-1/4}\sqrt{\log m})$ order.

\paragraph{How CFGA breaks the curse.}
The maximization over $\binom{|\Aset|}{M}$ subsets in the proof of
Prop.~\ref{prop:wc-bias-body} is the entire source of the $\sqrt{\log|\Aset|}$
inflation; everything else in the analysis would be a benign $O((\varepsilon
m)^{-1/2})$ Monte Carlo fluctuation. Sample-splitting (CFGA,
Section~\ref{sec:alt}) freezes the ranking on $\D_1$ before any $\D_2$
counts are revealed, so the maximization disappears: each $\D_2$ count that
enters $\widehat\FDP$ is a single \emph{unselected} Binomial deviation, a
$O((\varepsilon m)^{-1/2})$ relative fluctuation ($m^{-1/4}$ at the
default bandwidth)---splitting removes the selection effect, not the
shrinking-bin sampling variability, which the martingale argument of
Section~\ref{sec:alt} handles exactly rather than by concentration.
The cost is a constant-factor sample-splitting penalty,
which BONuS-GA (Section~\ref{sec:bonus-main}) in turn eliminates by replacing the
held-out half with a masking-measurable synthetic bag inside a backward
filtration: the synthetic counts $\widetilde N$ play the role of $\D_2$ and
the bias is again controlled by exchangeability rather than by independence.

\subsection{Proof of Proposition~\ref{prop:wc-bias-body}}
\label{app:wc-bias-body}

All expectations below are conditional on $\mathcal G$; on $\mathcal E$
the implicit constants depend only on $\sup_i\|f^{(i)}\|_\infty$,
$r_{\min}$, and $\lambda$.

\emph{Step 1: reduce the bias to a worst-subset deviation.}
For each fixed subset $S\subseteq\Aset$ with $|S|=M$ define the centered
subset sum
$U_S:=\sum_{(i,j)\in S}\bigl(\widehat h_j^{(i)}-h_j^{(i)}\bigr)$;
by \eqref{eq:body-wc-h-pop} (with
$\mathcal G=\sigma(\{(m^{(i)},N_>^{(i)})\})$ from \S\ref{sec:gap-setup};
the grid is $\mathcal G$-measurable) this
equals $\sum_S(\widehat h_j^{(i)}-\E[\widehat h_j^{(i)}\mid\mathcal G])$
\emph{exactly}, with no $o(1)$ slippage. Because $\widehat S(M)$ maximizes
$\sum_S \widehat h_j^{(i)}$ over $|S|=M$,
\begin{equation}
\bar h_M-\bar h_M^{\mathrm{sel}}\;=\;\frac{1}{M}\,U_{\widehat S(M)}
\;\le\;\frac{1}{M}\max_{|S|=M}U_S ,
\label{eq:wc-bias-step}
\end{equation}
where
$\bar h_M^{\mathrm{sel}}:=M^{-1}\sum_{(i,j)\in\widehat S(M)}h_j^{(i)}$.

\emph{Step 2: each $U_S$ is sub-gamma with parameters $(Mv,c)$.}
Disjointness of the within-node grid lets us collapse each $S$ to a
union-of-cells problem at every node: with
$U_i(S):=\bigcup_{j:\,(i,j)\in S}\widehat I_j^{(i)}\subseteq(0,\lambda]$,
\begin{equation*}
\sum_{(i,j)\in S}\widehat h_j^{(i)}
=(\varepsilon m)^{-1}\sum_i\sum_{k:\xi_k=0}\1\bigl\{P_k^{(i)}\in U_i(S)\bigr\},
\end{equation*}
a sum of independent $[0,(\varepsilon m)^{-1}]$-valued indicators. Each
cell contributes null mass at most $\widetilde p_j^{(i)}$, so by
additivity over $S$
\begin{equation*}
\begin{aligned}
\operatorname{Var}(U_S\mid\mathcal G)
&\le(\varepsilon m)^{-2}\!\!\sum_{(i,j)\in S}\!\!(m^{(i)}-N_>^{(i)})\widetilde p_j^{(i)}
\;\le\;Mv\qquad\text{on }\mathcal E,
\end{aligned}
\end{equation*}
with $v$ as in \eqref{eq:body-wc-subgamma} (each summand is bounded as in
\eqref{eq:body-wc-var}), and Bernstein (as in \S\ref{sec:gap-conc}) gives
$U_S\mid\mathcal G\in\Gamma_+(Mv,c)$ on $\mathcal E$ with
$c=1/(3\varepsilon m)$.

\emph{Step 3: take the maximum over $\binom{|\Aset|}{M}\le(e|\Aset|/M)^M$
subsets.} The displayed binomial bound is the elementary Stirling-type
inequality $\binom{n}{k}\le(en/k)^k$ (from $k!\ge(k/e)^k$). We use the
sub-gamma maximal inequality \citep[Cor.~2.6]{BoucheronLugosiMassart2013}:
if $Z_1,\dots,Z_K$ are each sub-gamma with parameters $(v',c')$ then
$\E[\max_\ell Z_\ell]\le\sqrt{2v'\log K}+c'\log K$. Applied with
$(v',c')=(Mv,c)$ and $K=\binom{|\Aset|}{M}\le(e|\Aset|/M)^M$,
\begin{equation*}
\E\!\bigl[\max_{|S|=M}U_S\bigm|\mathcal G\bigr]
\le M\sqrt{2v\log(e|\Aset|/M)}+cM\log(e|\Aset|/M) .
\end{equation*}
Dividing by $M$ and combining with \eqref{eq:wc-bias-step} gives
\eqref{eq:wc-Delta-body} on $\mathcal E$ with $C_1,C_2>0$ depending only on
$\sup_i\|f^{(i)}\|_\infty$, $r_{\min}$, and $\lambda$: explicitly,
$C_1=2\max_i\|f^{(i)}\|_\infty/(r_{\min}^2\lambda)$ and $C_2=1/3$, from
$v=\max_i\|f^{(i)}\|_\infty/(r_{\min}^2\lambda\,\varepsilon m)$ and
$c=1/(3\varepsilon m)$ above.

\emph{Step 4: pass to the unconditional bound.}
It remains to estimate the complement event
$\mathcal E^c=\{\inf_i\widehat r_0^{(i)}<r_{\min}\}$. Fix node $i$ and let
$X_k^{(i)}:=\1\{P_k^{(i)}>\lambda\}$. Under
Assumption~\ref{ass:model}, $(X_k^{(i)})_{k=1}^{m^{(i)}}$ are i.i.d.\
Bernoulli with mean
\begin{equation*}
\mu_i:=\E X_k^{(i)}
=r_0^{(i)}(1-\lambda)+r_1^{(i)}\!\!\int_\lambda^1\!\!f_i
\;\ge\;r_0^{(i)}(1-\lambda),
\end{equation*}
because non-nulls only add probability mass on $(\lambda,1]$. The
truncated Storey estimator
$\widehat r_0^{(i)}=\min\bigl\{1,\bigl(1+\sum_k X_k^{(i)}\bigr)/\bigl((1-\lambda)m^{(i)}\bigr)\bigr\}$
is a clipped $+1$-offset version of the sample mean
$\bar X^{(i)}:=m^{(i)-1}\sum_k X_k^{(i)}$; since $r_{\min}<1$, the clip
cannot cause $\widehat r_0^{(i)}<r_{\min}$, and the deviation event
rewrites as
\begin{equation*}
\bigl\{\widehat r_0^{(i)}<r_{\min}\bigr\}
=\bigl\{\bar X^{(i)}<(1-\lambda)r_{\min}-1/m^{(i)}\bigr\}.
\end{equation*}
Fix $r_{\min}\in(0,\inf_i r_0^{(i)})$. Then $\mu_i-(1-\lambda)r_{\min}\ge(1-\lambda)(r_0^{(i)}-r_{\min})>0$, and Hoeffding's one-sided
inequality for the average of $m^{(i)}$ independent $[0,1]$-valued summands
\citep[Thm.~2.8]{BoucheronLugosiMassart2013} yields
\begin{equation*}
\PP\bigl(\bar X^{(i)}\le\mu_i-t\bigr)\le e^{-2m^{(i)}t^2},\qquad t>0.
\end{equation*}
Applying this with $t=(1-\lambda)(\inf_i r_0^{(i)}-r_{\min})$ (positive
by choice of $r_{\min}$; the additional $1/m^{(i)}$ slack from the
$+1$-offset only tightens the bound) gives
\begin{equation*}
\PP\bigl(\widehat r_0^{(i)}<r_{\min}\,\big|\,m^{(i)}\bigr)
\le\exp\!\bigl(-2m^{(i)}(\inf_i r_0^{(i)}-r_{\min})^2(1-\lambda)^2\bigr) .
\end{equation*}
Union-bounding across the $N$ nodes, conditionally on the node sizes,
\begin{equation*}
\PP\bigl(\mathcal E^c\,\big|\,\{m^{(i)}\}\bigr)\le N\exp\bigl(-c_\star\min_i m^{(i)}\bigr),\qquad
c_\star:=2(\inf_i r_0^{(i)}-r_{\min})^2(1-\lambda)^2 .
\end{equation*}
Averaging over the multinomial node sizes and using the binomial lower
tail $\PP(m^{(i)}<q^{(i)}m/2)\le\exp(-q^{(i)}m/8)$,
\begin{equation*}
\PP(\mathcal E^c)\le N\exp\bigl(-c_\star\,q_{\min}m/2\bigr)
+N\exp\bigl(-q_{\min}m/8\bigr),\qquad q_{\min}:=\min_i q^{(i)},
\end{equation*}
which is exponentially small whenever $\log N=o(q_{\min}m)$.
\hfill$\square$

\subsection{Proof of Proposition~\ref{prop:wc-lower}}
\label{app:wc-lower}

\paragraph{Strategy.} The argument has three pieces. \emph{(i)}~Each
multinomial cell carries a Gaussian-strength lower tail (Slud's
inequality, \citealp{Slud1977}, quoted below as Theorem~\ref{thm:slud}).
\emph{(ii)}~The exceedance indicators of the $K$ cells, although
dependent through the multinomial conservation $\sum_j N(\widehat I_j)\le n_h$,
form a negatively associated family
\citep{JoagDevProschan1983}, and the
Chernoff--Hoeffding lower-tail bound still applies to their sum
\citep{DubhashiRanjan1998}; so the number of
exceedances concentrates around its mean $Kq$, where $q$ is the
per-cell exceedance probability from~(i).
\emph{(iii)}~On the event that at least $M$ cells exceed the threshold
$\tau$, the average of the top-$M$ counts is itself bounded below by
$\tau$, which translates the exceedance count into the claimed bias.

Throughout, condition on the grid $\mathcal G$ (equivalently, on the tail
count $N_>:=N_>^{(1)}$; $N=1$ so we drop node superscripts), and write
$n_h=m-N_>$ for the head count, $\widehat L=\widehat L^{(1)}$,
$K=K^{(1)}=\lfloor\lambda/\widehat L\rfloor$.

\paragraph{Grid event.} Under the global null
$N_>\sim\Bin(m,1-\lambda)$, so by Hoeffding the event
$\mathcal E_g=\{|N_>-(1-\lambda)m|\le \tfrac12\min(\lambda,1-\lambda)\,m\}$
has probability at least $1-2e^{-m\min(\lambda,1-\lambda)^2/2}$. Set
$\mu_*:=\tfrac12\min(\lambda,1-\lambda)\,m$. We claim that on $\mathcal E_g$,
the head count, the truncated Storey estimate, the grid spacing, and the
candidate count all sit in narrow ranges.

\smallskip
\emph{(i)~Head count $n_h\ge\lambda m/2$.}
The upper deviation bound is $N_>\le(1-\lambda)m+\mu_*$, so
$n_h=m-N_>\ge\lambda m-\mu_*$. We check the two regimes:
\begin{itemize}
\item If $\lambda\le\tfrac12$, then $\mu_*=\lambda m/2$ and
$n_h\ge\lambda m-\lambda m/2=\lambda m/2$.
\item If $\lambda>\tfrac12$, then $\mu_*=(1-\lambda)m/2\le\lambda m/2$ and
$n_h\ge\lambda m-(1-\lambda)m/2=(3\lambda-1)m/2\ge\lambda m/2$
(since $3\lambda-1\ge\lambda$ iff $\lambda\ge\tfrac12$).
\end{itemize}

\emph{(ii)~Storey estimate $\widehat r_0\in[\tfrac12,1]$.} The lower
deviation bound is $N_>\ge(1-\lambda)m-\mu_*$, so
$\widehat r_0=\frac{1+N_>}{(1-\lambda)m}\ge\frac{N_>}{(1-\lambda)m}
\ge1-\frac{\mu_*}{(1-\lambda)m}$. Substituting $\mu_*$ in the same two
regimes shows the lower bound is at least $\tfrac12$ in each:
$\lambda\le\tfrac12$ gives $\frac{N_>}{(1-\lambda)m}\ge\frac{1-3\lambda/2}{1-\lambda}\ge\tfrac12$;
$\lambda>\tfrac12$ gives $\frac{N_>}{(1-\lambda)m}\ge\tfrac12$ directly.
Truncation to $[0,1]$ then forces $\widehat r_0\in[\tfrac12,1]$.

\emph{(iii)~Grid spacing and candidate count.} From (ii),
$\widehat L=\varepsilon/\widehat r_0\in[\varepsilon,2\varepsilon]$, and
hence $K=\lfloor\lambda/\widehat L\rfloor\ge\lambda/\widehat L-1
\ge\lambda/(2\varepsilon)-1\ge\lambda/(4\varepsilon)$
once $\lambda/(4\varepsilon)\ge 1$, which holds since
$\varepsilon\le\lambda/8$ by the hypothesis of the proposition. The role of $\mathcal E_g$ is to give us a
deterministic grid: on it, $\widehat L$ is bounded above and below by
constants times $\varepsilon$, $K$ is at least $\Theta(1/\varepsilon)$, and
the head count $n_h$ is at least $\lambda m/2$, allowing us to handle
$\widetilde p$, $\sigma$, and $K$ as if they were deterministic.

\paragraph{Equal means.} By the head/tail decomposition in the proof of
Lemma~\ref{lem:binomial}, conditional on $\mathcal G$ the head p-values are
$n_h$ i.i.d.\ draws from $\mathrm{Uniform}[0,\lambda]$, so the interval
counts $(N(\widehat I_1),\dots,N(\widehat I_K))$ are jointly multinomial
cells with the \emph{common} cell probability $\widetilde p=\widehat L/\lambda$
and marginals $\Bin(n_h,\widetilde p)$ (the remaining mass
$1-K\widetilde p$ is carried by a ``remainder'' cell collecting head
p-values that fall outside the $K$ candidate intervals). In particular the
conditional means \eqref{eq:body-wc-h-pop} are equal across $j$,
$h_j\equiv h=n_h\widetilde p/(\varepsilon m)\le 2/\lambda=:h_{\max}$, and
since the greedy's selected set is exactly the top-$M$ counts,
$\Delta(M\mid\mathcal G)=\E[\bar h_M\mid\mathcal G]-h$ with $\bar h_M$ the
average of the $M$ largest $\widehat h_j$.

\paragraph{Step 1: Gaussian lower tail for each cell.}
The marginal $N(\widehat I_j)\sim\Bin(n_h,\widetilde p)$ admits a
\emph{lower} tail at the normal rate by Slud's inequality. We quote the
result we need:

\begin{theorem}[{\citet[Thm.~2.1]{Slud1977}}]\label{thm:slud}
Let $X\sim\Bin(n,p)$ with $0<p\le 1/4$. Then for every integer $k$ with
$np\le k\le n(1-p)$,
\begin{equation*}
\PP(X\ge k)\;\ge\;1-\Phi\!\Bigl(\tfrac{k-np}{\sqrt{np(1-p)}}\Bigr) ,
\end{equation*}
where $\Phi$ is the standard normal cdf.
\end{theorem}

To apply this in our setting, observe first that the proposition's
hypothesis $\varepsilon\le\lambda/8$, combined with the Grid-event bound
$\widehat L\le 2\varepsilon$, gives
$\widetilde p=\widehat L/\lambda\le 2\varepsilon/\lambda\le 1/4$, so
Slud's hypothesis $p\le 1/4$ is satisfied. Define
$\sigma^2:=n_h\widetilde p(1-\widetilde p)$, the variance of one cell.
Since $\widetilde p\le 1/4$, $\sigma^2\ge\tfrac34 n_h\widetilde p$, and the
lower bound $n_h\widetilde p\ge(\lambda m/2)(\varepsilon/\lambda)=\varepsilon m/2$
(valid on $\mathcal E_g$) gives
\begin{equation*}
\sigma^2\ge \tfrac{3}{8}\,\varepsilon m .
\end{equation*}
By Theorem~\ref{thm:slud}, for every integer $k$ with
$n_h\widetilde p\le k\le n_h(1-\widetilde p)$,
\begin{equation}\label{eq:slud-applied}
\PP\bigl(N(\widehat I_j)\ge k\bigr)\;\ge\;1-\Phi\Bigl(\tfrac{k-n_h\widetilde p}{\sigma}\Bigr).
\end{equation}

\paragraph{Step 2: Concentration of the number of exceedances.}
Choose
\begin{equation*}
t:=\sqrt{\log(K/M)},\qquad \tau:=n_h\widetilde p+t\sigma ,
\end{equation*}
and introduce the exceedance indicators
\begin{equation*}
\xi_j:=\1\{N(\widehat I_j)\ge\tau\},\qquad j=1,\dots,K .
\end{equation*}
Two preliminary verifications underwrite the use of Slud's inequality at
this $\tau$.

\smallskip
\emph{(I) $t\ge 1$.} The hypothesis $M\le K/C$ of the proposition, with
$C\ge e$, gives $K/M\ge e$ and hence $t\ge 1$. This is the range required
by the Mills-ratio bound used in substep~(a) below; we absorb the
requirement $C\ge e$ into the proposition's constant.

\emph{(II) Slud's range $n_h\widetilde p\le\lceil\tau\rceil\le n_h(1-\widetilde p)$.}
Cell counts are integer-valued, so
$\xi_j=\1\{N(\widehat I_j)\ge\tau\}=\1\{N(\widehat I_j)\ge\lceil\tau\rceil\}$
and we will apply Slud's inequality at the integer $k=\lceil\tau\rceil$.
Slud requires $n_h\widetilde p\le\lceil\tau\rceil\le n_h(1-\widetilde p)$;
the lower bound is automatic, and the upper bound
$\lceil\tau\rceil\le n_h(1-\widetilde p)$ is implied, via
$\lceil\tau\rceil\le\tau+1$, by
\begin{equation}\label{eq:slud-range-target}
t\sigma+1\;\le\;n_h(1-2\widetilde p).
\end{equation}

We bound the RHS of~\eqref{eq:slud-range-target} from below and the LHS
from above using the Grid-event bounds
$\widetilde p\le 2\varepsilon/\lambda\le 1/4$ (so $1-2\widetilde p\ge\tfrac12$),
$n_h\ge\lambda m/2$, $\log(K/M)\le\log K\le\log(\lambda/\varepsilon)$, and
$\sigma^2=n_h\widetilde p(1-\widetilde p)\le n_h\widetilde p\le m\cdot(2\varepsilon/\lambda)$:
\begin{align*}
n_h(1-2\widetilde p)&\ge\tfrac{\lambda m}{2}\cdot\tfrac12\;=\;\tfrac{\lambda m}{4},\\
t^2\sigma^2&\le\log(\lambda/\varepsilon)\cdot\tfrac{2\varepsilon m}{\lambda}.
\end{align*}
Splitting the RHS budget evenly between $t\sigma$ and the $+1$, it
therefore suffices that
\begin{equation*}
t\sigma\le\tfrac{\lambda m}{8}\qquad\text{and}\qquad 1\le\tfrac{\lambda m}{8}.
\end{equation*}
Squaring the first and substituting $t^2\sigma^2\le(2/\lambda)\varepsilon m\log(\lambda/\varepsilon)$:
\begin{equation*}
\tfrac{2\varepsilon m\log(\lambda/\varepsilon)}{\lambda}\;\le\;\tfrac{\lambda^2 m^2}{64}
\quad\Longleftrightarrow\quad
\varepsilon\log(\lambda/\varepsilon)\;\le\;\tfrac{\lambda^3}{128}\,m .
\end{equation*}
This is the condition $\varepsilon\log(\lambda/\varepsilon)\le c_\lambda m$
with $c_\lambda=\lambda^3/128$. Both this and the second ($m\ge 8/\lambda$)
are implied by the hypothesis $\varepsilon m\ge C$ for a sufficiently large
$\lambda$-dependent $C$ (recalling $\varepsilon\le 1$, so $m\ge\varepsilon m\ge C$).
We absorb $C$ into the proposition's constant.

\smallskip
Slud's inequality~\eqref{eq:slud-applied} at $k=\lceil\tau\rceil$ now
gives
\begin{equation*}
\PP(\xi_j=1)\;\ge\;1-\Phi\!\bigl((\lceil\tau\rceil-n_h\widetilde p)/\sigma\bigr)\;\ge\;1-\Phi(t+1/\sigma),
\end{equation*}
the second inequality using $\lceil\tau\rceil\le\tau+1$. We bound the
loss from the $1/\sigma$ shift as follows. By the mean-value theorem,
\begin{equation*}
1-\Phi(t+1/\sigma)\;\ge\;\bigl(1-\Phi(t)\bigr)-\phi(t)/\sigma .
\end{equation*}
The Mills-ratio bound $1-\Phi(t)\ge\phi(t)/(2t)$ (valid for $t\ge 1$,
used again in substep~(a) below) lets us rewrite the loss as
$\phi(t)/\sigma\le\tfrac{2t}{\sigma}\bigl(1-\Phi(t)\bigr)$. Hence, as
long as $\sigma\ge 4t$,
\begin{equation*}
\PP(\xi_j=1)\;\ge\;\tfrac12\bigl(1-\Phi(t)\bigr) ,
\end{equation*}
and the integer rounding costs at most a factor of $\tfrac12$. The
condition $\sigma\ge 4t$ rearranges to
$\varepsilon m\ge\tfrac{128}{3}\log(K/M)$. Since $K\le\lambda/\varepsilon$
and $M\ge1$, this is guaranteed by the proposition's hypothesis
$\varepsilon m\ge C\log(\lambda/\varepsilon)$ once $C\ge\tfrac{128}3$.
(At the default bandwidth $\varepsilon=\alpha m^{-1/2}$ the hypothesis
holds for all large $m$: $\varepsilon m=\alpha\sqrt m$ dominates
$\log(\lambda/\varepsilon)=O(\log m)$.) We absorb the
factor $\tfrac12$ into the proposition's constant, and in the sequel use
$\PP(\xi_j=1)\ge\tfrac12\bigl(1-\Phi(t)\bigr)$.

\emph{(a) Lower bound on the single-cell exceedance probability.}
The Mills-ratio bound $1-\Phi(t)\ge\frac{t}{1+t^2}\,\phi(t)
\ge\frac{\phi(t)}{2t}$ (valid for $t\ge 1$), combined with the
rounding-adjusted bound $\PP(\xi_j=1)\ge\tfrac12(1-\Phi(t))$ established
above, gives
\begin{equation}\label{eq:single-cell-q}
q\;:=\;\PP(\xi_j=1)
\;\ge\;\frac{e^{-t^2/2}}{4\sqrt{2\pi}\,t}
\;=\;\frac{(M/K)^{1/2}}{4\sqrt{2\pi}\,t} .
\end{equation}
A second algebraic step shows that~\eqref{eq:single-cell-q} implies
$q\ge 2M/K$ once $K/M$ is large enough: rearranging,
$q\ge 2M/K$ is equivalent to $(K/M)^{1/2}\ge 8\sqrt{2\pi}\,\sqrt{\log(K/M)}$,
which holds whenever $K/M\ge C_0$ for an absolute constant $C_0$. We
absorb $C_0$ into the constant $C$ appearing in the hypothesis $M\le K/C$.
Hence
\begin{equation}\label{eq:cell-mean}
\E\sum_j\xi_j\;=\;Kq\;\ge\;2M .
\end{equation}

\emph{(b) Negative association of the exceedance indicators.}
The cell counts $(N(\widehat I_1),\dots,N(\widehat I_K))$ are
negatively associated:

\begin{theorem}[{\citet{JoagDevProschan1983}, Def.~2.1 and \S3.1(a)}]\label{thm:multinomial-NA}
A multinomial random vector is negatively associated (NA), i.e.\ for every
pair of disjoint index sets $I,J$ and every increasing
$f:\mathbb R^{|I|}\to\mathbb R$, $g:\mathbb R^{|J|}\to\mathbb R$,
\begin{equation*}
\operatorname{Cov}\!\bigl(f(B_i,\,i\in I),\,g(B_j,\,j\in J)\bigr)\;\le\;0.
\end{equation*}
\end{theorem}

\begin{theorem}[NA closure {\citep[Prop.~$P_6$]{JoagDevProschan1983}}]\label{thm:NA-closure}
Increasing (or decreasing) functions of disjoint subsets of NA
random variables are themselves NA.
\end{theorem}

Each $\xi_j$ is an increasing function of the single cell $N(\widehat I_j)$,
and the singletons $\{j\}$ are pairwise disjoint, so
Theorems~\ref{thm:multinomial-NA}--\ref{thm:NA-closure} imply
$(\xi_1,\dots,\xi_K)$ is NA.

\emph{(c) Chernoff lower tail for the sum.}
The Chernoff--Hoeffding tail bounds extend to sums of negatively
associated $[0,1]$-valued summands, by repeating the standard MGF argument
and using the NA inequality
$\E\prod_j e^{tY_j}\le\prod_j\E e^{tY_j}$ in place of the product
formula valid under independence
\citep[Prop.~5]{DubhashiRanjan1998}. In particular, the
\emph{multiplicative} lower-tail bound

\begin{theorem}[NA Chernoff lower tail {\citep[Prop.~5]{DubhashiRanjan1998}}]\label{thm:NA-Chernoff}
If $Y_1,\dots,Y_K\in[0,1]$ are NA with $\mu:=\E\sum_j Y_j$, then for every
$\delta\in(0,1)$,
\begin{equation*}
\PP\Bigl(\sum_j Y_j\le(1-\delta)\mu\Bigr)\;\le\;e^{-\mu\delta^2/2}.
\end{equation*}
\end{theorem}

applies to $(\xi_1,\dots,\xi_K)$. With $\mu=Kq\ge 2M$ from \eqref{eq:cell-mean}
and $\delta=1/2$,
\begin{equation}\label{eq:exceedance-tail}
\PP\Bigl(\sum_j\xi_j<M\Bigr)
\le\PP\Bigl(\sum_j\xi_j\le\tfrac12 Kq\Bigr)
\le e^{-Kq/8}\;\le\;e^{-M/4}.
\end{equation}

\paragraph{Step 3: assemble.}
On the event $\mathcal G_{\rm exc}:=\{\sum_j\xi_j\ge M\}$, at least $M$
cells have count above $\tau$, so the top-$M$ counts each exceed $\tau$
and the average
$\bar h_M\ge\tau/(\varepsilon m)=h+t\sigma/(\varepsilon m)$. On its
complement we use the trivial lower bound $\bar h_M\ge 0$, together with
$h\le h_{\max}=2/\lambda$. Therefore
\begin{equation}\label{eq:Delta-assemble}
\begin{aligned}
\Delta(M\mid\mathcal G)
&\ge\frac{t\sigma}{\varepsilon m}\bigl(1-e^{-M/4}\bigr)-h_{\max}e^{-M/4}\\
&\ge\frac{t\sigma}{2\varepsilon m}
\ge c_\lambda\sqrt{\tfrac{\log(K/M)}{\varepsilon m}},
\end{aligned}
\end{equation}
where, writing $a:=t\sigma/(\varepsilon m)$ and $x:=e^{-M/4}$, the middle
inequality $a(1-x)-h_{\max}x\ge a/2$ holds as soon as $x\le\tfrac14$ and
$h_{\max}x\le a/4$ (then $ax+h_{\max}x\le a/2$). The first needs only
$M\ge4\log4$; for the second, $a\ge\sqrt{3/8}\,\sqrt{\log(K/M)/(\varepsilon m)}
\ge c\,(\varepsilon m)^{-1/2}$ and $h_{\max}x\le(2/\lambda)e^{-M/4}$, so it
holds once $M\ge C_\lambda\log(\varepsilon m)$ with a suitably larger
constant $C_\lambda$, itself implied by $M\ge C\log m$ for $C$ large. The
rate is unchanged. This is the bound~\eqref{eq:wc-lower}.

For the display at $\varepsilon=\alpha m^{-1/2}$:
$K\ge\lambda/(4\varepsilon)=\lambda\sqrt m/(4\alpha)$, so for $M\le\sqrt K$
we have $\log(K/M)\ge\tfrac12\log K\ge c\log m$, and~\eqref{eq:Delta-assemble}
gives $\Delta(M\mid\mathcal G)\ge c\,m^{-1/4}\sqrt{\log m}$.\hfill$\square$


\section{Adaptive bandwidth: practical details}
\label{app:eps}

\paragraph{Choice of grid.} A logarithmic grid spanning
$\varepsilon\in[\alpha\,m^{-1/2}/16,\;4\alpha\,m^{-1/2}]$ (powers of two
around the default) suffices in our experiments. The grid width does not
affect FDR validity; only computational cost is linear in $|\mathcal B|$.

\paragraph{Cross-fit version.} For CFGA (Algorithm~\ref{alg:cfga}),
apply Algorithm~\ref{alg:adaptive-eps} \emph{independently} in each of the
two folds (the inner CV uses the fold's ranking half). This yields two
adaptively chosen $\varepsilon^{(1)},\varepsilon^{(2)}$ that may differ.
The proof of Theorem~\ref{thm:cfga} (union bound on the two
per-fold FDPs) is unchanged.

\paragraph{Caveat under within-node dependence.}
Proposition~\ref{prop:adaptive-FDR} assumes within-node independence
(Assumption~\ref{ass:model}). The $\sigma(\D_1,U)$-measurability argument
preserves validity for any bandwidth chosen from the ranking half and
independent auxiliary randomness, but it does
not protect against violations of Assumption~\ref{ass:model} itself. Under
strong within-node positive dependence the Storey null-rate estimate
becomes anti-conservative, and because Algorithm~\ref{alg:adaptive-eps}
selects the $\varepsilon$ that \emph{maximizes} the validation rejection
count, it preferentially picks the value that most fully exploits this
optimism. Experiment~6 (Section~\ref{sec:sim-rho}) makes this concrete: under
within-node AR(1) correlation the Storey$+$adaptive-$\varepsilon$ variant's
FDR climbs to $0.32$ at $\rho=0.8$, whereas the fixed-$\varepsilon$ Storey
variant---which does not tune $\varepsilon$ to the data---and the
\emph{Inflated} adaptive variant (FDR $\approx0.04$ across the whole range)
both stay controlled. Confining the correlation to the hubs removes the
runaway entirely (FDR $0.087$--$0.094$ across $\rho\le0.8$;
Section~\ref{sec:sim-rho}), consistent with this mechanism: it is the
dependent quiet-node mass, whose Storey and validation statistics the
adaptive rule trusts, that fuels the over-rejection. We therefore recommend a fixed
$\varepsilon$, or the Inflated adaptive variant, whenever strong
within-node dependence is suspected.

\paragraph{Computational cost.}

Algorithm~\ref{alg:adaptive-eps} runs the inner greedy $T$ times per fold,
plus the final outer run.
Each invocation costs $O(\sqrt m\log m)$ communication and $O(m)$ local
computation; the inner cross-validation step adds no communication beyond
the inner greedy runs, except that the inner grids need their sub-fold
totals, which are included in the preliminary size exchange (two more
integers per node per fold). In rounds, the inner runs and the selection
of $\varepsilon$ form one extra uplink--downlink exchange that must
precede the upload of the final outer-grid counts --- unless each node
uploads its outer-grid counts for all $T$ candidates in advance, which
removes the extra round at $T$ times the outer-grid bits. Per fold, the
adaptive rule thus runs $T$ inner greedy passes plus the final outer pass,
against one pass without it; since a pass at bandwidth $\varepsilon_t$
touches $O(1/\varepsilon_t)$ cells, the counting passes cost $O(T)$ times
those of non-adaptive CFGA, with a constant set by the ratio of the default
bandwidth to the smallest candidate ($16$ on our grid). Including the
$O(N\log m)$-bit size exchange of Section~\ref{sec:method}, the total
adaptive communication on the stated grid is
$O\bigl((N+T\sqrt m)\log m\bigr)$ bits, which reduces to
$O(\sqrt m\log m)$ for fixed $T$ and $N=O(\sqrt m)$. The trade-off (Section~\ref{sec:sim},
Experiments~2 and~5) is a substantial power gain, especially at moderate $m$.

\section{Communication protocols}\label{app:protocols}

This appendix gives the message formats, round structure and bit counts
summarized in Table~\ref{tab:methods}. None of it affects the FDR
guarantees, which hold for the rejection rules as stated in the main text.

\subsection{CFGA}\label{app:protocols-cfga}
Step~1 is local at each node: $\D_1^{(i)}$ and $\D_2^{(i)}$ are drawn from
node $i$'s own randomness (no inter-node coordination), and the split itself
is not transmitted. The remaining steps are a short \emph{size exchange}
followed by a single \emph{uplink--downlink round}, in which each node
sends a batch of counts to the center and the center then returns the
rejection decisions.

\emph{Size exchange.} Node $i$'s grid uses the ranking-fold proportion
$\widehat q_b^{(i)}=m_b^{(i)}/m_b$, but after independent local Bernoulli
splitting a node knows its own $m_b^{(i)}$ and not the global fold total
$m_b$. Each node therefore first uploads its two fold sizes
$(m_1^{(i)},m_2^{(i)})$, and the center broadcasts $(m_1,m_2)$:
$O(N\log m)$ uplink and $O(\log m)$ downlink bits. (The original greedy
likewise requires every node to know the global total $m$.)

\emph{Uplink.} For each fold $b\in\{1,2\}$, node $i$ enumerates the
$K_i^{(b)}=\lfloor\lambda/\widehat L^{(i,b)}\rfloor$ candidate intervals
$\widehat I_j^{(i,b)}\subseteq[0,\lambda]$ on its Storey grid and sends
one tuple per interval to the center,
\begin{equation*}
\bigl(\,j,\ N_b^{(i)}(\widehat I_j^{(i,b)}),\ N_{b'}^{(i)}(\widehat I_j^{(i,b)})\,\bigr),
\end{equation*}
i.e.\ the interval index, the ranking-half count (used by the center to
rank intervals by estimated density), and the inference-half count (used
to evaluate the held-out FDP). The node identifier is implicit in the
message routing. Each integer is in $[0,m^{(i)}]$, so each tuple costs
$O(\log m)$ bits, and node $i$'s total uplink for both folds is
$(K_i^{(1)}+K_i^{(2)})\cdot O(\log m)$ bits.

\emph{Aggregation.} For each fold $b$ the center sorts the pooled tuples
by ranking-half count, forms the nested family
$\{\widehat\Gamma_M^{(i,b)}\}_{M=1}^{M_{\max}}$, evaluates
$\widehat\FDP_{b'}(M)=(1+m_{b'}\varepsilon M)/(R_{b'}(M)\vee 1)$ along
the ranking, and sets
$\widehat M^{(b)}=\max\{M:\widehat\FDP_{b'}(M)\le\alpha/2\}$. All of this
is computed at the center, and so is the union of the two folds'
accepted intervals.

\emph{Downlink.} The center returns to each node the list of indices $j$
of its accepted intervals (one index per acceptance, $O(\log m)$ bits
each). The node identifies its rejected p-values locally as those
falling in the assigned intervals.

\smallskip
\emph{Round count.} CFGA takes two synchronous rounds: the size exchange,
then one uplink message followed by one downlink message, with both folds
carried in the same round.

\emph{Total bits.} Across nodes the uplink costs
$\sum_i(K_i^{(1)}+K_i^{(2)})\cdot O(\log m)$; since
$\sum_i K_i^{(b)}=O(1/\varepsilon)$ for each fold $b$, at the standard
bandwidth $\varepsilon\asymp m^{-1/2}$ this is $O(\sqrt m\log m)$ bits.
The downlink sends one $O(\log m)$-bit interval index per accepted
interval. For each fold the stopping index cannot exceed the number
of candidate intervals,
$\widehat M^{(b)}\le\sum_i K_i^{(b)}=O(1/\varepsilon)$, so the total
number of accepted intervals across both folds is also $O(\sqrt m)$
and the downlink cost is $O(\sqrt m\log m)$ bits. With the size exchange
the total cost is $O((N+\sqrt m)\log m)$ bits, i.e.\ the original
greedy's $O(\sqrt m\log m)$ budget whenever $N=O(\sqrt m)$.

\subsection{Inflated CFGA}\label{app:protocols-aug}
Inflated CFGA uses the same size
exchange and uplink--downlink round as CFGA; the only addition is one extra
\emph{bin}, the tail $A=(\lambda,1]$, per node per fold. Concretely,
node $i$ appends the tail count
$N_b^{(i)}(A)=\#\{k\in\D_b^{(i)}:P_k^{(i)}>\lambda\}$ to its uplink
batch for each fold $b\in\{1,2\}$ --- one extra integer in
$[0,m_b^{(i)}]$, costing $O(\log m)$ bits. The center then assembles
the inflated Storey null-rate estimate
\begin{equation*}
\widehat\pi_b^{(i),+}\;=\;\frac{N_b^{(i)}(A)}{(1-\lambda)m_b}+\delta
\end{equation*}
from the uploaded tail count and the inflation~\eqref{eq:delta} (which
depends only on $N$, $m$, $\lambda$ and the fold totals $m_b$ broadcast
in the size exchange), plugs
$\widehat\pi_b^{(i),+}$ into the predicted null count
$m_{b'}\widehat\pi_b^{(i),+}\nu(\widehat\Gamma_M^{(i,b)})$ used by the
FDP estimator, and otherwise runs CFGA's aggregation and downlink
unchanged.

\smallskip
\emph{Round count.} The same two rounds as CFGA, with both folds carried
together.

\emph{Total bits.} Inflated CFGA adds $2N\cdot O(\log m)=O(N\log m)$
bits over CFGA (one tail count per node per fold); since CFGA itself
sends $O(\sqrt m\log m)$ bits at $\varepsilon\asymp m^{-1/2}$, the
Inflated CFGA total is $O((N+\sqrt m)\log m)$ bits in the same two rounds,
dominated by the CFGA upload whenever $N=O(\sqrt m)$.

\subsection{e-CFGA}\label{app:protocols-ecfga}

The $S$ splits are independent, and each needs its own fold totals; the
$S$ size exchanges (and those of the inner sub-splits used by the
adaptive bandwidth) are batched into one preliminary round of
$O(NS\log m)$ bits. Each split then runs one CFGA-style uplink--downlink
round at the fold level $\gamma$ using the tuple/index protocol of
Section~\ref{sec:method}, at $O(\sqrt m\log m)$ bits per split, so the
$S$ splits together cost $O(S\sqrt m\log m)$ bits in $S$ rounds --- or in
one round if the splits are pipelined in parallel.

\paragraph{Can the center compute the scores on its own?}
Partly. Each per-split downlink can broadcast the two fold-level
scalars
\begin{equation*}
\eta_{s,b}=\frac{m^s_{b'}}{1+\bar V(\widehat M_{s,b})},\qquad b\in\{1,2\} ,
\end{equation*}
alongside the accepted-interval indices, so both nodes and the center
know these scalars; if they are transmitted in quantized form they must be
rounded \emph{downward}, for the same reason as the scores below. But the per-hypothesis score~\eqref{eq:efold-evalue}
also depends on the identity of the fold each hypothesis lands in and
on whether that hypothesis's p-value falls in the fold's rejection
region --- data known only to the hosting node. The center therefore
sees per-interval counts, not per-hypothesis rejection bits, and cannot
reconstruct $\bar e_k$ without one more message.

\paragraph{One extra round closes the loop.}
Each node compiles the per-hypothesis average $\bar e_k$ locally (it
knows its own p-values, its own fold assignments across the $S$
splits, and the broadcast $\eta_{s,b}$), and sends the nonzero
$\bar e_k$'s to the center in a single synchronous round. The center
pools the values, sorts, and applies e-BH at level $\alpha$. The
hypotheses with $\bar e_k>0$ are exactly the p-values falling in some
fold's rejected region, so their number is $\bar R\le m$, in practice of
the order of the rejection count; the extra round transmits
$O(\bar R\log m)$ bits. Total budget:
$O\bigl((NS+S\sqrt m+\bar R)\log m\bigr)$ bits in $O(S)$ rounds (or in
three rounds pipelined: size exchange, split round, reporting round, for
a fixed bandwidth; with the adaptive bandwidth, selecting $\varepsilon$
needs one inner exchange before the final outer-grid counts are built,
adding a round, unless every node uploads its outer-grid counts for all
$T$ candidate bandwidths in advance at $T$ times the bits),
which matches the split rounds' order whenever $N=O(\sqrt m)$ and
$\bar R=O(S\sqrt m)$, and is at most $O((NS+S\sqrt m+m)\log m)$ in the
worst case.

\paragraph{Effect of quantization.}
Two kinds of numbers are transmitted in finite precision: the broadcast
fold scalars $\eta_{s,b}$ and the final averaged scores $\bar e_k$. A
complete choice is to round \emph{both downward} to the grid of spacing
$1/(2m)$. Each score $e_k^{(s,b)}$ is then lowered by less than $1/(2m)$,
so its split average $\bar e_k$ is lowered by less than $1/(2m)$ (each
hypothesis sits in exactly one inference half per split), and the final
rounding lowers it by less than another $1/(2m)$: the total error per
score is $0\le\bar e_k-\bar e_k^{\,\mathrm{q}}<\Delta_{\mathrm{tot}}:=1/m$.
Since $\eta_{s,b}\le m$, each transmitted number needs $O(\log m)$ bits, as
the accounting above assumes. Because every quantized score is at most its
ideal value, rounding can only decrease the expected total score of the
nulls, and e-BH on the
quantized scores at level $\alpha$ keeps the guarantee of
Theorem~\ref{thm:ecfga} verbatim: $\FDR\le\alpha$ for the oracle version
and $\FDR\le\alpha+\eta$ for the inflated version. Upward rounding of
either quantity would inflate the null mass and break validity.

The cost is that e-BH may reject fewer hypotheses, which can happen only
when the e-BH condition at the \emph{original final rejection count}
holds with little margin. Suppose e-BH on the ideal scores
rejects $R>0$ hypotheses and $R\,\bar e_{(R)}-m/\alpha\ge R\,\Delta_{\mathrm{tot}}$.
Every originally rejected hypothesis then keeps a quantized score at least
$\bar e_{(R)}-\Delta_{\mathrm{tot}}\ge m/(\alpha R)$, while every other score
was already below $m/(\alpha(R+1))$ and can only decrease; so the rejection
set is unchanged. Otherwise the rejection count can drop, and the losses
may involve many hypotheses at once --- including ones whose own rank
condition was never satisfied, since e-BH rejects every hypothesis ranked
above the largest qualifying rank. We have not quantified the resulting
power loss beyond this condition.

\subsection{BONuS-GA}\label{app:protocols-bonus}
Synthetic uniforms are generated and
stored locally and never leave the node, but the bag \emph{allocation} is
not local: the per-node sizes $\widetilde m^{(i)}$ are one joint
multinomial draw with probabilities $w^{(i)}/\sum_jw^{(j)}$, and the
weights (e.g.\ $m^{(i)}\widehat r_{0,\mathrm{init}}^{(i)}$) are distributed
across nodes. An \emph{allocation round} therefore comes first: each
node uploads its weight $w^{(i)}$ and the center draws the multinomial
and returns each node's size $\widetilde m^{(i)}$ ($O(N\log m)$ bits in
each direction), after which each node generates its uniforms locally.
(The real-data tail counts needed below can ride on the same uplink.)
Communication then proceeds in
two phases. As in \citet{Yang2021BONuS}, only quantities computed from
the pooled counts ($\F_{\rm mask}$-measurable quantities) may be used to
build the nested family $\{\widehat\Gamma_M\}$; the real and synthetic
counts are used separately only in the stopping step, which checks the
FDP estimate.

\emph{Initialization (one-shot).} Each node transmits its two tail counts
$N^{(i)}(A)$ and $\widetilde N^{(i)}(A)$ on $A=(\lambda,1]$ to the center
($O(\log m)$ bits per node). The center forms the global tail ratio
\begin{equation}\label{eq:rho-broadcast}
\rho\;:=\;\frac{N(A)+1}{\widetilde N(A)}
\end{equation}
and keeps its integer numerator and denominator \emph{exactly}; they are
also what is sent if $\rho$ must be broadcast, at $O(\log m)$ bits each.
(Should $\rho$ ever be transmitted in quantized form, it must be rounded
\emph{upward}: rounding it down would understate the FDP estimate.)
If $\widetilde N(A)=0$ the procedure makes no rejections. The ratio $\rho$
does \emph{not} enter the ranking; it is used only inside the stopping
rule below.

\emph{Per round.} The greedy admits one interval per round, ranked by the
raw pooled count $N^{(i)}(\widehat I_j)+\widetilde N^{(i)}(\widehat I_j)$,
which is $\F_{\rm mask}$-measurable and so guarantees that
$\{\widehat\Gamma_M\}$ is a valid interactive shrinking sequence in the
sense of \citet{Yang2021BONuS}. Each node maintains its own pooled-count
list and transmits the top remaining value when polled --- one integer in
$[0,m^{(i)}+\widetilde m^{(i)}]\subseteq[0,(1+c)m]$ (the realized local bag
size $\widetilde m^{(i)}$ is random under the multinomial allocation and
need not be at most $c\,m^{(i)}$, but never exceeds the global bag
$\widetilde m=cm$), costing at most $\log_2 m+\lceil\log_2(1+c)\rceil+O(1)$ bits. The host of
the admitted interval then sends the split pair
$\bigl(N^{(i)}(\widehat I_{j_\star}),\widetilde N^{(i)}(\widehat
I_{j_\star})\bigr)$ for the FDP update. Since the center already holds the
pooled count, it suffices for the host to send only the real count
$N^{(i)}(\widehat I_{j_\star})$; the center recovers the synthetic count
by subtraction. (Rescaling a count at the center, e.g.\ storing
$\lceil\rho\,\widetilde N^{(i)}(\widehat I_{j_\star})\rceil$, affects only
storage, not the cost of an upload that has already occurred.) At accepted candidates the stopping rule
gives $\rho(\widetilde N(\widehat\Gamma_M)+1)\le\alpha R(M)$; at rejected
candidates $\rho\widetilde N(\widehat\Gamma_M)$ can exceed $R(M)$, but it
is polynomial in $m$ (at most $(m+1)(\tilde m+1)$), so every running
quantity still fits in $O(\log m)$ bits. The stopping check $\widehat\FDP(M)\le\alpha$ is performed
in exact integer arithmetic,
$(N(A)+1)\bigl(\widetilde N(\widehat\Gamma_M)+1\bigr)\le\alpha\,\widetilde N(A)\,R(M)$;
any rounded-up rescaling of stored counts can only make it more conservative.
Step~5 is local at the center. Total cost is $O(\log m)+\lceil\log_2(1+c)\rceil$ bits
per round and $O(\sqrt m\,(\log m+\log c))$ bits over the $O(\sqrt m)$
rounds at $\varepsilon\asymp m^{-1/2}$. As with CFGA, the interactive
polling is optional: because the pooled ranking is a \emph{static}
function of the per-interval pooled counts, each node can equally upload
its full batch of tuples
$(j,\,N^{(i)}(\widehat I_j)+\widetilde N^{(i)}(\widehat I_j),\,
N^{(i)}(\widehat I_j),\,\widetilde N^{(i)}(\widehat I_j))$ in one
synchronous round, after which ranking, FDP scan, and step-up cut are
all local at the center and the downlink returns accepted interval
indices --- the same $O(\sqrt m\log m)$ total (plus the $O(N\log m)$-bit
allocation round) in two uplink--downlink rounds.

\paragraph{Cost of enlarging the bag.} The scale $c$ in Step~2 multiplies
the local bag size by a constant: $\widetilde m=c\,m$. The pooled-count
integer transmitted in each ranking step now ranges in
$[0,m^{(i)}+\widetilde m^{(i)}]\subseteq[0,(1+c)m]$,
so each round pays at most $\lceil\log_2(1+c)\rceil$ extra bits over the
baseline --- three for $c=4$ (e.g.\ counts in $[0,1000]$ need $10$ bits,
counts in $[0,5000]$ need $13$), lost in the $O(\log m)$ constant. The real
count sent for the FDP update is at most $m^{(i)}$, so it needs no extra
bits. Synthetic
p-values are generated and stored locally, so for any bounded $c$ the total
communication stays at $O((N+\sqrt m)\log m)$ bits including the allocation
round, i.e.\ the original $O(\sqrt m\log m)$ whenever $N=O(\sqrt m)$. Local computation
and memory at node $i$ scale with its realized bag size
$\widetilde m^{(i)}$ --- one uniform draw and one interval-membership lookup
per synthetic point --- whose expectation is $c\,m^{(i)}$ under the
proportional allocation; so $c$ is a knob
the operator can turn locally to reduce the global-tail-ratio noise
diagnosed in Section~\ref{sec:sim-hubsize}, at $c\times$ local compute and
at most $\lceil\log_2(1+c)\rceil$ extra bits per transmitted ranking integer.

\section{Full simulation overlays}\label{app:full-figs}

Figures~\ref{fig:full-quiet}--\ref{fig:full-rho} repeat the six
experiments of Section~\ref{sec:sim} with all twelve procedures overlaid
(without error bars);
the main-text figures show only the procedures needed for each
experiment's findings.

\begin{figure}[t]
\centering
\includegraphics[width=\textwidth]{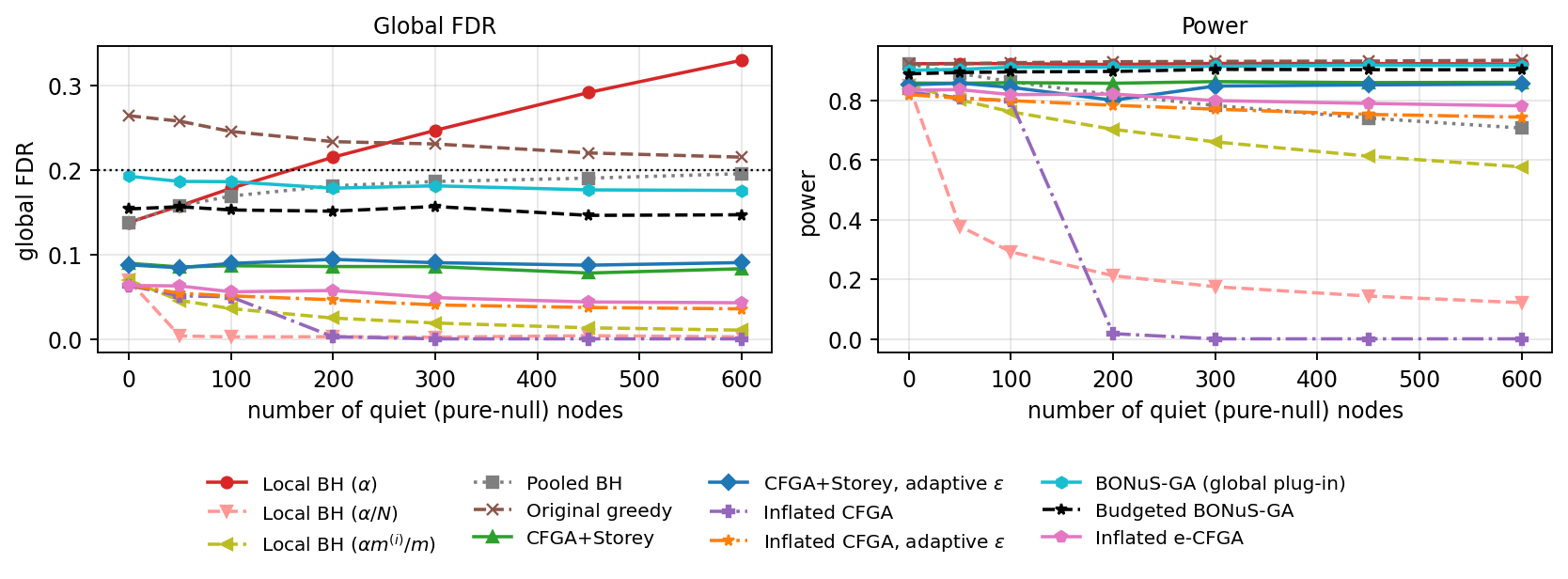}
\caption{All twelve procedures, Experiment 1 (quiet nodes).}
\label{fig:full-quiet}
\end{figure}

\begin{figure}[t]
\centering
\includegraphics[width=\textwidth]{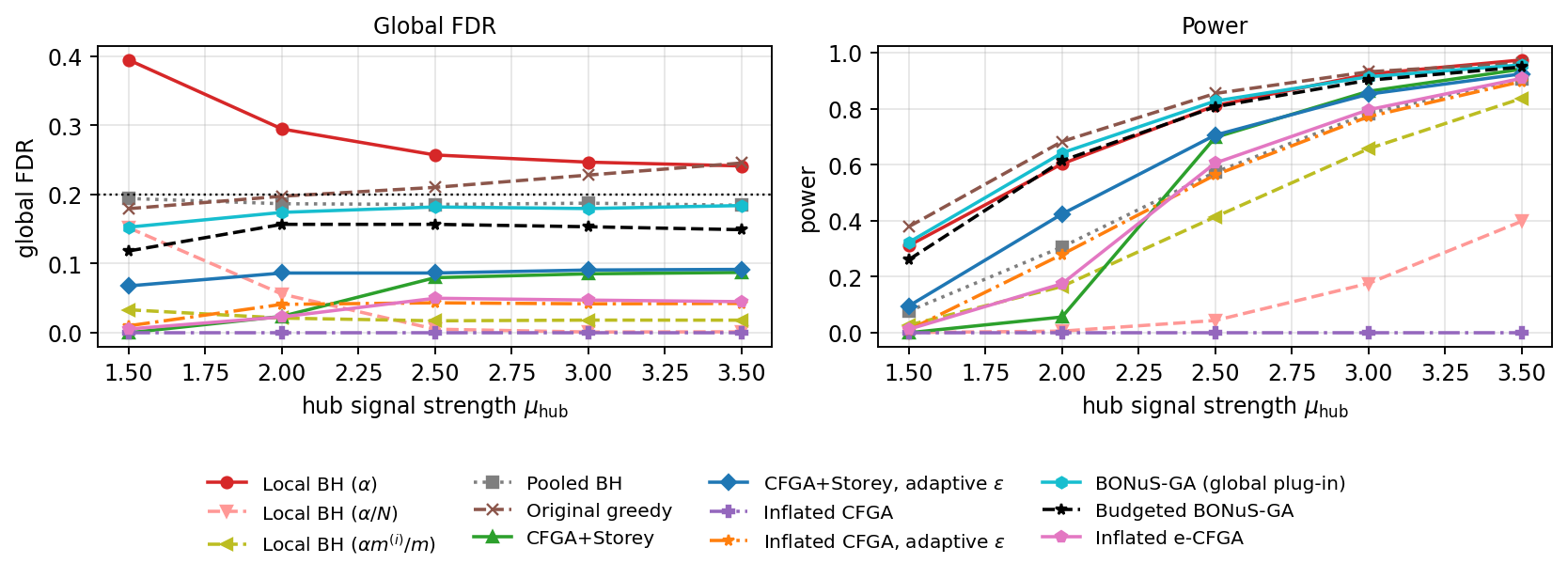}
\caption{All twelve procedures, Experiment 2 (signal strength).}
\label{fig:full-mu}
\end{figure}

\begin{figure}[t]
\centering
\includegraphics[width=\textwidth]{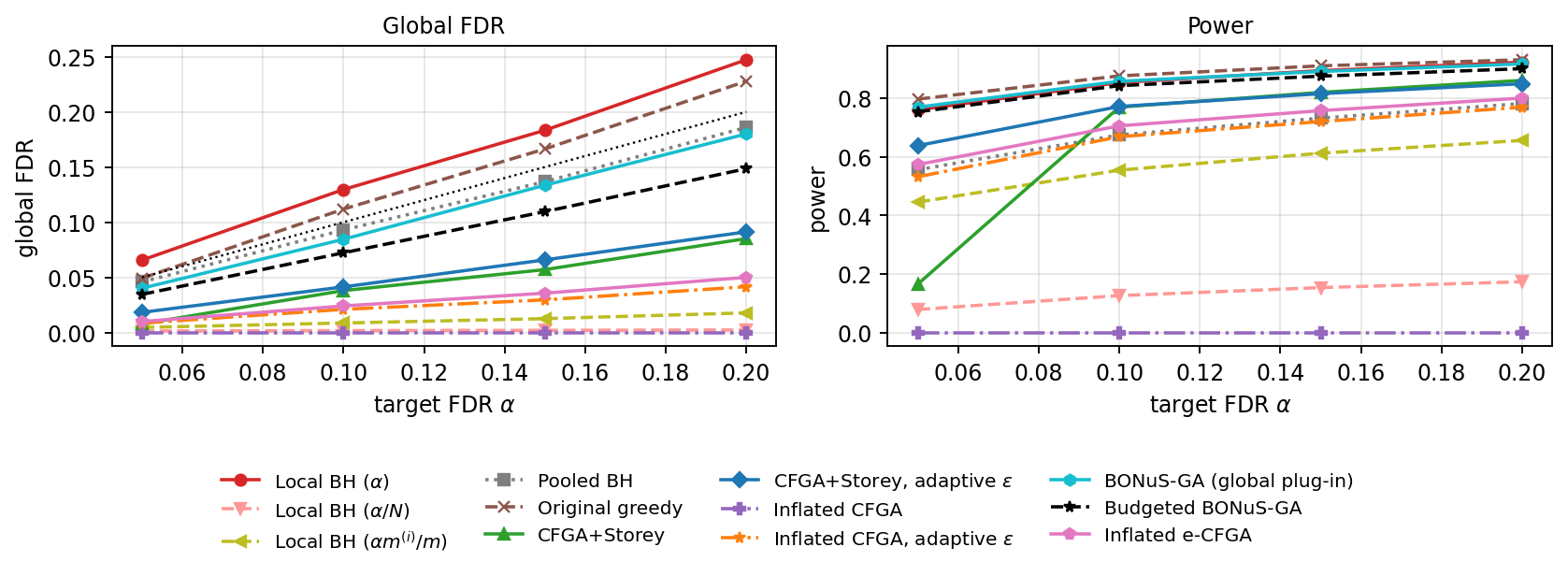}
\caption{All twelve procedures, Experiment 3 (target level; the dotted diagonal is $y=\alpha$).}
\label{fig:full-alpha}
\end{figure}

\begin{figure}[t]
\centering
\includegraphics[width=\textwidth]{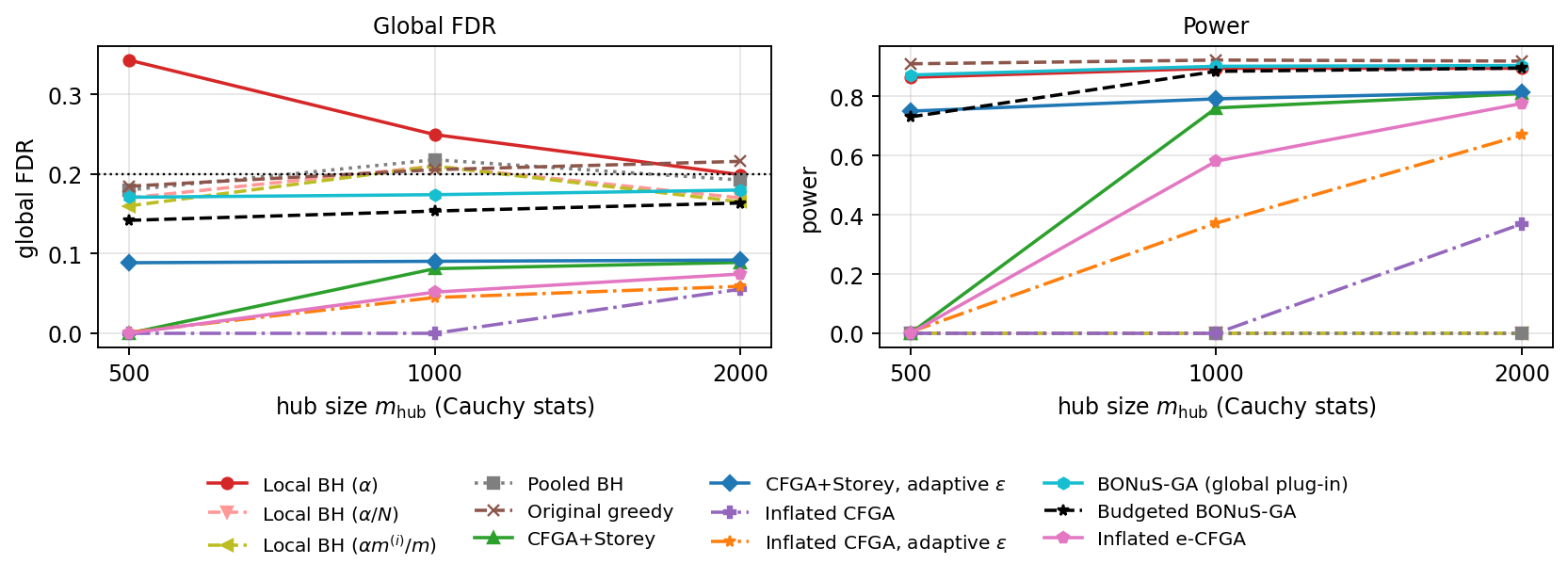}
\caption{All twelve procedures, Experiment 4 (heavy tails).}
\label{fig:full-cauchy}
\end{figure}

\begin{figure}[t]
\centering
\includegraphics[width=\textwidth]{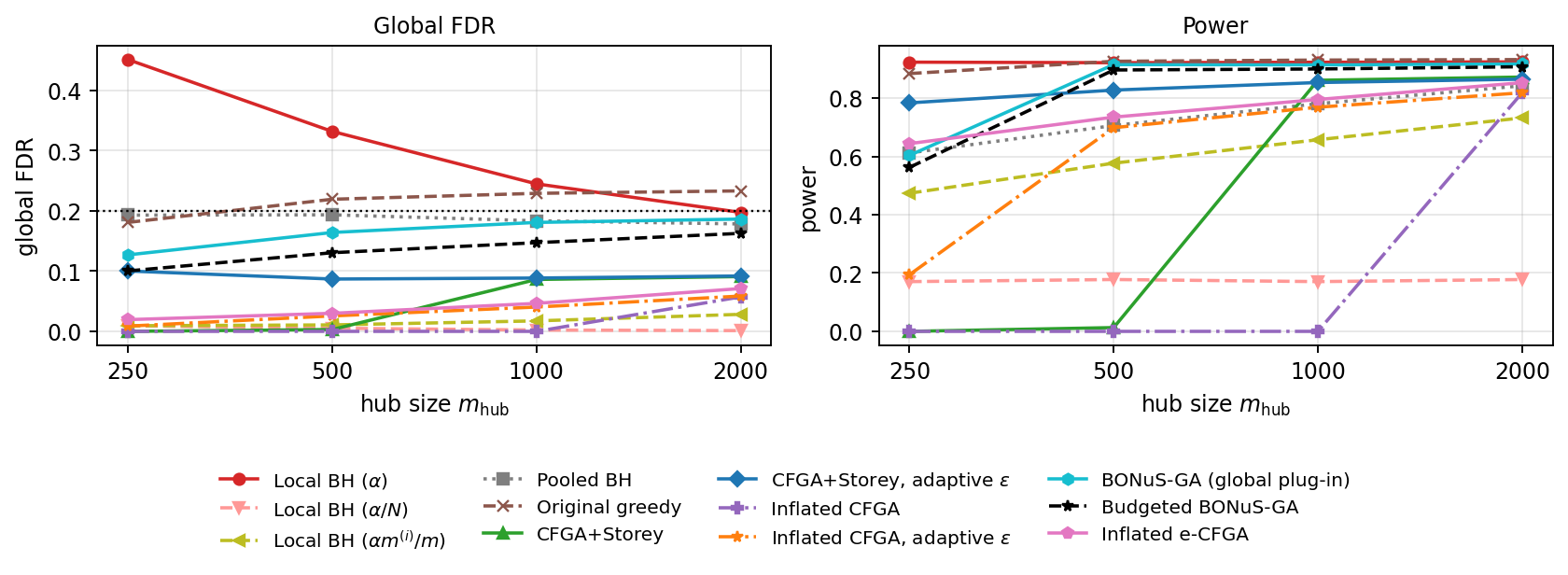}
\caption{All twelve procedures, Experiment 5 (hub size).}
\label{fig:full-n}
\end{figure}

\begin{figure}[t]
\centering
\includegraphics[width=\textwidth]{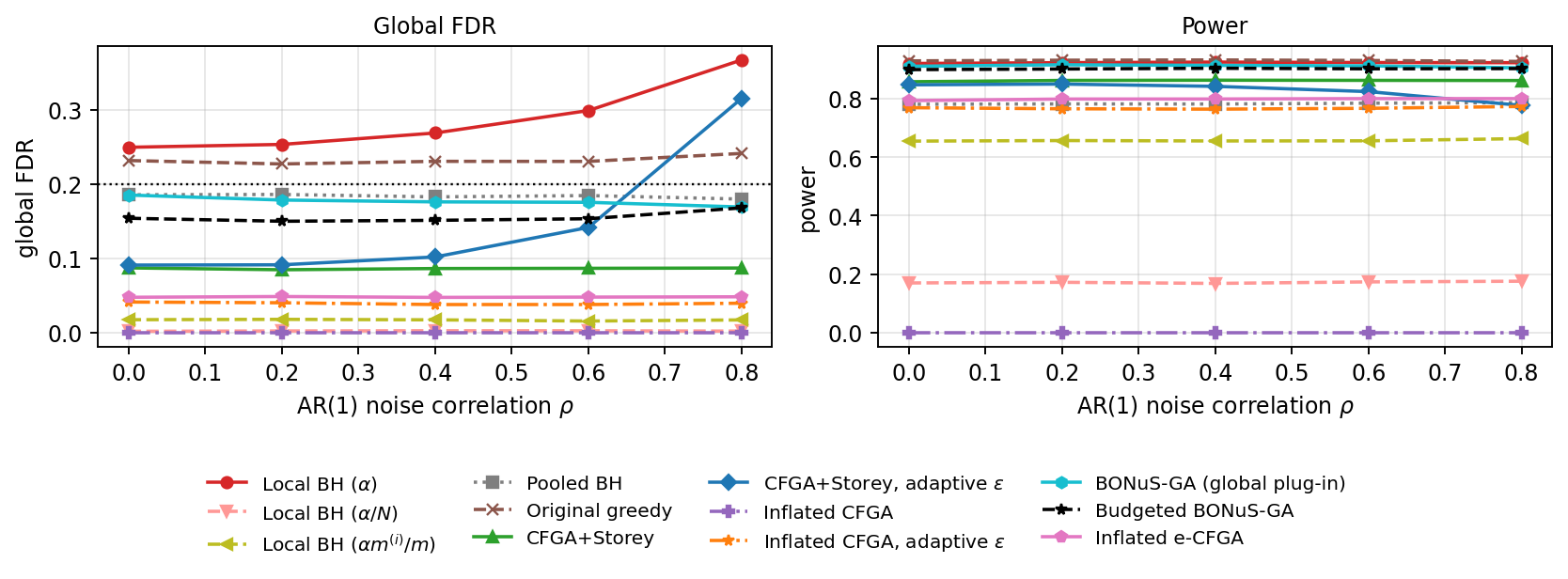}
\caption{All twelve procedures, Experiment 6 (AR(1) noise correlation).}
\label{fig:full-rho}
\end{figure}

\clearpage

\section*{Acknowledgments and disclosures}
During the preparation of this work the author used Claude (Anthropic;
models Claude Opus 4.7 and Claude Fable 5) as a research and writing
assistant: to discuss and check the mathematical arguments, to assist in
writing and reviewing the simulation code and in running the numerical
experiments, and to draft, edit and proofread the text. All proofs,
algorithms, experimental designs and conclusions were reviewed and verified
by the author, who takes full responsibility for the content of the paper.
The author thanks Javad Alipanah for providing access to the AI tools used
in this work. This work received no third-party funding. The author is
employed by Mofid Securities, which had no role in the research; the author
declares no competing interests. The simulation code is linked in
Section~\ref{sec:sim}.

\bibliographystyle{plainnat}
\bibliography{references}

\begin{thebibliography}{23}
\providecommand{\natexlab}[1]{#1}
\providecommand{\url}[1]{\texttt{#1}}
\expandafter\ifx\csname urlstyle\endcsname\relax
  \providecommand{\doi}[1]{doi: #1}\else
  \providecommand{\doi}{doi: \begingroup \urlstyle{rm}\Url}\fi

\bibitem[Barber and Cand{\`e}s(2015)]{BarberCandes2015}
Rina~Foygel Barber and Emmanuel~J. Cand{\`e}s.
\newblock Controlling the false discovery rate via knockoffs.
\newblock \emph{Annals of Statistics}, 43\penalty0 (5):\penalty0 2055--2085,
  2015.

\bibitem[Bates et~al.(2023)Bates, Cand{\`e}s, Lei, Romano, and
  Sesia]{Bates2023}
Stephen Bates, Emmanuel Cand{\`e}s, Lihua Lei, Yaniv Romano, and Matteo Sesia.
\newblock Testing for outliers with conformal $p$-values.
\newblock \emph{Annals of Statistics}, 51\penalty0 (1):\penalty0 149--178,
  2023.

\bibitem[Benjamini and Hochberg(1995)]{BH1995}
Yoav Benjamini and Yosef Hochberg.
\newblock Controlling the false discovery rate: a practical and powerful
  approach to multiple testing.
\newblock \emph{Journal of the Royal Statistical Society: Series B},
  57\penalty0 (1):\penalty0 289--300, 1995.

\bibitem[Benjamini and Yekutieli(2001)]{BY2001}
Yoav Benjamini and Daniel Yekutieli.
\newblock The control of the false discovery rate in multiple testing under
  dependency.
\newblock \emph{Annals of Statistics}, 29\penalty0 (4):\penalty0 1165--1188,
  2001.

\bibitem[Boucheron et~al.(2013)Boucheron, Lugosi, and
  Massart]{BoucheronLugosiMassart2013}
St{\'e}phane Boucheron, G{\'a}bor Lugosi, and Pascal Massart.
\newblock \emph{Concentration Inequalities: A Nonasymptotic Theory of
  Independence}.
\newblock Oxford University Press, 2013.

\bibitem[Chernozhukov et~al.(2018)Chernozhukov, Chetverikov, Demirer, Duflo,
  Hansen, Newey, and Robins]{Chernozhukov2018DML}
Victor Chernozhukov, Denis Chetverikov, Mert Demirer, Esther Duflo, Christian
  Hansen, Whitney Newey, and James Robins.
\newblock Double/debiased machine learning for treatment and structural
  parameters.
\newblock \emph{The Econometrics Journal}, 21\penalty0 (1):\penalty0 C1--C68,
  2018.

\bibitem[Dubhashi and Ranjan(1998)]{DubhashiRanjan1998}
Devdatt Dubhashi and Desh Ranjan.
\newblock Balls and bins: a study in negative dependence.
\newblock \emph{Random Structures \& Algorithms}, 13\penalty0 (2):\penalty0
  99--124, 1998.

\bibitem[G{\"o}lz et~al.(2022)G{\"o}lz, Zoubir, and Koivunen]{Golz2022}
Martin G{\"o}lz, Abdelhak~M. Zoubir, and Visa Koivunen.
\newblock Multiple hypothesis testing framework for spatial signals.
\newblock \emph{IEEE Transactions on Signal and Information Processing over
  Networks}, 8:\penalty0 771--787, 2022.

\bibitem[Heesen and Janssen(2016)]{HeesenJanssen2016}
Philipp Heesen and Arnold Janssen.
\newblock Dynamic adaptive multiple tests with finite sample {FDR} control.
\newblock \emph{Journal of Statistical Planning and Inference}, 168:\penalty0
  38--51, 2016.

\bibitem[Jin and Cand{\`e}s(2023)]{JinCandes2022}
Ying Jin and Emmanuel~J. Cand{\`e}s.
\newblock Selection by prediction with conformal p-values.
\newblock \emph{Journal of Machine Learning Research}, 24\penalty0
  (244):\penalty0 1--41, 2023.

\bibitem[Joag-Dev and Proschan(1983)]{JoagDevProschan1983}
Kumar Joag-Dev and Frank Proschan.
\newblock Negative association of random variables, with applications.
\newblock \emph{Annals of Statistics}, 11\penalty0 (1):\penalty0 286--295,
  1983.

\bibitem[Pournaderi and Xiang(2023{\natexlab{a}})]{PournaderiXiang2022Sample}
Mehrdad Pournaderi and Yu~Xiang.
\newblock Sample-and-forward: communication-efficient control of the false
  discovery rate in networks.
\newblock In \emph{IEEE International Symposium on Information Theory (ISIT)},
  pages 1949--1954, 2023{\natexlab{a}}.

\bibitem[Pournaderi and Xiang(2023{\natexlab{b}})]{PournaderiXiang2024}
Mehrdad Pournaderi and Yu~Xiang.
\newblock On large-scale multiple testing over networks: an asymptotic
  approach.
\newblock \emph{IEEE Transactions on Signal and Information Processing over
  Networks}, 9:\penalty0 442--457, 2023{\natexlab{b}}.
\newblock (arXiv:2211.16059v4).

\bibitem[Pournaderi and Xiang(2023{\natexlab{c}})]{PournaderiXiangDistFree2023}
Mehrdad Pournaderi and Yu~Xiang.
\newblock Communication-efficient distribution-free inference over networks.
\newblock \emph{arXiv preprint arXiv:2307.09850}, 2023{\natexlab{c}}.

\bibitem[Ramdas et~al.(2017)Ramdas, Chen, Wainwright, and
  Jordan]{Ramdas2017QuTE}
Aaditya Ramdas, Jianbo Chen, Martin~J. Wainwright, and Michael~I. Jordan.
\newblock {QuTE}: decentralized multiple testing on sensor networks with false
  discovery rate control.
\newblock In \emph{IEEE Conference on Decision and Control (CDC)}, 2017.

\bibitem[Ren and Barber(2024)]{RenBarber2024evalues}
Zhimei Ren and Rina~Foygel Barber.
\newblock Derandomised knockoffs: leveraging $e$-values for false discovery
  rate control.
\newblock \emph{Journal of the Royal Statistical Society: Series B},
  86\penalty0 (1):\penalty0 122--154, 2024.

\bibitem[Ren et~al.(2023)Ren, Wei, and Cand{\`e}s]{RenWeiCandes2023}
Zhimei Ren, Yuting Wei, and Emmanuel~J. Cand{\`e}s.
\newblock Derandomizing knockoffs.
\newblock \emph{Journal of the American Statistical Association}, 118\penalty0
  (542):\penalty0 948--958, 2023.

\bibitem[Slud(1977)]{Slud1977}
Eric~V. Slud.
\newblock Distribution inequalities for the binomial law.
\newblock \emph{Annals of Probability}, 5\penalty0 (3):\penalty0 404--412,
  1977.

\bibitem[Storey et~al.(2004)Storey, Taylor, and
  Siegmund]{StoreyTaylorSiegmund2004}
John~D. Storey, Jonathan~E. Taylor, and David Siegmund.
\newblock Strong control, conservative point estimation and simultaneous
  conservative consistency of false discovery rates: a unified approach.
\newblock \emph{Journal of the Royal Statistical Society: Series B},
  66\penalty0 (1):\penalty0 187--205, 2004.

\bibitem[Wang and Ramdas(2022)]{WangRamdas2022}
Ruodu Wang and Aaditya Ramdas.
\newblock False discovery rate control with e-values.
\newblock \emph{Journal of the Royal Statistical Society: Series B},
  84\penalty0 (3):\penalty0 822--852, 2022.

\bibitem[Xiang(2019)]{XiangQuantized2019}
Yu~Xiang.
\newblock Distributed false discovery rate control with quantization.
\newblock In \emph{IEEE International Symposium on Information Theory (ISIT)},
  pages 246--249, 2019.

\bibitem[Yang et~al.(2021)Yang, Lei, Ho, and Fithian]{Yang2021BONuS}
Chiao-Yu Yang, Lihua Lei, Nhat Ho, and William Fithian.
\newblock {BONuS}: multiple multivariate testing with a data-adaptive test
  statistic.
\newblock \emph{arXiv preprint arXiv:2106.15743}, 2021.

\bibitem[Zhang et~al.(2025)Zhang, Pournaderi, Xiang, and
  Varshney]{ZhangPournaderiXiangVarshney2025}
Daofu Zhang, Mehrdad Pournaderi, Yu~Xiang, and Pramod Varshney.
\newblock Distributed multiple testing with false discovery rate control in the
  presence of {B}yzantines.
\newblock \emph{arXiv preprint arXiv:2501.13242}, 2025.

\end{thebibliography}

\end{document}